\documentclass[12pt, reqno]{amsart}
\usepackage[thm]{macros}
\usepackage{makecell}
\usepackage{afterpage}
\usepackage{placeins}
\usepackage{tabularx}

\usepackage{times}
\usepackage{inconsolata}

\usepackage{appendix}
\usepackage{microtype}

\usepackage[]{mdframed}
\usepackage{epigraph}

\usepackage{titletoc}
\newcommand\DoToC{%
 \vskip1cm
 \startcontents[sections]
  \printcontents[sections]{l}{1}{\textbf{Contents}\vskip3pt\hrule\vskip5pt\setcounter{tocdepth}{2}}
  \vskip5pt\hrule\vskip5pt
}

\usepackage[all=normal, paragraphs=tight, wordspacing=tight, paragraphs=tight,
bibbreaks=tight,
floats=tight, bibnotes=tight]
{savetrees}
\definecolor{CORAL}{HTML}{F26D21}
\definecolor{ALICE}{HTML}{107895}
\definecolor{ACCENT}{HTML}{268BD2}
\definecolor{RUBY}{HTML}{9A2515}

\newcolumntype{L}[1]{>{\raggedright\arraybackslash}p{#1}}

\newcommand{\bw}{{w}}
\newcommand{\bz}{{z}}
\newcommand{\bW}{{W}}
\newcommand{\btau}{\tau}
\newcommand{\brho}{\rho}
\newcommand{\cW}{\mathcal{W}}
\newcommand{\bx}{{x}}
\newcommand{\by}{{y}}
\newcommand{\bomega}{{\omega}}

\newcommand{\cT}{\mathcal{T}}
\newcommand{\cD}{\mathcal{D}}
\setlist{leftmargin=*}
\newgeometry{margin=1.25in}

\newcommand{\bpi}{\bm{\pi}}
\newcommand{\bY}{{Y}}

\numberwithin{equation}{section}

\title{Potential weights and implicit causal designs in linear regression}

\author{Jiafeng Chen \\ Department of Economics, Stanford University \\
\href{mailto:jiafeng@stanford.edu}{jiafeng@stanford.edu}}
\begin{document}
\date{\today. I thank Marcella Alsan, Isaiah Andrews, Dmitry Arkhangelsky, Tim Armstrong,
 Kirill Borusyak, Federico Bugni, Ambarish Chattopadhyay, Avi Feller, Matthew Gentzkow,
 Paul Goldsmith-Pinkham, Han Hong, Peter Hull, Guido Imbens, Patrick Kline, Lihua Lei,
 Jimmy Lin, Yiqi Liu, Jos\'e Luis Montiel Olea, Francesca Molinari, Matt Notowidigdo, Chen
 Qiu, Ashesh Rambachan, David Ritzwoller, Brad Ross, Jonathan Roth, Andres Santos, Yucheng
 Shang, Jesse Shapiro, Jann Spiess, Tymon S\l {}ocy\' {n}ski, Elie Tamer, Alex
 Torgovitsky, Winnie van Dijk, Davide Viviano, Anqi Zhao, Jos\'e Zubizarreta, and
 participants in seminars and workshops at Berkeley, the Chinese University of Hong Kong,
 Cornell, Harvard, Stanford, UC Irvine, SEA 2024, and CES 2025 for helpful comments. I
 thank Syed Ashal Ali for research assistance. AI tools (GPT, Claude, and Refine.ink) are
 used in writing, validating results, and implementing empirical analyses; I take
 responsibility for all errors. }

\maketitle

\begin{abstract}

Applied researchers routinely use linear regression to estimate causal effects, justified
by \emph{quasi-experimental treatment variation}, while leaving assumptions on treatment
assignment implicit. We formalize a minimal criterion for quasi-experimental
interpretation---that the regression estimates some contrast of potential outcomes under
the true assignment process, regardless of potential outcomes---and characterize its
implications for arbitrary regressions. This criterion implies linear restrictions on the
true treatment distribution, whose solutions we call \emph{implicit designs}. A regression
is exactly quasi-experimental if and only if the true design is an implicit design, and
approximately so when it is close to one, in a sense we formalize. Our framework unifies
existing results and uncovers new ones across many settings. Qualitatively, an AI-assisted
census of 1,051 recent papers finds quasi-experimental regression pervasive and often
vulnerable to our negative results. Quantitatively, we assess exact and approximate
quasi-experimental interpretation in nine studies by computing their implicit designs and
estimands.

\end{abstract}

\newpage

\section{Introduction}

Linear regression is nearly universal in empirical economics for estimating causal
effects. Its causal interpretation is often justified through \emph{quasi-experimental}
treatment assignment \citep{pgp,angrist2010credibility,currie2020technology}.  
Applied work often proceeds by informally arguing that treatment is ``as good as
randomized,'' choosing a regression specification, and  interpreting some of its
coefficients as causal. Awkwardly, a regression (e.g., $Y_i = W_i\tau + x_i'\gamma +
\epsilon_i$ for treatment $W$ and covariates $x$) looks like a structural model for the
\emph{outcomes}, but its causal
interpretation comes instead from assumptions on the \emph{treatment assignment}
process. These assumptions are often left implicit.

It is not obvious whether such assumptions exist or are plausible. As questioned by \citet
{freedman2008regression}, ``randomization does not justify the assumptions behind the OLS
model.'' Nevertheless, practitioners appear optimistic that such \emph {quasi-experimental
interpretation} is typically available and overwhelmingly prefer regression to modern
causal inference methods \citep{wager2024causal}.\footnote{We find little
takeup of
standard causal inference methods (inverse propensity weighting, matching, double/debiased
machine learning) in an extensive AI-assisted literature review in \cref{sec:empirical}.}
As \citet{angrist2008mostly} put it in \emph{Mostly Harmless Econometrics},
\begin{quote}
    Most econometrics texts appear to take econometric models very
seriously [\ldots Instead,] a principle [here] is that estimators in
common use almost always have a simple interpretation that is not heavily
model-dependent.
\end{quote} From this perspective, a regression---despite looking like a model for
 outcomes---is actually just an estimator for \emph{some} causal parameter under
 treatment effect heterogeneity, usually valid under \emph{some} assumptions on treatment
 assignment. Is this optimism towards regression justified in theory and in practice? 
 
A large applied econometrics literature suggests that the answer is delicate. This
literature\footnote{Among others, \citet
{imbensangrist,angrist1998,lin2013agnostic,sloczynski2022interpreting,sloczynski2020should,blandhol2022tsls,aronow2016does,goldsmith2022contamination,borusyak2024negative,athey2018design,kline2011oaxaca,bugni2023decomposition,mogstad2024instrumental,arkhangelsky2023fixed,arkhangelsky2021double,chetverikov2023logit,kolesar2024dynamic,zhao2025interacted,arganaraz2024randomly,kwon2025estimating}.}
produces a rich ``zoo'' of results. Certain specifications have quasi-experimental
interpretations under some treatment assignments, but others may not (e.g., produce
negatively weighted causal effects). Seemingly small differences can be critical: With
binary $W_i$, $Y_i = W_i\tau + x_i'\gamma + \epsilon_i$ estimates a weighted average
treatment effect when the propensity score is linear in $x_i$ \citep{angrist1998}, but the
analogous specification with multi-valued $W$ produces uninterpretable estimands
\citep{goldsmith2022contamination}.

These results offer a partial evaluation of regression's quasi-experimental
interpretation, one specification at a time, falling short of assessing
\citet{angrist2008mostly}'s general optimism. For instance, what treatment assignment
assumptions justify interpreting the interacted regression $Y_i = W_i\tau_0 + W_i
x_i'\tau_1  + x_i'\gamma
+ \epsilon_i$ quasi-experimentally---or a panel regression like $Y_ {it} =
\alpha_i + \beta_t + W_{it}\tau + x_{it}'\gamma + \epsilon_{it}$? When such a regression can be interpreted quasi-experimentally, what
causal effect does it target under heterogeneity? When the implicit assumptions are misspecified---as they often are---how poor
are the resulting estimates?

This paper provides a general framework that answers these questions for an
\emph{arbitrary} regression with finite-valued treatments. For any such specification, the
framework \emph{computes} a set of candidate treatment-assignment processes
(\emph{designs}) and, for each candidate, the corresponding estimand. We call these
candidates \textbf{implicit designs} and their estimands \textbf{implicit estimands}.
These objects are simple to estimate in practice.

This framework characterizes when a regression admits an \emph{exact} quasi-experimental
interpretation---estimating a causal effect under any configuration of potential outcomes.
A regression admits one if and only if the true assignment process is an implicit design,
in which case it estimates the corresponding implicit estimand. Our framework also
characterizes when a regression admits an \emph{approximate} quasi-experimental
interpretation under a given design: We make precise a sense in which the regression's
distance from quasi-experimental interpretation equals the covariance between its
residuals and the distance of that given design from an implicit design. This covariance
can be large if outcomes were adversarially chosen, but it may be small in practice.

Our framework follows naturally by formalizing quasi-experimental interpretation as
requirements on the regression estimand in the idealized experiment that redraws
treatment according to some true assignment process $\bpi^* = (\pi_1^*, \ldots,
\pi_n^*)$---where $\pi_i^*$ denotes the true assignment probabilities of unit $i$'s
treatment $W_i$. Given any potential outcomes $y_i(\cdot)$, repeatedly drawing $W_i \sim
\pi_i^*$ defines an estimand $\tau$ for regression as a function of $y_i(\cdot)$. We argue
that a minimal requirement for calling $\tau$ ``quasi-experimental'' is that it be a
contrast under this experiment, \emph{regardless of potential outcomes}:
\begin{enumerate} [label=\text{[MQE]}]
    \item \label{item:M} 
Under $\bpi^*$, the estimand
$\tau$ is a \textbf{contrast}\footnote{If there are $J+1$ treatments
$\br{0,\ldots,J}$ and $n$ units $\br{1,\ldots, n}$, a contrast of individual potential
outcomes $y_{i}(j)$ is  a parameter $\frac{1}{n}
\sum_{i=1}^n \sum_{j=0}^{J}
    \omega_ {i}(j)
    y_i(j)$, for weights that sum to zero across $j$, $\sum_{j=0}^{J}
    \omega_ {i} (j) =
    0$.} of
    individual potential outcomes for \textbf{any} potential-outcome distribution---even a
    worst-case one.
\end{enumerate}
If a regression satisfies \cref{item:M}, practitioners can safely interpret its estimates
as some---though not necessarily {useful}---causal effects, identified in the idealized
experiment $\bpi^*$. 

\cref{item:M}---for \emph{minimally quasi-experimental}---codifies two requirements.
First, it enforces the \emph{treatment-based} logic practitioners appeal to by requiring
that causal interpretation be agnostic of potential outcomes. Second, \cref{item:M}
imposes \emph{level independence} \citep{blandhol2022tsls}: For any configuration of
potential outcomes where all individual treatment effects are zero,  the estimand should
always be zero. For binary treatments, this amounts to requiring $\tau$ to be a weighted
average treatment effect.

Because \cref{item:M} depends on the unknown true design $\bpi^*$, we cannot test it
directly. We therefore decompose it into an \emph{existence} and a \emph{correctness} question:
\begin{enumerate}[label=\text{[MQE-\arabic*]}]
    \item \label{item:M1} Does \emph{any} design $(\pi_1,\ldots,\pi_n)$, consistent
    with
    the data, satisfy \cref{item:M}?
    \item \label{item:M2} Is any design in \cref{item:M1} equal to $\pi_1^*,\ldots,
    \pi_n^*$?
\end{enumerate} 
\noindent Our framework systematically evaluates \cref{item:M} through
\cref{item:M1}, which asks whether an idealized experiment even \emph{exists} that
supports quasi-experimental interpretation. We focus on \cref{item:M1} because it can be
objectively settled and is relevant on its own. If regression estimators have
causal interpretations on treatment assignment grounds, as broadly predicted by
\citet{angrist2010credibility}, there should at least be some defensible design that
answers \cref{item:M1}. Tests of \cref {item:M1} also produce a narrow set of candidate
designs that makes important debates about \cref{item:M2}---whether a model of treatment
is correctly specified---explicit.

To characterize \cref{item:M1}, we observe that the regression specification aggregates
outcomes linearly by assigning a weight $\rho_i(w)$ to unit $i$'s \emph{potential outcome}
under treatment $w$, $y_i(w)$. Correspondingly, we call $\rho_i(w)$ \textbf{potential
weights}. These potential weights are known in population and consistently estimable in
sample.  If a design $\pi = (\pi_1,\ldots,\pi_n)$ satisfies \cref{item:M1}, then for every
unit $i$, the weights assigned to unit $i$ must be \emph{ex ante} mean zero under $\pi$:
\[\sum_w \pi_i(w) \rho_i (w) = 0 \quad \sum_w
\pi_i(w) = 1\numberthis
\label{eq:li_intro}.\] Otherwise, shifting every potential outcome of unit $i$ by a
 constant would change the regression estimand under $\bpi$ without changing the
 treatment effects. 

 Since the potential weights are known in population, \eqref {eq:li_intro} defines linear
 equations in $\pi_i$.  In many leading specifications, restrictions
\eqref{eq:li_intro} are sharp---often yielding either no solution or a unique one. 
These solutions to \eqref{eq:li_intro} are the set of \textbf{implicit
designs}---candidates in \cref{item:M1}.  For each implicit design, its corresponding
\textbf{implicit estimand} puts weight $\omega_i(w) \equiv \pi_i(w) \rho_i(w)$ on the
 unit-$i$ potential outcome $y_i (w)$. These weights $\omega_i(w)$ can likewise be
 explicitly computed. 

This paper makes the following contributions. Our main result formalizes how implicit
designs and implicit estimands characterize (i) which idealized experiments are
consistent with quasi-experimental interpretation of a regression coefficient and
(ii) which causal contrast a regression targets. Namely, a regression satisfies \cref
{item:M} if and only if implicit designs exist (\cref{item:M1}) and the true assignment
process is equal to an implicit design (\cref{item:M2}). When it satisfies \cref
{item:M}, the estimand is the implicit estimand under $\bpi^*$.

Second,  when an exact quasi-experimental interpretation is rejected---either
\cref{item:M1} or \cref{item:M2} fails---we formalize a novel sense in which regressions
are close to doubly robust estimators for \emph {some} causal contrast under some
explicitly given design $\tilde \bpi$. Given a target causal contrast and a design $\tilde
\bpi$, an attractive estimator for that contrast is an augmented inverse propensity
weighting (AIPW) estimator that uses $\tilde
\bpi$ for the propensity weighting and the regression specification as the outcome model
\citep{bang2005doubly}. Such an estimator is doubly robust in the sense that it estimates
the target treatment effect if either $\tilde \bpi$ or the regression is correctly
specified.

We show that the difference between the regression estimator and this AIPW estimator, for
\emph{some} estimand, is numerically equal to the empirical covariance between the
regression residuals and the
degree to which $\tilde\bpi$ violates
\eqref{eq:li_intro}. Thus, if $\tilde \bpi$ is itself an implicit design, then the
 regression estimator is numerically equivalent to AIPW (moreover, for the implicit
 estimand). Otherwise, the regression estimator is still approximately doubly robust, to
 the extent that violations of \eqref{eq:li_intro} covary little with misspecification of
 the outcome model. In that case, the optimism of \citet{angrist2008mostly} is  justified
 in an approximate sense, despite adversarial choices of the potential outcomes that make
 the regression ill-behaved.

Third, our results on exact quasi-experimental interpretation unify a wide swath of
results in the applied econometrics literature and produce converses for some existing
results. \Cref {item:M} captures the shared logic across starkly different settings \citep
{angrist1998,blandhol2022tsls,goldsmith2022contamination,kline2011oaxaca,athey2018design}.
Computing implicit designs and implicit estimands for these regression specifications
exactly recovers the designs and estimands posited in these papers. Since implicit designs
are necessary implications of \cref{item:M}, doing so delivers \emph {converses} to these
results---namely, that quasi-experimental interpretation is only possible under exactly
those designs.

Fourth, we uncover new results for specifications that interact treatment with covariates
\citep
 {lin2013agnostic,miratrix2013adjusting,imbens2009recent,kline2011oaxaca,zhao2025interacted}
 and for two-way fixed effect (TWFE) regressions. In both cases, exact quasi-experimental
 interpretation can be brittle, in the sense that implicit designs need not exist outside
 special cases. Taken together, these results suggest that exact quasi-experimental
 interpretation of regression is perhaps less generic than predicted by \citet
 {angrist2008mostly}.

Finally, we extend the framework to two-stage least-squares (TSLS). There, the framework
characterizes requirements on the instrument assignment process for interpreting TSLS
coefficients as \emph {instrument}-on-outcome contrasts.
The implicit estimand here additionally pins down restrictions for treatment compliance
patterns for interpreting TSLS estimands as reasonable \emph
{treatment}-on-outcome effects. Our framework similarly unifies and extends the TSLS
literature\footnote{We share a focus on unified analysis with related papers by
\citet{navjeevan2023identification} and \citet{goff2024does}. Compared to these papers,
 our starting point is the interpretation of a particular TSLS estimator.} by recovering converses
to results in \citet
{blandhol2022tsls,imbensangrist,behaghel2013robustness,sloczynski2020should,bhuller20242sls}---they
even help clarify a small gap in recent work on TSLS with multiple treatments.

We illustrate the practical relevance of this framework by documenting stylized facts
about empirical strategies in 1,051 recent papers and by assessing exact and approximate
quasi-experimental interpretation in nine recent papers. We find that our negative results
on interactions and TWFE potentially reject quasi-experimental interpretation in many
specifications. A detailed re-analysis of nine recently published papers  shows that exact
quasi-experimental interpretation is rejected for all of them. Nevertheless, all but one
appear approximately quasi-experimental under reasonable designs that one may
alternatively consider. In terms of estimands, many of these applications have poor
overlap, and the regression estimand predominantly aggregates over units with good
overlap---sometimes a very small fraction of all units.

Our results thus generically and automatically assess the optimism of
\citet{angrist2010credibility} for any given regression specification. Applying them
widely, both theoretically and in our reanalysis, we find that exact quasi-experimental
interpretation rarely holds. Thus, contrary to a narrow reading of
\citet{angrist2008mostly}, the implicit model for treatment assignment imposed by
regression is often untenable. Nevertheless, in our re-analysis, regression is often
approximately quasi-experimental relative to some reasonable alternative
designs---partially vindicating the optimism for regression as an empirical regularity.

Ultimately, many of the practical implications of this paper are conceptual. Using a
regression for estimating causal effects, from a quasi-experimental perspective, imposes
nontrivial modeling assumptions on the treatment assignment process and automatically
selects particular estimands. Transparent empirical analyses should make these choices
explicit, for instance by estimating and reporting implicit designs and implicit
estimands. Our empirical analyses provide a template for practitioners.

This paper proceeds as follows. \Cref{sec:potential_weights_and_implicit_designs} contains
our main results. \Cref{sec:theoretical_applications} applies our framework to a litany of
regression specifications, yielding new theoretical results.
\Cref{sec:extension_two_stage_least_squares} extends the framework to TSLS. \Cref
{sec:empirical} contains our empirical results. \Cref {sec:conclusion} concludes.

\section{Potential weights and implicit designs}
\label{sec:potential_weights_and_implicit_designs}

Consider a finite population of units $i \in [n] \equiv \br{1,\ldots, n}$. Each unit
receives one treatment $\bw$ from a finite set $\cW$ of size $J+1$. Each unit has
covariates $\bx_i$ and vector-valued potential outcomes of length $T$, $\br{\by_i (\bw)
\in
\R^T : \bw \in \cW}$.\footnote{For expositional clarity, we assume that the dimension of
 the outcome vector is the same across individuals (i.e., balanced panels).  \Cref
 {asub:imbalance} discusses imbalanced panels.}  We denote by $\bW_i$ the realized
 treatment. After assignment, we observe a corresponding realized outcome $\bY_i = \by_i
 (\bW_i)$.

To emphasize that identification comes from variation in treatment assignment, we isolate
this variation by treating $(\bx_i, \by_i(\cdot))$ as fixed numbers and only considering
the randomness in $\bW_i$. This design-based perspective aligns with how
quasi-experimentalists argue identification and with how we compute implicit designs. It
also allows for treatment to be correlated across units. \Cref{rmk:superpop} shows that
our key results continue to apply under i.i.d. sampling.

Let $\bpi^*$ denote the marginal treatment assignment probabilities (i.e., propensity
scores):\[\bpi^* = (\pi^*_1, \ldots,
\pi^*_n) \text{ where } \pi_i^*(\bw) = \P(\bW_i = \bw). \] We call $\bpi =(\pi_1
 (\cdot),\ldots,\pi_n(\cdot))$ a \emph{design} and call the actual probabilities $\bpi^*$
 the true design. Write a regression generically as \[\bY_{it} = z_t(\bx_i, \bW_i)'\beta +
\epsilon_{it}
\numberthis \label{eq:general_regression}
\] with known $z_t(\cdot, \cdot)$. For a known matrix $\Lambda \in \R^{k
\times K}$, we would like to interpret certain coefficient contrasts $\btau \equiv
\Lambda\beta \in \R^k$ as causal effects. 
To emphasize, \eqref{eq:general_regression} does not specify a structural model; it simply
specifies an estimand $\btau$ given $(\bpi^*, (\by_i(\cdot),
\bx_i)_ {i=1}^n)$. Since it is common in practice to specify a regression first
 and interpret its estimated coefficients as causal effects, our analysis starts with a
 regression and investigates which $\bpi^*$ are compatible with
 interpreting the regression under \cref{item:M}.

This setup is general: It encompasses cross-sectional ($T=1$), panel ($T > 1$), $
(J+1)$-valued treatment, scalar contrast ($k = 1$), and
multiple contrasts ($k > 1$) settings.\footnote{Our results do extend to continuous
treatments, but they become much less powerful, essentially because there are only
finitely many restrictions for infinitely many objects. }
\cref{sec:extension_two_stage_least_squares} extends these results to TSLS. To fix ideas,
 we start with the binary-treatment, scalar-outcome, and scalar-contrast case $
 (T=J=k=1)$.  We then push this intuition to the general case.

\subsection{Core intuition} 
\label{sub:leading} To motivate, \citet{blakeslee2020way} study the impact
 of water loss in rural India on economic outcomes.  Water loss is measured by a
 binary $W_i$, indicating whether the first borewell household $i$ drilled has failed.
 The authors motivate identification by emphasizing that well failure ``depends on highly
 irregular, quasi-random subsurface properties'' (p. 206). The true design $\bpi^*$---the
 natural process of borewell failure---is unknown, but the authors argue that failure is
 difficult to predict, making treated and untreated households plausibly comparable, and
 correspondingly marshal detailed hydrogeological evidence.

After informally arguing that treatment is plausibly quasi-experimentally assigned,
\citet{blakeslee2020way} then estimate a simple
regression across multiple outcomes at the household level:
\[Y_i =
\tau W_i
+ x_i'\gamma + \epsilon_i, \text{ for which } \numberthis  z (x_i, W_i) = [W_i,
x_i']', \beta = [\tau, \gamma']',  \Lambda =
[1,0_{\dim(x)}'].
\label{eq:angrist98intro}
\] 
A tension here is that the regression itself does not encode the rich substantive
discussion of $\bpi^*$. Instead,
$\bpi^*$ is left implicit as whatever assignment process that  justifies interpreting
\eqref{eq:angrist98intro} quasi-experimentally. 

This pattern---arguing that a treatment is quasi-experimentally assigned and estimating a
causal effect via regression, without explicitly modeling the design nor explicitly naming
an estimand---is common in empirical work. We document it through an AI-assisted census of
the literature in \cref{sec:empirical}. When estimates are presented this way, what must a
reader believe about $\bpi^*$ for \eqref{eq:angrist98intro} to have a quasi-experimental
interpretation, and which causal contrast does $\tau$ represent when effects are
heterogeneous?

To answer these questions, let us return to a regression of a scalar $Y_i$ on some known
transform $z (x_i, W_i)$. The population regression coefficient is defined as: \[
\beta \equiv \pr{{\frac{1}{n}\sum_{i=1}^n \E_{W_i \sim \pi_i^*} \bk{z
 (x_i, W_i) z(x_i, W_i)'}}}^{-1} \pr{\frac{1}{n}\sum_{i=1}^n \E_{W_i \sim \pi_i^*} [z
 (x_i, W_i) y_i (W_i)]}.
\] 
This definition is simply the design-based analogue of the usual ``$\E[x_i x_i']^{-1}\E
[x_i y_i]$''
formula. 

Let $G_n \equiv G_n(\bpi^*) \equiv {\frac{1}{n}\sum_ {i=1}^n \E_{W_i
\sim \pi_i^*} \bk{z(x_i, W_i) z(x_i, W_i)'}}$ denote the population \emph {Gram matrix} of
this regression. Since $G_n$ is consistently estimable, we treat it as known. Empirically,
we produce feasible estimates by replacing $G_n$ with its estimate.\footnote{In
particular, since the regression \emph{estimator} replaces $G_n$ with $\hat G_n =
\frac{1}{n} \sum_{i=1}^n z(x_i, W_i) z(x_i, W_i)'$, it is implausible that the regression
 estimator is consistent but $\hat G_n$ is far from $G_n$. Appendices D and F.2 of a
 working paper version of this paper provide formal guarantees for $\hat G_n$
 \citep{chen2024potential}.

 Because $G_n$ depends on the unknown $\bpi^*$, treating
 $G_n$ as known implicitly restricts $\bpi^*$ to those designs that are consistent with
 the realized treatment assignment. We discuss its interpretation further in \cref
 {asec:known_gram_matrix}.} 
Under these definitions, $\tau$ is a linear aggregation of potential outcomes:
For $\pi_i^* = \pi_i^*(1)$, \[
\tau = \Lambda\beta = \frac{1}{n}\sum_{i=1}^n \pi_i^* \underbrace{\Lambda G_n^{-1} z(x_i,
1)}_{\rho_i(1)} y_i (1) + (1-\pi_i^*) \underbrace{\Lambda G_n^{-1} z(x_i, 0)}_{\rho_i(0)}
y_i(0). \numberthis \label{eq:simple_aggregation}
\]
Here, the \emph{potential weights} $\rho_i(w) = \Lambda G_n^{-1} z(x_i, w)$ are known up
to $G_n$. 
In the case of 
\eqref{eq:angrist98intro} where $x_i$ includes a constant, we can compute $\rho_i(w)$ in
closed form:
\begin{align*}
\rho_i(w) = \frac{w -  x_i'\delta}{\frac{1}{n}\sum_{i=1}^n  x_i'\delta
(1- x_i'\delta)} \quad  \delta \equiv \pr{\frac{1}
{n}\sum_ {i=1}^n  x_i  x_i'}^{-1}\pr{\frac{1}{n} \sum_{i=1}^n  x_i
\pi_i^*}.
\end{align*}
$\rho_i(w)$ is proportional to $w-x_i'\delta$, where $\delta$
 is the projection coefficient of $\pi^*$ on $x_i$.

If the regression is quasi-experimental in the sense of \cref{item:M}, the true design
$\bpi^*$ is such that the estimand \eqref{eq:simple_aggregation} satisfies \emph{level
independence} \citep{blandhol2022tsls}:

\begin{defn}
\label{defn:level_irrelevance_simple} We say that $\tau$ is \emph{minimally
quasi-experimental} under
$\bpi^*$ if $\tau$ is always unchanged when we replace all potential outcomes $y_i(w)$
 with $y_i (w)  + c_i$ for arbitrary $c_i \in \R$, holding fixed $(\bpi^*, x_1,\ldots, x_n)$. Since
 $\tau$ is a linear aggregation, equivalently, $\tau$ is minimally quasi-experimental if
 there are some
 $\omega_1,\ldots, \omega_n \in \R$, not dependent on $y_i(\cdot)$, such that $   \tau =
 \frac{1}{n} \sum_{i=1}^n \omega_i (y_i(1) - y_i(0))$ for all choices of $y_i(1), y_i (0)
 \in \R$.
\end{defn}

\Cref{defn:level_irrelevance_simple} is a natural minimal requirement for
quasi-experimental estimands. It imposes that a quasi-experimental estimand should be
invariant to any changes to the potential outcomes that do not alter individual treatment
effects---holding fixed the treatment assignment process. For linear estimands, this
condition is equivalent to $\tau$ being a weighted average treatment effect (these weights
$\omega_i$ may be negative).\footnote {Negative weights are intended, for example, when
the estimand is meant as a contrast of subgroup average effects. Thus to preserve
generality, we allow for negative weights. Since we could recover the estimand itself, we could additionally inspect
whether the weighting is convex or whether it satisfies further restrictions. 
\citet{blandhol2022tsls} term an
estimand ``weakly causal'' if it additionally satisfies $\omega_i \ge 0$.}

Importantly, estimands that rely on modeling $y_i(0)$ do not qualify as quasi-experimental
per our definition. This definition thus rules out difference-in-differences under
parallel-trends assumptions, or assumptions like $\E[Y_i(0) \mid W_i=1] = \E [Y_i(0) \mid
W_i = 0]$ that allow the ATT to be identified by the difference-in-means estimand. These
empirical strategies do not mimic a \emph{randomized} experiment (though
quasi-experimental assignment is often invoked to informally justify them). While looser
definitions of quasi-experiments are reasonable \citep[e.g.,][]{card2022design}, we argue
this stricter one is both principled and useful. It is principled by taking very seriously
that quasi-experiments should emulate randomized experiments
\citep{angrist2010credibility,leamer1983let}. It is also not overly stringent---the
theoretical applications in \cref{sec:theoretical_applications} show that much of the
applied econometrics literature is consistent with this definition.

Returning to \eqref{eq:simple_aggregation}, observe that $\tau$ satisfies 
\cref{defn:level_irrelevance_simple}
under
$\bpi^*$ if and only if \[
\pi_i^* \rho_i(1) + (1-\pi_i^*)\rho_i(0) = 0 \text{ for all $i=1,\ldots,n$.} \numberthis
\label{eq:equation_simple}
\]
We can separate \eqref{eq:equation_simple} into two further questions: \cref{item:M1} which
assignment vectors $\bpi = (\pi_1,\ldots, \pi_n) $
solve \eqref{eq:equation_simple}, and \cref{item:M2} whether the true assignment vector $\bpi^*$ is
plausibly among those solutions. For \cref{item:M1}, viewing 
\eqref{eq:equation_simple} as an equation in $\bpi^*$, we can solve for \[
\pi_i = \frac{-\rho_i(0)}{\rho_i(1) - \rho_i(0)}, \numberthis 
\label{eq:implicit_design_formula_binary}
\]
and we call such a $\bpi$ an \emph{implicit design}.

Identifying $\bpi$ immediately pinpoints the estimand. If $\pi_i$ were the true design,
then, for $\omega_i(\bpi, w) \equiv \pi_i(w)\rho_i(w)$ and $\omega_i \equiv \omega_i
(\bpi, 1) = -\omega_i(\bpi, 0)$ by \eqref{eq:equation_simple}, $\tau$ would be a weighted
ATE
\[
    \tau = \frac{1}{n} \sum_{i=1}^n \omega_i(\bpi, 1) y_i(1) + \omega_i(\bpi, 0) y_i(0) = 
    \frac{1}{n} \sum_{i=1}^n \omega_i (y_i(1) - y_i(0)). 
\] Thus, simply solving \eqref{eq:equation_simple} yields both candidate designs $\bpi$
 and their corresponding \emph{implicit} estimands. Crucially, \eqref
 {eq:equation_simple} only involves objects that are known in the population and are
 hence easily estimable in sample. This immediately enables assessing \cref
 {item:M} through \cref{item:M1} and examining how the regression aggregates
 heterogeneous treatment effects.

The following corollary of 
\cref{thm:main}, to be introduced, summarizes how implicit designs and implicit estimands
relate to \cref{item:M}:
\begin{restatable}{cor}{corbinary}
    \label{cor:binary_main}
    When $k=T=J=1$, $\tau$ is minimally quasi-experimental 
     if and only if
\begin{enumerate}
    \item $\rho_i(1) \rho_i(0) \le 0$ for all $i$. Some implicit design $\bpi$ satisfies
    \eqref{eq:gram_criterion} and has $\pi_i =
    \frac{-\rho_i (0)} {\rho_i
    (1) - \rho_i(0)}$ for all $i$ with one of $\rho_i(1)$ and $\rho_i(0)$ nonzero.
    \item For all units $i$ with one of $\rho_i(1)$ and $\rho_i(0)$ nonzero, $\pi^*_i =
    \frac{-\rho_i(0)}{\rho_i(1) - \rho_i(0)}$.
\end{enumerate}
    When this happens, the implicit estimand is $
    \tau = \frac{1}{n} \sum_{i=1}^n \omega_i^* (y_i(1) - y_i(0))
    $
  for  $\omega_i^* \equiv \omega_i(\bpi^*, 1) =
    \pi_i^* \rho_i(1).$ The  weight $\omega_i^* < 0$ if and only if $\rho_i(1) < 0 <
    \rho_i(0)$.
\end{restatable}

The two conditions in \cref{cor:binary_main} separate \cref{item:M} into \cref{item:M1}
and \cref{item:M2}.
\cref{cor:binary_main}(1) formalizes \cref{item:M1}. 
  If an implicit
 design exists, it is
 uniquely and explicitly defined (up to units with $\rho_i(1) = \rho_i(0) = 0$). 
 \Cref{cor:binary_main}(2) formalizes \cref{item:M2}, which requires that $\pi_i^*$ is
  equal to the unique implicit design $\frac{-\rho_i(0)}{\rho_i(1) - \rho_i(0)}$. 

  There are at least two ways that \cref{item:M1} is rejected. The more obvious one is if
  $\pi_i
\not\in [0,1]$ for any $i$. When this happens, the implicit design is not even a
probability distribution. More subtly, $\pi_i$ is also indefensible if it \emph
{generates} a Gram matrix that is different from $G_n(\bpi^*)$: \[G_n(\bpi) =
\frac{1}{n}\sum_{i=1}^n \pi_i z(x_i, 1)z(x_i, 1)' + (1-\pi_i) z(x_i, 0)z(x_i, 0)' \neq
G_n(\bpi^*).
\numberthis \label{eq:gram_criterion}
\] This restriction is useful when we analyze specifications theoretically under this
 framework. It is harder to implement when we do not know and have to estimate $G_n$,
 though, with a confidence set for $G_n$, one could use it as a basis for inference on
 $\bpi^*$.\footnote{See Appendix D  of a working paper version of this paper 
 \citep{chen2024potential}.}

Here, the implicit estimand is a weighted average treatment effect, where weights
 $\omega_i$ are all nonnegative provided no unit has $\rho_i(1) < 0 <
\rho_i(0)$.
Applied to \eqref{eq:angrist98intro}, \cref
{cor:binary_main} shows that the implicit design is precisely $\pi_i = x_i'\delta$ and
the corresponding estimand is a weighted ATE\[
\tau = \frac{1}{n} \sum_{i=1}^n \omega_i(y_i(1) - y_i(0)) \quad \omega_i = \frac{\pi_i
(1-\pi_i)} {\frac{1}{n}\sum_{j=1}^n\pi_j
(1-\pi_j)} \numberthis \label{eq:angrist_estimand_simple}.
\]
Computing them thus recovers results in \citet{angrist1998} and
\citet{blandhol2022tsls}.\footnote{Both \citet{angrist1998} and \citet{blandhol2022tsls}
consider a superpopulation sampling setup. \citet{angrist1998} considers a binary $x_i$ in
his equation (9), but the argument can be easily generalized, e.g., in
\citet{borusyak2024negative,goldsmith2022contamination}. Corollary 1 in
\citet{blandhol2022tsls}---which specializes their TSLS result to OLS---shows that assuming unconfoundedness, $\tau$
 is a positively weighted average treatment effect if and only if the propensity score is
 linear. This is effectively what we find, and thus we view our result (formally in
\cref{thm:zoo}(\ref{item:angrist98})) as a reinterpretation of theirs. We clarify the
exact relationship to Theorem 1 in \citet{blandhol2022tsls} in 
\cref{sub:relation_to_blandhol}. } 

To summarize, our analysis proceeds in four steps:
\vspace{0.2em}
\begin{mdframed}
\begin{enumerate}[label=(\roman*)]
    \item\label{item:i} We treat the triplet $(\br{z(x_i, \cdot)}_{i=1}^n, \Lambda, G_n)$
    as known (at least in the population).
     
     \item\label{item:ii} We write the population regression estimand $\tau$ in the form  \eqref{eq:simple_aggregation}. Because we treat $G_n$ as known,
      the \emph{potential weights} $\rho_i (w)$ are known for all units.

     \item\label{item:iii} We observe that \cref{item:M} imposes linear restrictions on
     $\pi_i^*$, where the coefficients are the potential weights.

      \item\label{item:iv} Separating \cref{item:M} into \cref{item:M1,item:M2}, we call
      the solutions to these linear equations \emph{implicit designs}.     Computing
 implicit designs also yields the corresponding estimands by \eqref
 {eq:simple_aggregation}. If \cref{item:M} holds, then $\tau$ has a quasi-experimental
 interpretation, and one may then assess the extent to which $\tau$ is substantively
 relevant.
 \end{enumerate}
 \end{mdframed}
 \vspace{0.2em}

We conclude this subsection by stating the
 superpopulation analogue of these results.

\begin{rmksq}[Superpopulation]
\label{rmk:superpop} Suppose instead $ (Y_i (0), Y_i(1), W_i, X_i) \iid P$. We can convert
 the sampling setup to a design-based setup by setting $\pi_i^* = P(W_i = 1 \mid X_i, Y_i
 (1), Y_i(0))$ and conditioning on $(X_i, Y_i(1), Y_i(0))$. Now, consider a hypothetical
 set of potential outcomes $Y_i' (w) = Y_i (w) + C_i$ where $C_i$ is some random variable
 satisfying $C_i \indep W_i \mid X_i, Y_i(\cdot)$. This independence restriction makes
 sure that $Y_i'(\cdot)$ does not introduce new selection concerns: $P(W=1 \mid Y_i'(1),
 Y_i'(0), Y_i(1), Y_i(0), X_i) = \pi_i^*.$

 The sampling analogue of
 \cref{defn:level_irrelevance_simple} is that $\tau$ is unchanged for all
  such $Y_i'$: \begin{align*}
 \tau &= \E[\underbrace{\Lambda \E[z(X_i, W_i) z(X_i, W_i)']^{-1} z(X_i, W_i)}_{\rho_i
 (W_i)} Y_i(W_i)] = \E[ \rho_i(W_i) Y_i' (W_i)]. \numberthis \label{eq:sampling_cond}
 \end{align*} By the law of iterated expectations, conditioning on $(Y_i(1), Y_i(0), C_i,
  X_i)$, \eqref{eq:sampling_cond} is equivalent to $0 = \E\bk{C_i \pr{\pi_i^* \rho_i
  (1) + (1-\pi_i^*) \rho_i(0)}}
 $. 
 Since we can choose $C_i$ as an arbitrary function of $X_i, Y_i(1), Y_i(0)$ and in
 particular as $C_i = \pi_i^* \rho_i(1) + (1-\pi_i^*) \rho_i(0)$, we can force the
 following   analogue of \eqref{eq:pop_level_irrelevance_condition}: \[
     \pi_i^* \rho_i(1) + (1-\pi_i^*) \rho_i(0) = 0 \quad \text{$P$-almost surely.}
 \]
 See \cref{asec:superpop} for a formalized analogue with general
     $J,k,T$.
\end{rmksq}

\subsection{General setup}
\label{sub:main_results}

We now generalize to panel data and multivalued treatments. Return to the general
regression \eqref{eq:general_regression}, with general outcome dimension $T$, number of
treatments $J+1$, and coefficient dimension $k$. Following \cref{item:i}, let $\bz
(\bx_i,\cdot)
\in \R^{T\times K}$ stack $z_t
 (\bx_i, \cdot)$; we treat $ (\Lambda, G_n, \bz (\bx_1, \cdot),
    \ldots, \bz(\bx_n, \cdot))$ as known.
As in \cref{item:ii}, the regression estimand is a linear aggregation of potential
outcomes \begin{align*}
    \btau  &= \frac{1}{n}\sum_{i=1}^n \sum_{\bw \in \cW}
    \pi_i^*(\bw) \cdot \underbrace{\bk{
         \Lambda G_n^{-1}\bz(\bx_i, \bw)'
    }}_{\brho_i(\bw) \in \R^{k \times T}}\by_i(\bw).
    \numberthis
    \label{eq:estimand}
    \end{align*} Relative to the simple case
    \eqref{eq:simple_aggregation}, we sum over $J+1$ values, and potential
    weights $
    \brho_i(\bw) \equiv \Lambda G_n^{-1} \bz(\bx_i, \bw)'
    $ are matrices of dimension $k \times T$. The potential weights do not depend on how
    the regression is parametrized, nor are they changed by Frisch--Waugh--Lovell
    transforms (\cref{asub:invariances}).\footnote{There are two subtleties of
    parametrization in panel settings. First, since treating $G_n$ as known is motivated
      by its consistent estimation, we require representing fixed effects through the
      within-transformation for $\bz(\bx,
    \cdot)$, rather than through unit-level dummy variables, since the latter does not
    lead to consistently estimable Gram matrices. Second, assuming $\bz (\bx,
  \cdot)$ is known  precludes mediators (e.g. lagged outcomes) in the right-hand side of
  the regression, since we do not know counterfactual values of the mediator.
}
For \cref {item:iii}, a natural generalization of \cref{defn:level_irrelevance_simple}
imposes that the estimand is invariant to shifts in potential outcome
\emph{paths} that do not alter treatment effects:
\begin{defn}[Minimally quasi-experimental]
\label{defn:level_irrelevance}
$\btau$ is minimally quasi-experimental if it is always unchanged when we replace all potential
outcomes $\by_{it}(\bw)$ with $\by_{it} (\bw) + c_ {it}$ for arbitrary $c_{it}
\in \R$, fixing $\bpi^*, \bx_1,\ldots, \bx_n$. For linear estimands, this is equivalent
to 
\begin{align*}
\btau &= \frac{1}{n} \sum_{i=1}^n \sum_{\bw\in \cW} {\bomega_{i}
   (\bw)} \by_ {i}(\bw) \text{ for some $\bomega_{i} (\bw) \in \R^{k\times T}
   $ where }
   0 = \sum_ {\bw
  \in \cW} \bomega_{i}
   (\bw)
\end{align*}
\end{defn}

\Cref{defn:level_irrelevance} is equivalent to the following linear system \[\text{For $i
= 1,\ldots,
n$}, \sum_{\bw \in \cW}\pi_i^*(\bw) \brho_i(\bw) = 0 ,\quad \sum_ {\bw \in \cW} \pi_i^*
(\bw) = 1.
\numberthis
\label{eq:pop_level_irrelevance_condition}
\] Since $\rho_i(w)$ is a $k\times T$ matrix and $|\mathcal W| - 1 = J$, there are $kT$
 restrictions in $J$ unknowns. We call any solution an \emph{implicit design}. Implicit
 designs are typically unique when they exist, because often the number of equations $kT$
 is at least the number of unknowns $J$.\footnote{For instance, $J + 1$ treatments
 generate $k = J$ contrasts; panels under staggered adoption admit fewer unique treatment
 times $(J + 1)$ than time horizon $T$. Lemma F.4 in a working paper draft of this paper
 verifies uniqueness when $T=1$ \citep{chen2024potential}.} In practice, we can estimate
 implicit designs by solving sample analogues of
 \eqref{eq:pop_level_irrelevance_condition}, replacing $G_n$ with its sample counterpart
  $\hat G_n = \frac{1}{n} \sum_{i=1}^n z(x_i, W_i) z(x_i, W_i)'$.

For a given implicit design, the corresponding \emph{implicit
estimand} is the following, for $\bomega_i(\bpi,\bw) \equiv \pi_i(\bw) \brho_i(\bw)$: 
\begin{align*}
\btau(\bpi) &\equiv \frac{1}{n} \sum_{i=1}^n \sum_{\bw \in \cW} \bomega_i(\bpi, \bw) \by_i
(\bw) = \frac{1}{n} \sum_{i=1}^n \sum_{\bw \neq \bw_0} \bomega_i(\bpi, \bw) (\by_i
    (\bw) - \by_i(\bw_0)).
    \numberthis
    \label{eq:implicit_estimand}
\end{align*}
These observations result in the following theorem formalizing how implicit designs answer
\cref{item:M1} and \cref{item:M2}, as in \cref{item:iv}. For a given implicit design
$\bpi$, we call it \emph{proper} if all $\pi_i(\cdot)$ are probability distributions. We
say it \emph{generates $G_n$} if it satisfies \eqref{eq:gram_criterion}: $G_n(\bpi) =
G_n$.

\begin{restatable}{theorem}{thmmain}
    \label{thm:main}

$\btau$ is minimally quasi-experimental if and only if (1) some implicit design $\bpi$
exists, is proper, and generates $G_n$, and (2) the true design $\bpi^*$ is equal to
$\bpi$. When this happens, the estimand $\btau$ is equal to the implicit estimand under
$\bpi$.
\end{restatable}

\Cref{thm:main} separates \cref{item:M} into an objectively computable question
\cref{item:M1} and a substantive question \cref {item:M2}. Proper implicit designs that
generate $G_n$ answer \cref{item:M1}. If no such designs exist, then $\tau$ cannot be
minimally quasi-experimental. On the other hand, judging  whether the true design is
plausibly equal to an implicit design $\bpi$ \cref{item:M2} is context-specific. Computing
implicit designs makes this judgment concrete and transparent.

In practice, both can be assessed by estimating implicit designs through solving the
empirical analogue of
\eqref{eq:pop_level_irrelevance_condition} by replacing $G_n$ with $\hat G_n$. If 
\cref{item:M} is rejected through either \cref{item:M1} or \cref{item:M2}, then the
regression does not have a quasi-experimental interpretation, in the sense that the
potential outcomes can be arranged adversarially so that the regression estimand is
arbitrarily far from any aggregation of individual causal effects.

\subsection{A novel double robustness for the OLS estimator}  

\Cref{thm:main} formalizes when the regression estimand is \emph {exactly}
 quasi-experimental. In many cases, a proper implicit design may not exist, or it may be
 implausible that \cref{item:M2} holds even if a proper implicit design were to exist. For
 instance, in the leading case of \eqref{eq:angrist98intro}, the implicit design is the
 linear projection of $\pi^*$ onto the covariates. This linear projection can have $\pi_i
 \not\in [0,1]$; linear functional form restrictions on the propensity score are also
 frequently rejected \citep{blandhol2022tsls}. Likewise,
 \cref{sec:theoretical_applications} shows that exact quasi-experimental interpretation
 for many specifications is fragile. 

Nevertheless, one may expect that if either the implicit design is somehow ``close'' to
the true design or the regression specification provides a reasonable approximation of
the outcomes, then the regression estimand is still approximately quasi-experimental for
some causal effect. The next result formalizes this intuition through a numerical
equivalence. 

Suppose one wishes to target some causal contrast  $\frac{1}{n}\sum_{i}
\sum_w \tilde \omega_i(w) y_i (w)$ under some explicitly hypothesized design $\tilde
\bpi$, possibly
equal to $\bpi^*$.  Suppose too that the regression functional form $z(w,
x)'\beta$ is a plausible approximation of conditional means of potential outcomes. Then
an attractive estimator for $\frac{1} {n}\sum_{i=1}^n \sum_{w\in \cW} \tilde \omega_i(w)
y_i(w)$ is the AIPW estimator which uses the hypothesized design as propensity weights
and the regression as an outcome model:
\[
    \tilde\tau_{\mathrm{AIPW}} = \frac{1}{n} \sum_{i=1}^n \sum_{w\in \cW} \tilde \omega_i
    (w) \br{\frac{\one(W_i = w)} {\tilde \pi_i(w)} (Y_i - z(x_i, w)'\hat\beta)  + z(x_i,
    w)'\hat\beta}. \numberthis \label{eq:aipw}
\]
It is well-known that $\tilde\tau_{\mathrm{AIPW}}$ is doubly robust
\citep{bang2005doubly}: $\tilde\tau_{\mathrm{AIPW}}$  recovers the target estimand if
either $z(x_i, w)'\hat\beta$ is correctly specified or if $\tilde\bpi=\bpi^*$.

We analyze how far the OLS estimator $\Lambda \hat\beta$ is from this AIPW estimator, for
\emph{some} $\tilde \omega$.

\begin{restatable}{theorem}{thmaipwnew}
\label{thm:aipw_new}

Let $\tilde \bpi$ be strictly proper ($\tilde\pi_i(\cdot) > 0$) and let
$\hat\beta$ be the OLS estimator of \eqref{eq:general_regression}. 
\begin{enumerate}[wide]
    \item If $\tilde\bpi$ is an implicit design that generates $G_n$, then
    $\Lambda\hat\beta$ is the AIPW
    estimator for the implicit estimand: $\Lambda\hat\beta = \tilde\tau_{\mathrm{AIPW}}$,
    for $\tilde \omega_i(w) = \tilde\pi_i(w)\brho_i(w)$ the implicit estimand.\footnote{
    \cref{thm:aipw_new}(1) is closely related to Proposition 3.2 in
\citet{bruns2025augmented} and to section 3 of \citet{robins2007comment}. Applying
Proposition 3.2 in \citet{bruns2025augmented} would show that $\hat\tau_{
\mathrm{AIPW}}$
is numerically equivalent to the imputation estimator targeted to implicit
estimand $ \frac{1} {n} \sum_ {i} \sum_w
\omega_i (w) z(x_i, w)'\hat \beta$, and further algebra shows that this imputation
estimator is numerically equivalent to the OLS coefficients $\hat\tau_{\mathrm{OLS}} =
\Lambda
\hat\beta$. Discussions in
\citet{bruns2025augmented} and \citet{robins2007comment} mainly focus on cases where
 regressions are separately fit within treatment groups; \cref{thm:aipw_new}(1) allows the
 regression specification  to be arbitrary over the entire sample.
 \citet{sloczynski2025covariate} show related numerical equivalence results for estimators
 of average treatment effects. }

    \item Suppose instead that, after reordering and rescaling regressors,
    $
        z(x_i, w) = (z_{1}(x_i, w)', z_2(x_i)')' $ for a dimension-$k$ $z_1$, $\Lambda =
[I_k, 0]$, and a rank condition \cref{as:aipw_recentered_rank} holds. Then there are
weights $\tilde \omega_i(w)$, with $\sum_w \tilde \omega_i(w) = 0$, and a nonsingular
$k\times k$ matrix $\tilde A$, defined in
\eqref{eq:tilde_omega}, such that
    \[
        \tilde\tau_{\mathrm{AIPW}}-\Lambda\hat\beta
        =
        \frac{\tilde A}{n} \sum_i\tilde\brho_i \cdot (Y_i - z(x_i,
        W_i)'\hat\beta),
        \text{ where }
        \tilde\brho_i
        =
        \sum_{w\in\cW}\tilde\pi_i(w)\brho_i(w). \numberthis \label{eq:ols-aipw}
    \]

\item Specializing part (2) to the regression in \eqref{eq:angrist98intro}, under the
assumptions of \cref{tab:regression_specifications_from_the_literature}(1), we have that
part (2)
holds with $\tilde\omega_i (1) =
\frac{\tilde \pi_i(1-\tilde \pi_i)}{\frac{1}{n} \sum_j \tilde \pi_j (1-\tilde \pi_j)}$
and 
\[
    \tilde\tau_{\mathrm{AIPW}} - \Lambda \hat\beta
    =
    -\frac{1}{\frac1n\sum_j\tilde\pi_j(1-\tilde\pi_j)} \frac1n
    \sum_{i=1}^n(\tilde\pi_i-\pi_i)(Y_i - z(x_i, W_i)'\hat\beta),
\] where $\tilde\pi_i=\tilde\pi_i(1)$ and $\pi$ is the linear projection of $\pi^*$ on the
 covariates.
	
\end{enumerate}

\end{restatable}

Remarkably, when $\tilde \bpi$ is equal to a proper implicit design (thus requiring
\cref{item:M1}), the OLS estimator is numerically equivalent to \eqref{eq:aipw} with
$\tilde \omega_i$ indexing the implicit estimand. In general, it turns out that there
\emph{exists} some estimand $\tilde \omega_i (\cdot)$ for which the gap between OLS and
AIPW is proportional to the empirical covariance between the extent to which $\tilde \bpi$
violates
\eqref{eq:pop_level_irrelevance_condition} and the OLS residuals: The quantity $\tilde
\rho_i$ in \eqref{eq:ols-aipw} is exactly the former, where $\tilde\rho_i = 0$ if
$\tilde \bpi$ is an implicit design.

Applying \eqref{eq:ols-aipw} to the specification in \citet{angrist1998} shows that the
OLS-AIPW gap is exactly proportional to the covariance between $\tilde \pi_i - \pi_i$,
for $\pi_i$ the implicit design, and
the regression residual. Here, the estimand $\tilde \omega_i$ is the weighted average
treatment effect with weights proportional to $\tilde \pi_i(1-\tilde \pi_i)$. This gives
\eqref{eq:angrist98intro} an \emph{approximate} quasi-experimental interpretation even
 when \cref{item:M} fails: It is approximately an AIPW estimator \eqref{eq:aipw}  for the
 variance-weighted estimand under the true design $\bpi^*$---if the linear projection
 residuals $\pi^*_i - \pi_i$ and $Y_i - z(x_i, W_i)'\hat\beta$ are not very correlated.
 Empirically, as we show in \cref{sec:empirical}, this covariance seems frequently small,
 driven by a lack of correlation between $\tilde \pi -
 \pi$ and the regression residuals.

To summarize, this section characterizes the necessary implications of \cref {item:M}
by breaking it into \cref{item:M1} and \cref{item:M2} and defining implicit designs.
A regression coefficient has a quasi-experimental interpretation if both are satisfied: An
implicit design exists, and the true design is exactly equal to it (\cref {thm:main}).
From this perspective, an implicit design, if it exists, can be thought of as a model of
treatment assignment implicitly chosen by the regression specification. Again, since
regressions are often used to estimate causal effects without an explicit model, making
this implicit model explicit is useful and enables important debates about whether it is
correctly specified. 

Another benefit of making the design explicit is to assess whether the regression is
nevertheless approximately quasi-experimental under the explicit design $\tilde \bpi$. It
is still approximately robust for some causal effect under $\tilde \bpi$, so long as the
regression's failures as a model of treatment and its failures as a model of outcome are
uncorrelated (\cref{thm:aipw_new}).

 {
\begin{landscape}
\renewcommand{\arraystretch}{1.55}
\begin{table}[h]
\begin{tabularx}{1\linewidth}{@{} l l l l X @{}}
\toprule
  \# & Setting & Specification & Contrast & Additional conditions\\
\midrule
    (1) & $k=T=J=1$ & $Y_i=\tau W_i + x_i'\gamma
    +\epsilon_i$ & $\tau$ &  $x_i$ includes a constant \\ 
    (2)  & $k=J, T=1$ & $Y_i=\sum_{j=1}^J \tau_j W_{ij} +
    x_i'\gamma +
    \epsilon_i$ & $(\tau_1,\ldots, \tau_J)$ & $x_i$ includes a constant. $\mathcal W = 
    \br{0,\ldots, J}$, $W_ {ij}=\one
    (W_i = j)$ \\
    (3)  & $k=T=J=1$ &
    $Y_i = \alpha_0 + \gamma_1'x_i + \tau W_i + W_i (x_i-\bar x)'\gamma_2 + \epsilon_i$ &
    $\tau$ & $x_i$ saturated for some discrete covariate $x_i^*$ taking values in $
    \br{0,\ldots, L}$: $x_i = [x_ {i1},\ldots,
    x_ {iL}]'$ for $x_ {i\ell} = \one
    (x^*_i = \ell)$, $\bar x = \frac{1}{n} \sum_i x_i$ \\
    (4)  & $k=T=J=1$ & $Y_i = \alpha_0 + \gamma_1'x_i + \tau W_i + W_i(x_i
    - \bar x_1)'\gamma_2 + \epsilon_i$ & $\tau$ & $\bar x_1 = \frac{\sum_i \pi_i^* x_i}
    {\sum_i \pi_i^*}$ \\ 
    (5)  & $T > 1$ & $Y_{it} = \alpha_i + \mu_t + \tau W_{it} +
    \epsilon_{it}$ & $\tau$ & $\mathcal W \subset \br{0,1}^T$ is the set of treatment
     paths. The nonzero elements of $\mathcal W$ are linearly independent vectors whose
     span excludes $1_T = (1,\ldots,1)'$. This condition is satisfied by staggered
     adoption that excludes always-treated units.
    \\\bottomrule
\end{tabularx}
\caption{Regression specifications analyzed in \cref{thm:zoo}}
\label{tab:regression_specifications_from_the_literature}
\begin{proof}[Notes]
    (1) is discussed in \citet{angrist1998} and section 2.1 of \citet{blandhol2022tsls}; 
    (2) is discussed in \citet{goldsmith2022contamination}; (3) is discussed in 
    \citet{miratrix2013adjusting,imbens2009recent,lin2013agnostic}, among others; (4) is
    discussed in
    \citet{kline2011oaxaca}; (5) is discussed in \citet{athey2018design}. 
\end{proof}
\end{table}
\renewcommand{\arraystretch}{1.0}
\end{landscape}

}

\section{Theoretical applications and examples}
\label{sec:theoretical_applications}

We now illustrate the generality of this framework by applying it to a wide swath of
regression specifications and  discussing them in self-contained vignettes.  To emphasize,
our results essentially reduce the problem to \emph{computing} the potential weights and
the set of implicit designs---unifying results across starkly distinct settings. We start
with several specifications with known causal interpretations under specific designs
\citep
{angrist1998,goldsmith2022contamination,imbens2009recent,lin2013agnostic,kline2011oaxaca,athey2018design}.

\subsection{A unified analysis of quasi-experimental interpretation in regression}
\label{sub:zoo}
Assume throughout that the population Gram matrix is invertible.

\begin{restatable}{theorem}{thmzoo}
\label{thm:zoo} We compute the implicit designs and estimands of the  regression
 specifications (1)--(5) described in \cref
 {tab:regression_specifications_from_the_literature}. In every specification, the
 implicit design exists uniquely. The implicit design generates $G_n$ regardless of
 whether $\bpi=\bpi^*$, for all specifications except (\cref{item:kline}).
\begin{enumerate}
    \item \label{item:angrist98}
    \begin{enumerate}
        \item  $\pi_i = x_i'\delta$ for $\delta = \pr{\frac{1}{n}\sum_{i=1}^n
    x_ix_i'}^{-1} \frac{1}{n} \sum_{i=1}^n
    \pi_i^* x_i$
    \item $\pi_i^* = \pi_i$ if and only if $\pi_i^* = x_i'\delta$
    \item $\omega_i \equiv \omega_i(\bpi, 1) = -\omega_i(\bpi,0) = \frac{\pi_i
    (1-\pi_i)} {\frac{1}{n} \sum_{i=1}^n \pi_i (1-\pi_i)}$. When $\pi_i^* = \pi_i$,
    $\omega_i \ge 0$.     
    \end{enumerate}
    \item \label{item:pgp22}
    \begin{enumerate}
        \item $\pi_i(j) = x_i'\delta_j$ for  $\delta_j = \pr{\frac{1}{n}\sum_{i=1}^n
    x_ix_i'}^{-1} \frac{1}{n} \sum_{i=1}^n
    \pi_i^*(j) x_i$
    \item $\pi_i^* = \pi_i$ if and only if $\pi_i^*(j) = x_i'\delta_j$ for all $j \in [J]$
    \item The implicit estimand is shown in \eqref{eq:implicit_estimand_pgp}. This
    estimand is generally contaminated (that is, $\omega_{ij}(\bpi, \ell) \neq 0$ for some
    $j \in [J]$ and $\ell    \not\in \br{0,j}$).

    \end{enumerate}
    \item \label{item:imbens-wooldridge} 
    \begin{enumerate}
        
        \item The implicit design
        equals the mean of $\pi_i^*$ among the units with the same $x_i$-value

        \item $\pi_i^* = \pi_i$ if and only if $\pi_i^*$ is the same for all units with
        the same $x_i$-value
        
        \item The implicit estimand is the ATE. That is, $\omega_i = \omega_i(\bpi, 1) =
        -\omega_i(\bpi,0) = 1$. 
    \end{enumerate}

    \item \label{item:kline} 
    \begin{enumerate}
        \item $\pi_i = \frac{\delta_0 + (x_i-\bar x)'\delta_1}{1+\delta_0 + 
        (x_i-\bar x)'\delta_1 }$, where $\delta_0, \delta_1$ are equal to the population
        weighted least-squares coefficients of $\pi_i^*/(1-\pi_i^*)$ on $x_i-\bar x$ and a
        constant, \emph{weighted by} $1-\pi_i^*$

        \item $\pi_i^* = \pi_i$ if and only if $\pi_i^*/(1-\pi_i^*) = \delta_0 +
        \delta_1'(x_i-\bar x)$
        \item When $\bpi^* = \bpi$, the implicit estimand is the ATT: $\omega_i =
        \omega_i (\bpi, 1) = \frac{\pi_i}{\frac{1}{n} \sum_{i=1}^n
        \pi_i}$.
    \end{enumerate}

    \item \label{item:athey-imbens} 
    \begin{enumerate}
        \item The implicit design is constant in $i$ and is unique, $\pi_i(\bw) = \frac{1}
        {n} \sum_
        {i=1}^n
        \pi_i^* (\bw)$
        \item $\pi_i^* = \pi_i$ if and only if $\pi_i^*$ is the same for all $i$
        \item The implicit estimand is shown in \eqref{eq:twfe_estimand}, which matches
        Theorem 1(ii) in \citet{athey2018design} under staggered adoption.\footnote{One
        might wish to further impose that the post-treatment weights are nonnegative
        (i.e., $\bomega_{it}(\bpi^*, \bw) \ge 0$ if $\bw_t = 1$). Failure of this
        condition implies that post-treatment units are very much used as comparisons for
        newly treated units 
        \citep{roth2023s,borusyak2024revisiting,de2020two,goodman2021difference}.
        \Cref{prop:forbidden_comparisons} shows that when $\cW$ only has two elements and
        includes a never-treated unit, all weights post-treatment are nonnegative, but
        such forbidden comparisons are possible in \emph{all} other cases.
        Additionally, \cref
 {sub:additional_panel} extends the analysis to one-way FE and event-study designs.}
    \end{enumerate}
\end{enumerate}
\end{restatable}

\Cref{thm:zoo} computes implicit designs and estimands for several specifications that
are analyzed piecemeal in the literature. Simply examining
 \cref{item:M1} shows the implicit design exists, is unique, and matches the form studied;
 the implicit estimand matches as well. \Cref{thm:zoo} is thus a set of \emph{converses} to the
 existing results---the regression estimand satisfies \cref{item:M} \emph{only if} $\bpi^* =
 \bpi$ and the target causal effect is the implicit estimand. These necessity results are
  new, to our knowledge, except for \cref{thm:zoo}(\ref{item:angrist98}).  These
  calculations, combined with \cref{thm:aipw_new}(1), also immediately imply regression
  estimators are equivalent to AIPW for the implicit estimand---regardless of
  \cref{item:M2}---for all but (\cref {item:kline}).

We next turn to a set of novel results by computing the implicit design for specifications
generalizing \cref{thm:zoo}(3)--(5). \Cref{sub:forbidden_interactions_} studies interacted
specifications like \cref{thm:zoo}(3)--(4), but requires causal interpretation for
\emph{both} the coefficients on $W_i$ and on $W_i x_i$.  \Cref{sub:panel} studies TWFE's
quasi-experimental interpretation with additional covariates. Both sets of results are
broadly negative: Exact quasi-experimental interpretation only exists in certain
knife-edge scenarios.

\subsection{Forbidden interactions}
\label{sub:forbidden_interactions_}

Assume $T=J=1$ and split $x_i$ into subvectors $x_{1i},x_{2i}$ (possibly overlapping).
Consider the specification \[Y_i =
\gamma_0 +
\tau_0 W_i + \tau_1' W_i x_{1i} + \gamma_1'x_{2i} + \epsilon_i.
\numberthis \label{eq:interaction}
\] Viewed as an outcome model, $\tau_0$ is the treatment effect for a baseline covariate
 value, and $\tau_1$ captures how treatment effect varies with $x_1$. One might hope that
 even without the outcome model, $\btau = (\tau_0, \tau_1')$ retains causal
 interpretation in a more flexible manner than the specification 
 \eqref{eq:angrist98intro} without interactions. This hope generally fails:
 Quasi-experimental interpretation of $\btau$ necessitates that \emph {both} $\pi_i^*$ and
 $\pi_i^* x_ {1i}$ be linear in $x_ {2i}$. When this fails, some contrast $\tau_0 +
\tau_1'x_{1}$ does not satisfy level independence.\footnote{This result was novel at the
 time of a working paper draft of this paper (arXiv:2407.21119v2, January 13, 2025);
 concurrent and independent work by \citet{zhao2025interacted}
 (arXiv:2502.00251, February 1, 2025) provides a similar result, whose necessity
 direction relies on \cref{prop:forbidden_interactions}.}

\begin{restatable}{prop}{propforbidden}
\label{prop:forbidden_interactions}
 Consider the specification \eqref{eq:interaction} and let $\btau =
(\tau_0, \tau_1')'$ be the coefficients of interest. Then the corresponding implicit
design exists if and only if, for some conformable matrices $(\Gamma_0, \Gamma_1)$ and
all $i$, 
$(\delta_0 + \delta_1'x_ {2i}) x_ {1i} =
\Gamma_0 + \Gamma_1 x_{2i}$, where $\delta_0, \delta_1$ are population projection
coefficients of $\pi_i^*$ on $x_ {2i}$. When this happens, the unique implicit design is
$\pi_i =
\delta_0
+
\delta_1' x_ {i2}$. Therefore, if $\btau$ satisfies \cref{item:M}, 
 then $\pi_i^* = \delta_0 + \delta_1'x_{2i}$ and $\pi_i^*x_{1i} = \Gamma_0 +
\Gamma_1 x_{2i}$. 
\end{restatable}

The necessary condition for interpreting $\btau$ as minimally quasi-experimental is that \emph
{both} the propensity score $\pi_i^*$ and its interaction with the covariates $\pi_i^*
x_ {1i}$ are linear functions of $x_{2i}$. When $x_{1i}$ is included in $x_{2i}$, this
condition is unlikely to hold in general, as $\pi_i^* x_{1i}$ would involve nonlinear
transformations of $x_ {1i}$ and thus cannot be linear. This condition does hold if
$\pi_i^*$ is constant or if $x_{1i}$ represents a saturated categorical variable and $x_
{2i}$ contains all other covariates interacted with $x_ {1i}$.\footnote{That is, $x_
{1i}$ contains mutually exclusive binary random variables, and $x_{2i}$ contains $x_
{1i}$, some set of other covariates $x_{3i}$, and all interactions $x_{3ik}x_{1i\ell}$.}

Why can we not interpret $\tau_0 + \tau_1'x_ {1i}$ as a linear approximation of the
conditional average treatment effect? One could think of
\eqref{eq:interaction} as two regressions, one on the treated $W=1$ and one on the
untreated $W=0$. Both regressions are indeed best linear approximations to $\E [Y(1) \mid
x, W=1]$ and $\E[Y(0) \mid x, W=0]$, which are equal to the mean potential outcomes $\E [Y
(1) \mid x],
\E[Y(0) \mid x]$ under unconfoundedness. The contrast $\tau_0 + \tau_1'x_1$ is then the
 difference of the fitted values of these two regressions. However, the two regressions
 are best linear approximations with respect to  {different distributions} of
 the covariates. Thus, their difference is not a best
 linear approximation to the conditional average treatment effect. Shifting $Y(1)$ and $Y(0)$ by the same arbitrary amount
 causes asymmetric behavior in the two regressions.

When $x_{1i} = x_{2i} = x_i$, this result supplements \cref{thm:zoo}(\ref
{item:imbens-wooldridge})--(\ref{item:kline}). One could observe that \cref{thm:zoo}(\ref{item:kline}) is asymmetric: Flipping treatment and control yields an exact
quasi-experimental interpretation for the ATU where the implicit design is that the 
\emph{inverse} odds is linear in $x_i$. \Cref{prop:forbidden_interactions} then shows that
 because different contrasts necessitate different and incompatible designs like with
 ATT/ATU, requiring all of them to be causal rejects \cref{item:M1}.\footnote{\Cref
 {sub:interact_more} shows formally that requiring
\cref{item:M1} for the contrast $\tau_\lambda = \lambda_0 \tau_0 + \lambda_1'\tau_1$ in
this regression implies implicit designs $\pi_\lambda$, generally fractional-linear in
$x_i$, that depend on the contrast $\lambda_0, \lambda_1$.

In particular, these results are relevant for the ATE contrast $\tau_0+\tau_1'\bar x$,
which is separately studied in Theorem 1 in \citet{chattopadhyay2023implied}. \citet
{chattopadhyay2023implied} show that if we insist that $\tau_0+\tau_1'\bar x$ equal the
ATE, then we need both propensity odds and reciprocal odds to be linear. In contrast, we
show that if $\tau_0+\tau_1'\bar x$ is only required to be some treatment effect contrast
(not necessarily the ATE), the implicit design exists but is fractional-linear $\pi_i =
 \frac{\theta_{0} + \theta_1'(x-\bar x)}{1-\Gamma_2'(x-\bar x)}$. However, the
  requirement \cref{item:M2} that $\pi_i = \pi_i^*$ then imposes additional (unpleasant)
  restrictions on
  $
  (\theta_0, \theta_1,
 \Gamma_2)$, formalized in \cref{rmk:ate_wrong}.} Without saturated covariates, \cref{thm:zoo}(\ref{item:kline}) shows that
  particular contrasts (e.g., the ATT) maintain causal interpretation, at the expense of
  others.

Taken together, interacted regressions are \emph{less} robust in terms of \cref{item:M}
than the simple regression \eqref{eq:angrist98intro}, contrasting with intuition in 
\citet{lin2013agnostic} and \citet{negi2021revisiting}. The uninteracted
regression introduces variance weighting for the estimand, but maintains validity under a
simple design. The interacted regression removes the weighting when covariates are
saturated but loses exact quasi-experimental interpretation in general.

 Is there a simple regression that targets the ATE under linear propensity scores?
 Unfortunately, \cref{prop:no-ate-regression} shows that the answer is \emph{no}, at
 least not with specifications that are linear in $[1, x_i, W_i, W_ix_i]$. As a result,
 targeting the ATE under the same implicit design as
 \eqref{eq:angrist98intro} necessitates moving beyond regression estimators.

 \subsection{The fragility of quasi-experimental TWFE}
\label{sub:panel}

\Cref{thm:zoo}(\cref{item:athey-imbens}) shows TWFE is minimally quasi-experimental only
 under totally randomized treatment paths. We now show that adding time-varying covariates
 often destroys even that.

\begin{restatable}{prop}{proptimevarying}
\label{prop:time_varying_cov}
Assume $\cW \subset \br{0,1}^T$. Consider the regression $\bY_{it} =
\alpha_i + \gamma_t + \tau \bW_{it} + \delta'\bx_{it}$ where $\tau$ is the coefficient of
 interest. Let $\beta_{w\to x}$ be the population projection coefficient of $W_{it}$ on
 $x_{it}$ under $\bpi^*$, with individual and time fixed effects. If an implicit design
 exists, then, for $x_i \in \R^{T \times \dim(x_{it})}$ that stacks the covariates $x_
 {it}$,
\[\pr{\bx_i - \frac{1}{n}\sum_{j=1}^n \bx_j}\beta_{w\to x} \in \Span(\cW \cup \br{1_T})\,
\text{ for all $i=1,\ldots,n$}.\] When $\beta_{w\to x} = 0$, if $\cW$ contains linearly
independent vectors whose span excludes $1_T$, the implicit design is uniquely equal to $\pi_i (\bw) =
\frac{1}{n} \sum_{i=1}^n \pi_i^*(\bw)$ for all $i$ as in \cref{thm:zoo}(\ref{item:athey-imbens}).

\end{restatable}

An implicit design exists only if a linear combination of demeaned covariates lies in the
span of $\cW$ and $1_T$ for every unit. This condition arises because we essentially need
that the mean treatment $\E[\bW_i] = \sum_{\bw \in \cW}
\pi_i^*
(\bw) \bw$ is \emph{exactly} described by two-way fixed effects with time-varying
covariates, analogous to the intuition for \eqref{eq:angrist98intro}. This then restricts
the space of covariates, since they need to generate vectors that lie in the linear span
of $\cW$.

With staggered adoption, $\Span\pr{\cW \cup
\br{1_T}}$ is the subspace of vectors that are piecewise constant between adjacent
adoption dates. This subspace is highly restrictive if there are relatively few adoption
dates. If $\beta_ {w\to x} \neq 0$, it is thus knife-edge that $
\pr{\bx_i - \frac{1} {n}\sum_ {j=1}^n \bx_j}
\beta_{w\to x}$ happens to be located in that subspace, unless columns of $\bx_i$ happen
 to be piecewise constant over $t$ as well.\footnote{This \emph{is} plausible if the
 time-varying covariates are interactions of fixed covariates with the time fixed
 effects ($\bx_ {it}'\delta = x_i'\delta_t$). \Cref{prop:unit-covariate-interact} shows
 that for this specification, causal interpretation is possible necessarily under linear
 generalized propensity scores $\pi_i(\bw) = \delta_0(\bw) +
 \delta_1(\bw)'x_i$.} On the other hand, if $\beta_ {w \to x}$ is
 zero under $\bpi^*$, including the covariates makes no difference to the coefficient on
 $W_{it}$. Thus TWFE with time-varying covariates rarely retains a
 quasi-experimental interpretation. Researchers using such a specification either believe
 that the covariates do not affect treatment assignment and are irrelevant for
 identification, or they are embedding outcome modeling assumptions.

Finally, a similar fragility afflicts regressions with imbalanced panels. Such a
regression only has a quasi-experimental interpretation when the missingness pattern is
uncorrelated with the treatment assignment pattern, in which case the design must  again
be total randomization of treatment paths. We detail this result in \cref{asub:imbalance}.

\section{Extension: Two-stage least-squares}
\label{sec:extension_two_stage_least_squares}

Similar ideas to \cref{thm:main} extend to two-stage least-squares (TSLS): We can use
level irrelevance to recover some design---now a distribution of the \emph
{instrument} $W_i$---under which TSLS estimands have a causal interpretation in the
instrument $W$ (as weighted averages of intent-to-treat effects). Interestingly, the
implicit estimand also provides necessary conditions on \emph{compliance behavior} for
TSLS to estimate properly weighted causal effects in terms of the endogenous treatment.

For instance, examining the implicit estimand for a binary treatment, binary instrument
TSLS regression recovers (strong) \emph{monotonicity} as a necessary condition \citep
{imbensangrist,sloczynski2020should}. Doing so for TSLS with multiple treatments yields a
compliance restriction in \citet{bhuller20242sls}. These results are recovered simply by
enumerating which compliance types for each unit are consistent with the implicit
estimand assigning proper weights to said unit's potential outcomes in the treatment.

Consider the following TSLS specification of a scalar outcome on a covariate transform
\[
    Y_i = t(D_i, x_i)'\beta + \epsilon_i,
\]
instrumenting $t(D_i, x_i)$ with $z(W_i, x_i)$. Here, $D_i = d^*_i(W_i) \in \cD$ is the
endogenous treatment, $d^*_i(\cdot)$ is the \emph{compliance type}  for unit $i$, and $t
(\cdot,
\cdot), z (\cdot, \cdot)$ are again known transforms. Assume the exclusion restriction
 holds so that $y_i(d_i^* (w), w) = y_i(d_i^*(w))$. In this notation, a binary treatment,
 binary instrument TSLS regression can be represented by $t(D_i, x_i) = [1, D_i]'$ and
 $z (W_i, x_i) = [1, W_i]'$.

We extend steps \cref{item:i}--\cref{item:iv} in \cref{sub:leading}. For \cref{item:i},
define the TSLS estimand  $\btau = \Lambda\beta$ as \[
    \btau = \Lambda \pr{G_{tz} G_{zz}^{-1} G_{zt}}^{-1} \pr{G_{tz} G_{zz}^{-1} \frac{1}
    {n} \sum_{i=1}^n \E_{W_i\sim \pi_i^*}[z(W_i, x_i) y_i(W_i)]},
\]
where $ G_{tz} \equiv \frac{1}{n}\sum_{i=1}^n \E_{W_i \sim \pi_i^*}\bk{ t(d_i^*(W_i), x_i)
z(W_i, x_i)' } = G_{zt}'$ and $G_{zz} \allowbreak  \equiv \allowbreak
\frac{1}
{n}\sum_ {i=1}^n
\allowbreak\E\bk{ z(W_i, x_i) z(W_i, x_i)'} $. This representation
 simply replaces all averages in the TSLS estimator with expectations over the instrument
 $W_i$. Let $H_n \equiv \pr{G_{tz} G_{zz}^{-1} G_{zt}}^{-1} G_{tz} G_{zz}^{-1}$ be
 the analogue of the inverse Gram matrix $G_n^{-1}$.\footnote{Indeed, if $t(d_i(W_i), x_i)
 = z (W_i, x_i)$ so that the TSLS specification is equivalent to OLS, then $H_n$ is
 exactly the inverse Gram matrix.} Like $G_n^{-1}$, $H_n$ is known in population and
 consistently estimable in sample. Thus, we similarly treat $H_n$ as known.

Next, for \cref{item:ii}, write $\btau$ in the form of a linear contrast:
\[
    \btau = \frac{1}{n} \sum_{i=1}^n \sum_{w \in \cW} \pi_i^*(w)\Lambda H_n z(w, x_i) y_i
    (w) \equiv \frac{1}{n} \sum_{i=1}^n \sum_{w \in \cW}  \pi_i^*(w) \underbrace{\brho_i
    (w)}_{k \times 1} y_i(d_i^*(w)).
\]
We define \emph{potential weights} analogously by $\brho_i(w) \equiv \Lambda H_n z(w,
x_i)$. For \cref{item:iii}--\cref{item:iv}, the requirement that $\btau$ is minimally
quasi-experimental continues to be reasonable.
Maintaining this restriction again yields
\eqref{eq:pop_level_irrelevance_condition} for $\pi_i^*(w)$, whose solutions we
continue to call \emph{implicit designs}---they are again plausible candidates for the
true design $\pi_i^*(\cdot)$ in the sense of \cref{item:M1}.

Just-identified TSLS specifications have enough equations\footnote{For a TSLS
specification to be non-collinear, an instrument that takes $J+1$ values can support
$k\le J$ endogenous coefficients of interest. Since $\pi_i(\cdot)$ is a $J$-dimensional
unknown vector, we need $k \ge J$ restrictions to have a unique implicit design. } to pin
down an implicit design $\pi_i (\cdot)$. If there are more distinct instrument values
than coefficients of interest, then we may have too few restrictions on $\pi_i^*(w)$ from
$\btau$ alone. However, it may be reasonable to also impose level irrelevance for certain
first-stage coefficients, which would add more restrictions to recover a unique implicit
design.

The estimand for TSLS depends on units' unknown compliance types $d_i^*
(\cdot)$. Therefore, interpreting the estimand as a causal effect of the treatment $d$
requires restricting compliance patterns. This can be operationalized as follows. Given
an implicit design $\pi_i (\cdot)$, the corresponding implicit estimand $\btau$ can be
written as a weighted sum of individual potential outcomes, which can be grouped into
treatment conditions: \begin{align*}
\btau &=  \frac{1}
{n}\sum_{i=1}^n \sum_{w \in \cW} \bomega_i(w; \bpi) y_i(d_i^*(w)) \quad \text{  where  }
\bomega_i (w;
\bpi) \equiv
\pi_i(w) \brho_i(w) \numberthis 
\label{eq:iv-implicit-estimand}
\\
&= \frac1n \sum_{i=1}^n \sum_{k \in \cD} \underbrace{\pr{\sum_{w: d_i^*
(w)=k} \bomega_i(w; \bpi)}}_{ \equiv \bomega^*_i
(k; \bpi, d_i^*)} y_i(k) \equiv \frac1n \sum_{i=1}^n \sum_{k\in \cD}
\bomega^*_i
(k; \bpi, d_i^*) y_i(k). \numberthis \label{eq:iv-implicit-estimand-grouped}
\end{align*}
\eqref{eq:iv-implicit-estimand} represents the estimand as an aggregation of  $w$-on-$y$
 causal effects.
\eqref{eq:iv-implicit-estimand-grouped} then groups together $w$ values that lead to the
 same $d_i^*(w) = k$, thereby translating \eqref{eq:iv-implicit-estimand} to $d$-on-$y$
 effects. In
\eqref{eq:iv-implicit-estimand-grouped}, the weight on the $k$\th{} treatment is
 $\bomega^*_i(k; \bpi, d_i^*) \equiv \sum_{w: d_i^*(w)=k} \bomega_i(w; \bpi)$, which is
 known given $d_i^* (\cdot)$. If $\btau$ were to have a causal interpretation, we can then
 \emph{enumerate} all compliance types $d_i^*$ for each unit and check which ones lead to
  weights $\bomega^*_i (k; \bpi, d_i^*)$ that are consistent with the causal
  interpretation.

To illustrate, consider a particular class of TSLS specifications: For $x_i$ that includes
a constant, consider a specification with $J+1$ values of an unordered
treatment $\mathcal D = \br{0,\ldots, J}$
\begin{align*}
t(d, x_i) &= [\one(d = 1),\ldots, \one(d=J), x_i']' \\
z(w, x_i) &= [\one(w = 1),\ldots, \one(w=J), x_i']'. \numberthis \label{eq:tsls-setup}
\end{align*}
In this TSLS specification, the coefficients of interest are $\btau = (\tau_1,\ldots,
\tau_{J})'$, where $\tau_k$ is the coefficient on $\one(d=k)$, meant to capture the causal
effect of $d=k$ relative to $d=0$.

Examining entries in \eqref{eq:iv-implicit-estimand-grouped}, we have
\[
    \tau_k = \frac{1}{n} \sum_{i=1}^n  \sum_{k' = 1}^{J} \omega_{i}^{(k,k')}(d_i^*) (y_i(k') -
    y_i(0))
\quad \text{ for  } \quad 
    \omega_i^{(k, k')} \equiv (\bomega^*_{i}(k'; \bpi, d_i^*))_k. \numberthis 
    \label{eq:estimand_tsls}
\] If $\tau_k$ is to be interpreted as a causal effect of $d=k$ relative to $d=0$, then we
 should at least restrict $\omega_i^{(k,k)} \ge 0$ and $\omega_i^{(k,k')} = 0$ for
 $k \neq k'$. If so, $\tau_k$ equals a convex aggregation of $y_i(k) - y_i
 (0)$ that is not contaminated by treatment effects of some other arm   $y_i(\ell) - y_i
 (0)$. If this is true and if $\bpi^* = \bpi$, following \citet{bhuller20242sls}, we say
 that TSLS
\emph{assigns proper weights}.\footnote{When $J=2$, $\btau$ is minimally quasi-experimental and
assigns proper weights if and only if it is weakly causal in the sense of 
\citet{blandhol2022tsls}.}

Given $\omega_i^{(k,k')}(\cdot)$, for each unit, we can then enumerate all compliance
types $d_i(\cdot)$ and retain those consistent with proper weights. Analogous to implicit
designs, we refer to each element of the following set as an \emph{implicit compliance
profile}: For $\mathcal D = \br{0,\ldots, J}$,\[
    \br{ (d_1(\cdot),\ldots, d_n(\cdot)): \text{ for all $i$, $k\neq k' \in \mathcal D$},
        \omega_i^ {(k, k)} (d_i) \ge 0
        \text{ and  } \omega_i^{(k, k')}(d_i) = 0
    } \numberthis \label{eq:implicit_compliance_profile}.
\]
The following proposition summarizes these results:

\begin{restatable}{prop}{propivmain}
\label{iv:main}
    In TSLS, $\btau$ is minimally quasi-experimental if and only if
    \begin{enumerate}
        \item An implicit design $\bpi$ exists
        \item $\bpi^* = \bpi$.
    \end{enumerate} Additionally, $\btau$ from the specification \eqref
     {eq:tsls-setup} assigns proper weights if and only if the following holds for the
     implicit estimand under $\bpi$:
    \begin{enumerate}[resume]
        \item An implicit compliance profile $d_1(\cdot), \ldots, d_n(\cdot)$ in
        \eqref{eq:implicit_compliance_profile} exists

        \item Some implicit compliance profile $d_1(\cdot), \ldots, d_n(\cdot)$
        is equal to $d_i^*(\cdot),\ldots, d_i^*(\cdot)$.
    \end{enumerate}
\end{restatable}

Like \cref{thm:main}, \cref{iv:main} separates requirements for causal interpretation into
objective and subjective components. We can directly compute items (1) and
(3) since the potential weights, implicit design, and implicit estimand are known in the
population. Results from this computation are plausible candidates for items (2) and
(4)---if no such candidate is found, then causal interpretation must be rejected.

These computations are informative. To illustrate, simply computing the implicit design
and compliance profiles recovers necessary conditions for several recent results in the
instrumental variables literature. To introduce these, we first give terminology to
compliance
patterns.

\begin{defn}[Compliance restrictions]
     \begin{itemize}
         \item With $J+1=2$, we say that a profile $d_1(\cdot), \ldots,
         d_n
          (\cdot)$ satisfies \emph{strong monotonicity} if either $d_i(1) \ge d_i(0)$ for
          all $i$ or $d_i(1)
\le d_i (0)$ for all $i$. 

\item With $J+1>2$, for $k = 1,\ldots, J$, we say that $d(\cdot)$ is a $k$-always taker
if
$d (\cdot) = k$; it is a $k$-never taker if $d(\cdot) \neq k$; otherwise we say $d(\cdot)$
is a $k$-complier. We say $d(\cdot)$ is a \emph{full complier} if it is a $k$-complier for
all $k$. 

\item We say that a compliance profile $ (d_1 (\cdot),\ldots, d_n(\cdot))$ satisfies \emph
 {common compliance} if for any $k = 1,\ldots, J$ and any two $k$-compliers $d_i
 (\cdot), d_j(\cdot)$, we have $d_i(w) = k$ iff $d_j(w) = k$.  

\item We say that a compliance profile satisfies \emph{extended monotonicity} if there
 exists some permutation $f (\cdot)$ of the instrument values $\br{0,\ldots, J}$ such
 that, for all $i$, either (i) for all $w$, $d_i (f (w))
        \in \br{0,w}$ or (ii) $d_i(\cdot)$ is constant.\footnote{For three instrument
         values, up to permutation of the instruments, extended monotonicity limits $d
         (\cdot)$ to one of six types $(d(0), d(1), d(2)) \in \br{(000), (111),(222),
         (010),(002), (012)}$---for, respectively, never-taker, 1-always-taker,
         2-always-taker, 1-complier, 2-complier, or full complier \citep
         {bhuller20242sls}.

         This condition is a generalization of Assumption 3 in
         \citet{behaghel2013robustness}, who call this assumption \emph {extended
         monotonicity}. Indeed, the condition is equivalent to that, for all $i$, $w\neq
         0$ and $w', w'' \neq w$, $
             \one\br{d_i(f(w)) = w} \ge \one\br{d_i(f(w')) = w} = \one\br{d_i(f(w'')) = w}.
         $}
     \end{itemize}
 \end{defn} 

\begin{restatable}{prop}{proptsls}
    
\label{item:tsls} Consider the TSLS specification in \eqref{eq:tsls-setup}.
\begin{enumerate}
    \item The unique implicit design satisfies $\pi_i(j) =
    x_i'\delta_j$
    for
    $\delta_j = \pr{
   \sum_{i=1}^n
    x_ix_i'}^{-1} \sum_{i=1}^n
    \pi_i^*(j) x_i$.
    \item When $\bpi^* = \bpi$, the implicit compliance profiles relative to the implicit
    design satisfy:
    \begin{enumerate}
        \item When $J+1=2$, all implicit
        compliance profiles satisfy strong monotonicity.
        \item When $J+1>2$ and $x_i$ is a
        constant, all implicit compliance profiles satisfy {common
         compliance}; all implicit compliance profiles containing a full complier
         satisfy extended monotonicity.
    \end{enumerate}
\end{enumerate}
\end{restatable}

\Cref{item:tsls} recovers several results for TSLS.
\Cref{item:tsls}(1) and (2)(a) recover the necessary direction for Corollary 3.4 in
\citet{sloczynski2020should} and Theorem 1 in \citet{blandhol2022tsls}: With binary
treatment, monotonicity is required for interpreting the TSLS coefficient causally, in the
sense that it assigns proper weights.\footnote {Theorem 1 in \citet{blandhol2022tsls}
imposes exogeneity and monotonicity and finds that $\btau$ is minimally quasi-experimental
and has proper weights if and only if $\pi_i^*$ is linear. \Cref {item:tsls}(1) and (2)(a)
in turn show that if $\btau$ is minimally quasi-experimental and has proper weights, then
linear propensity scores \emph{and monotonicity} are satisfied (see
\cref{sub:relation_to_blandhol} for details). Likewise, Corollary 3.4 in
\citet{sloczynski2020should} shows that strong monotonicity implies proper weights, but
not the converse. } Without covariates, this is a converse to \citet {imbensangrist}.

\Cref{item:tsls}(2)(b) recovers---and \emph{corrects}---Propositions 5 and B.1 in \citet
 {bhuller20242sls}. Proposition 5 in \citet{bhuller20242sls}  claims that if TSLS assigns
 proper weights, then compliance satisfies extended monotonicity. Unfortunately, just
 assuming TSLS assigns proper weights does not suffice for this conclusion (see
 \cref{sub:relation_to_blandhol} for a counterexample).\footnote{We are grateful to Henrik
 Sigstad for discussion. } Instead, the essence of their argument implies that compliance
 profiles satisfy common compliance; their conclusion in turn stands if it is known that
 some full complier exists. Both implications are captured by \cref{item:tsls}(2)(b).

\section{Empirical illustrations}
\label{sec:empirical}

\subsection{An AI-assisted census of AER and AEJ:AE}

Our framework is motivated by the impression that applied researchers predominantly
estimate causal effects by regression, justified quasi-experimentally by arguing that
treatment is as-good-as-randomly assigned. To assess the pervasiveness of this practice,
we conduct a \emph{census} of papers published in the \emph{American Economic Review} and
\emph {American Economic Journal: Applied Economics} from January 2020 to June 2026,
totaling 1,051 papers. Specifically, an AI agent (Claude Haiku 4.5 or Sonnet 4.6, across
different runs) reads the full main text of each paper and answers questions according to
a rubric about that paper's empirical methodology. Of course, this exercise is an
automated review, and thus should be interpreted as a rough but informative snapshot of
empirical practice.\footnote{See the replication repository
\href{https://github.com/jiafengkevinchen/potential-weights-public}
{github.com/jiafengkevinchen/potential-weights-public} for details of this exercise,
including prompts and per-paper results.}

\begin{figure}[htb]
    \centering
    \includegraphics[width=\textwidth]{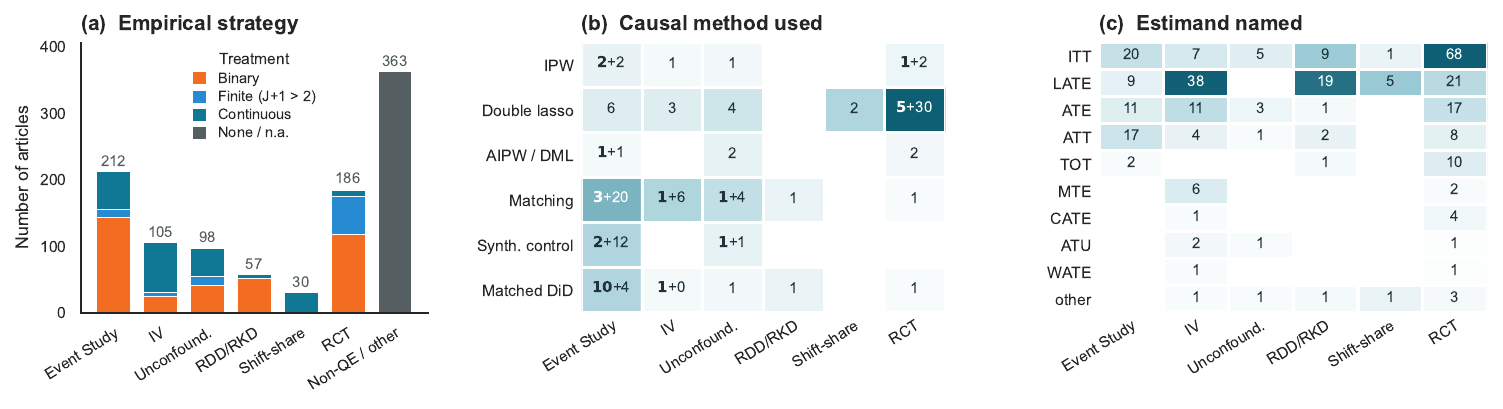}

\begin{proof}[Notes]
    Panel B returns AI classification for whether the listed method is used in a paper. If
    the AI agent judges that the method is used in a headline result, then the
    corresponding counts are in bold.
\end{proof}

    \caption{AI classification of empirical strategies and settings in AER and AEJ:AE 
    (January
    2020--June 2026)}
    \label{fig:methods}
\end{figure}

We task AI agents to classify each paper's empirical strategy according to standard
quasi-experimental categories and record (i) whether the paper uses non-regression methods
in causal inference  and (ii) whether the paper names an explicit estimand (\cref
{fig:methods}). The vast majority of quasi-experimental papers do not implement causal
inference methods outside of regression: Among papers arguing unconfoundedness or
exogenously assigned IV, AI agents report that only four use (augmented)
inverse-propensity weighting \citep{hirano2003efficient} or double/debiased machine
learning \citep{chernozhukov2018double}; 12 papers feature (propensity-score) matching,
but only 2 use matching in their headline results. Likewise, practitioners are also not
explicit about their target estimand: Only 9 out of 98 papers classified under
unconfoundedness report a target estimand; 48 IV papers do, but 38 of them mention LATE,
perhaps inappropriately so per the results of \citet{blandhol2022tsls}. Essentially no
paper names weighted average effects, even though those are typically what regression or
IV estimates under heterogeneity.

Our framework is primarily applicable to studies that argue for quasi-experimental
assignment of a treatment (Unconfound.) or an instrument (IV) with a finite-valued
treatment, which corresponds to 85 papers, or 18\% of all quasi-experimental papers. This
calculation excludes event studies, but 46\% of those are actually classified to contain
language that suggests treatment-based identification, for which the analysis of
\cref{thm:zoo}(\cref{item:athey-imbens}) applies.

Furthermore, the negative results on interacted regressions
and difference-in-differences 
(\cref{prop:forbidden_interactions,prop:time_varying_cov,prop:unbalanced}) would
immediately
 reject exact quasi-experimental interpretation for a sizable proportion. Among papers
 classified as using a quasi-experimentally assigned binary treatment/instrument, 90\%
 report some heterogeneity analysis. 42\% run the main specification separately by
 covariate category, but 48\% solely interact the treatment as in \eqref
 {eq:interaction}, possibly subject to the negative result of 
\cref{prop:forbidden_interactions}. Among the event study papers, at least 40\% are
 classified to include time-varying covariates and 36\% are classified to involve
 unbalanced panels. The results in \cref{prop:unbalanced,prop:time_varying_cov} show that
 these specifications are quasi-experimental essentially only under total random
 assignment of treatment, and implicit designs may frequently not exist.

\subsection{Diagnostics for  the implicit design}

We next turn to a battery of nine papers that estimate \eqref{eq:angrist98intro} like 
\citet{blakeslee2020way}\footnote{\citet{ang2023birth,dippel2021leadership} are IV
 applications, for which we take the reduced form specification. 

\citet{blakeslee2020way} report weighted least squares estimates but
  state that their results are robust to the weighting. We thus report OLS versions of
  their exercises while replicating their WLS estimates.

Some covariates used by \citet{ang2023birth} are not included in the replication data. We
     drop these covariates from the specification. Doing so changes the coefficient from
     0.146 (0.026) to 0.179.

     In \citet{cage2023heroes}, treatment can be fractional for 4,599 of 34,947
     observations. We keep the treatment binary by dropping these observations. } and
     illustrate our diagnostics and recommendations in the context of these papers. We
     start by assessing their implicit designs according to \cref{thm:main,thm:aipw_new}.
     We estimate implicit designs by computing
\eqref{eq:implicit_design_formula_binary}, substituting in the sample Gram matrix $
\frac{1}{n} \sum_{i=1}^n z(x_i, W_i)z(x_i, W_i)'$ for $G_n$. For the simple
 specification \eqref{eq:angrist98intro}, this is simply the sample OLS regression of the
 realized $W$ on covariates. For each estimated implicit design $\hat\bpi$, we then test
 whether it equals some implicit design that satisfies \cref{item:M1} and \cref
 {item:M2}. 

\begin{table}[htb]
    \caption{Diagnostics for the implicit design}
    \label{tab:diagnostics_implicit_design}
    \centering
    \small
    \begin{tabular}{lrrrrr}
\toprule
Paper & $N$ & \% $\hat\pi \notin [0,1]$ & \% $\hat\pi \in [0.1,0.9]$ & RESET $p$ & Permutation $p$ \\
\midrule
Xu (2018) & 3,510 & 33.4 & 28.6 & $<0.001$ & -- \\
Blakeslee et al.\ (2020) & 785 & 4.2 & 88.7 & $<0.001$ & -- \\
Dippel and Heblich (2021) & 261 & 23.0 & 41.4 & $<0.001$ & 0.005 \\
Herrnstadt et al.\ (2021) & 41,720 & 0.2 & 99.6 & $<0.001$ & -- \\
Ang (2023) & 2,998 & 7.5 & 78.5 & $<0.001$ & -- \\
Cag\'e et al.\ (2023) & 30,342 & 39.5 & 18.3 & $<0.001$ & -- \\
Covert and Sweeney (2023) & 1,297 & 25.6 & 49.1 & $<0.001$ & -- \\
Cervellati et al.\ (2024) & 1,638 & 1.0 & 91.2 & 0.008 & 0.064 \\
Chambru et al.\ (2024) & 3,211 & 6.8 & 83.9 & $<0.001$ & 0.030 \\
\bottomrule
\end{tabular}

    \begin{proof}[Notes] 
     The RESET test regresses the     treatment indicator on a quartic polynomial of
     $\hat\pi$ and reports a Wald test of the coefficients on nonlinear terms, clustering
     at whatever level the target paper's specification clusters.
    \end{proof}
\end{table}

\Cref{tab:diagnostics_implicit_design} conducts this exercise. All nine papers have
 estimated implicit designs that feature $\hat\pi_i \not\in [0,1]$, with a median of
 7.5\%. While possibly due to estimation noise, having even estimated propensity scores
 that are out-of-bounds is not ideal. Following \citet{blandhol2022tsls}, we can further
 assess whether $\bpi^*$ is actually linear in the covariates by conducting a Ramsey RESET
 test of linearity, assuming treatments are independently drawn. This test decisively
 rejects in all nine applications. Finally, three applications either have permutation
 tests or explicitly known designs such that the joint distribution of the treatments
 $(W_1,\ldots, W_n) \sim \Pi^*$ is known. We can then compare the realized discrepancy
 $\frac{1}{n}\sum_{i=1}^n (\hat\pi_i - \pi_i^*)^2$ to its null distribution under $\Pi^*$
 and compute a $p$-value for \cref{item:M2}. This exercise rejects \cref{item:M2} in  two
 of the three applications.

Overall, \cref{tab:diagnostics_implicit_design} shows that failure of exact
quasi-experimental interpretation---in the sense of \cref{thm:main}---is pervasive. While
all nine papers justify causal identification through exogenously assigned treatments,
they implicitly estimate poor models of treatment. Interpreted strictly, the optimism of
\citet{angrist2008mostly} is misplaced: These common estimators do not have defensible
quasi-experimental interpretations that are agnostic of potential outcomes.

We next assess whether the specifications can be said to be approximately
quasi-experimental under some posited design $\tilde \bpi$, following the decomposition in
\cref{thm:aipw_new}(3). We consider a host of different candidate $\tilde \bpi$'s (see
\cref{asec:emp} for details):
\begin{enumerate}[wide, label=(\roman*)]
    \item (Recalibrated) We recalibrate $\hat\bpi$ by binning it and compute the empirical frequency of treatment within
     each bin, following chapter 17 of \citet{imbens2015causal}.

     \item (Logit) We run a logistic regression of $W$ on the same covariates as in
     $\hat\bpi$. 
     \item (ML) For sample sizes above 1000, we use gradient boosting (lightGBM) to
      predict $W$ on the covariates and $\hat\bpi$. 
     \item (Paper) For the three applications with reported designs, we take $\tilde \bpi$
      as the known design. In addition, \citet{covert2023relinquishing} include a
      random-forest estimated propensity score, which we also treat as a candidate
      $\tilde \bpi$.
\end{enumerate}

\begin{figure}[htb]
    \centering
    \includegraphics[width=\textwidth]{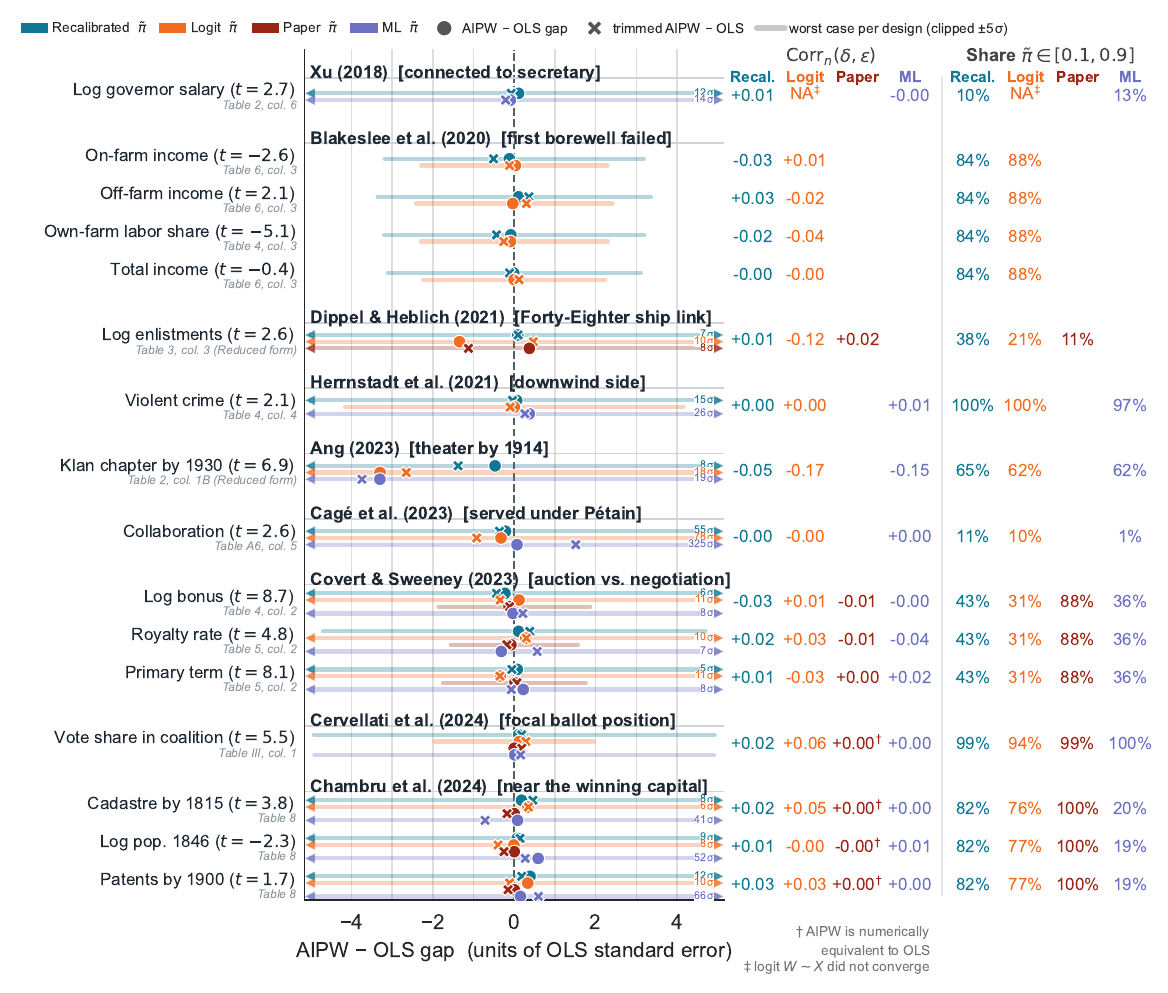}
    \vspace{-2em}

\begin{proof}[Notes]
    For each specification in the paper and each $\tilde \bpi$, we compute two AIPW
    estimates \eqref{eq:aipw}. One weights by $\omega_i(w) \propto \tilde\pi_i
    (1-\tilde\pi_i)$ ($\bullet$) and one weights by $\omega_i(w) \propto \one(\tilde\pi_i
    \in [.1,.9])$ ($\times$). We show the difference of these estimates against the OLS
    $\hat\tau$, normalized by the standard errors of the OLS estimate. We also show the
    possible range of this difference, following
\eqref{eq:worst-case} (truncated to 5 OLS SEs, and labeling those truncated ranges). 
Following the decomposition \eqref{eq:decomp_aipw_corr}, we print the realized $\corr_n
(\delta, \epsilon)$. For the trimmed AIPW estimator, we also print the share of
observations that survive the trimming. 
     \end{proof}

    \caption{Comparison of AIPW estimate (\cref{thm:aipw_new}(3)) and trimmed AIPW estimate
    to OLS estimate}
    \label{fig:forest}

\end{figure}

The decomposition \Cref{thm:aipw_new}(3) produces an informative decomposition of the
AIPW-OLS gap as a covariance between discrepancies on the designs and the regression
residual, which we can further decompose into the product of a correlation and two
standard deviations: For $\epsilon = (\epsilon_1,\ldots, \epsilon_n)'$ the sample
regression residuals, \[
    \tilde\tau_{\mathrm{AIPW}} - \hat\tau_{\mathrm{OLS}} = \mathrm{SD}_n(\delta) \times 
    \mathrm{SD}_n(\epsilon) \times \corr_{n} (\delta, \epsilon) \quad \delta_i \equiv \frac{-
    (\tilde\pi_i -
    \hat \pi_i)} {\frac{1}{n} \sum_{i=1}^n \tilde \pi_i (1-\tilde\pi_i)} \numberthis 
    \label{eq:decomp_aipw_corr}
\] The size of the correlation term is an empirical question. If outcomes were
 adversarially chosen, it could be as large as the correlation of $\delta$ with its
 residuals when projected onto $W$ and $x$:\footnote{The decomposition 
 \eqref{eq:decomp_aipw_corr} also holds if we replace $\pi_i$ with $\hat\pi_i$ since both
  are orthogonal to the residuals $\epsilon$. The worst case in \eqref
  {eq:worst-case} is achieved if the potential outcomes are proportional to
  $\delta_i$ and treatment effect is zero: $y_i(1) = y_i (0) = t \delta_i$, where $t$ is
  such that the regression residual has standard deviation matching $\mathrm{SD}_n
  (\epsilon)$.} \[
     \abs{\corr_{n} (\delta, \epsilon)} \le \abs{\corr_n(\delta, \delta - \hat\gamma_1 W -
     \hat\gamma_2'x)} \numberthis \label{eq:worst-case}.
 \]

\Cref{fig:forest} assesses this decomposition across the nine papers. The dots $\bullet$
display  the AIPW-OLS gap for each specification and each $\tilde \bpi$ considered, in
units of the standard error of the OLS estimate. The possible range of AIPW-OLS gaps,
corresponding to the worst case of \eqref{eq:decomp_aipw_corr}, is displayed in a band. We
also print the realized correlation in \eqref{eq:decomp_aipw_corr} on the side.

Overall, the vast majority of the 50 specification-by-$\tilde\bpi$ scenarios do not
feature AIPW-OLS gaps that are much larger than statistical uncertainty---and OLS has an
approximate quasi-experimental interpretation in that sense under $\tilde \bpi$. Among the
50, only 3 change the OLS estimate by more than one standard error. In terms of the original reported units, 9
scenarios change the OLS estimate by more than 10\%. Notably, the four papers that report
an explicit design are especially robust for the reported design. AIPW is even numerically
equivalent to OLS in \citet {cervellati2024random} and
\citet{chambru2024dynamic}.\footnote{This numerical equivalence is not covered by
\cref{thm:aipw_new}(1) since $\tilde \bpi$ is not an implicit design. Nevertheless, in
these applications, $\tilde \bpi$ lives in the linear span of the covariates, and so the
covariance in \eqref{eq:decomp_aipw_corr} is exactly zero.}

In almost all of the scenarios, the upper bound \eqref{eq:worst-case} is large, indicating
that the discrepancy in the designs could in principle result in very different
estimates. An adversary could choose potential outcomes to be perfectly aligned with the
differences $\tilde\pi_i - \pi_i$, and regression would do poorly if $\tilde \pi_i$ were
the true design. Fortunately, in all scenarios,  the empirical correlation in \eqref
{eq:decomp_aipw_corr} is much smaller than its worst-case value
 \eqref{eq:worst-case}---with  all but three scenarios below 10\% of \eqref
  {eq:worst-case}. It appears that, throughout the scenarios, misspecification of the
  regression as a model of the outcomes and its misspecification as a model for treatment
  assignment are sufficiently uncorrelated that the regression is nevertheless a
  reasonable estimator. 

While it seems that \citet{angrist2008mostly}'s optimism---that simple estimators almost
always have quasi-experimental interpretations---mostly holds up in this approximate
sense, a clear exception is \citet{ang2023birth} (with the caveat that we only use
publicly available covariates in \citet{ang2023birth}'s specification). There, alternative
estimates of the treatment probability (logit or ML) change the estimate by more than two
standard errors---roughly \emph{half} of the OLS-estimated effect. This difference is
attributable to the regression residuals having relatively high correlation with the
discrepancy in designs. There, the plausible misspecification of regression as a treatment
model interacts with its failures as an outcome model, making the estimate sensitive to
model specification.

\subsection{Diagnostics for  the implicit estimand} We conclude by assessing the implicit
estimand. In general, one could compute the empirical analogue of $\omega_i$ in 
\eqref{eq:implicit_estimand} by again plugging in the sample Gram matrix of the OLS
 estimator. Practitioners can assess whether any of these weights are negative, and
 reweight the regression estimator if alternative weighted averages are desirable. For a
 given weighted ATE, \citet{poirier2024quantifying} propose methods for assessing the size
 of the implicit subpopulation the weighting targets.

 In the context of \eqref{eq:angrist98intro}, that implicit estimand is simply the
 weighted average treatment effect whose weights are proportional to $\pi_i
 (1-\pi_i)$. Since essentially all implicit designs are rejected in our battery of papers
 (\cref{tab:diagnostics_implicit_design}), we cannot assess the corresponding estimand
 directly. We instead study the estimand that  weights proportional to $\tilde \pi_i
 (1-\tilde\pi_i)$ for a design considered in \cref{fig:forest}. The estimand under the
 recalibrated
 design $\tilde
 \bpi$ is the closest analogue to the estimand of the regression.

If $\tilde\bpi$ were the true design, then this weighted estimand would be a treatment
effect that is in a sense easiest to estimate, per a result by
\citet{goldsmith2022contamination}. However, per \cref{fig:methods}, essentially no paper
in our census explicitly targets a weighted ATE.\footnote{In fact, according to the AI
agents, one of the only two papers mentioning a weighted average effect is in fact \citet
{goldsmith2022contamination}. The other is a methodology paper by \citet
{frandsen2023judging}.}  Among the named estimands, the ATE is clearly more popular. It
also has clearer theoretical and policy relevance \citep
{kwon2025estimating,poirier2024quantifying,sloczynski2022interpreting,mogstad2024instrumental}.
We thus ask whether changing the estimand makes a big difference. Unfortunately, due to
poor overlap, direct estimates of the ATE are usually too noisy or unavailable. We follow
\citet{crump2009dealing} and trim out observations with $\tilde \pi_i \not\in [0.1,0.9]$.

Estimates for the trimmed ATE are easy to compute by simply reweighting the AIPW estimator
\eqref{eq:aipw} with a weight proportional to $\one(\tilde\pi \in [0.1,0.9])$. We mark
 these estimates with a $\times$ in \cref{fig:forest}. The overall pattern is similar to
 the diagnostics for the implicit design: Most of these estimates also hug tightly their
 OLS or weighted AIPW counterparts. The gaps between $\bullet$ and
 $\times$---reflecting the effect of reweighting the estimand---are only larger than 1 OLS
 standard error for three scenarios, concentrated in applications with particularly poor
 overlap, \citet{dippel2021leadership} and \citet{cage2023heroes}.

\begin{table}[htb]\centering

\caption{Representativeness of the implicit estimand across designs, in the sense of
  \citet{poirier2024quantifying}: the size of the subpopulation for which the weighted
  estimand is an average treatment effect. This statistic is $\frac{1} {n}\sum_i\omega_i /
  \max_i \omega_i$, for $\omega_i = \omega_i(1)$.}
\label{tab:representativeness}

\begin{tabular}{lcccccccc}
\toprule
 & \multicolumn{2}{c}{Recalibrated} & \multicolumn{2}{c}{Logit} & \multicolumn{2}{c}{ML} & \multicolumn{2}{c}{Paper} \\
Paper & $\bullet$ & $\times$ & $\bullet$ & $\times$ & $\bullet$ & $\times$ & $\bullet$ & $\times$ \\
\midrule
Xu (2018) & 0.12 & 0.10 & -- & -- & 0.13 & 0.13 & -- & -- \\
Blakeslee et al.\ (2020) & 0.71 & 0.84 & 0.71 & 0.88 & -- & -- & -- & -- \\
Dippel and Heblich (2021) & 0.28 & 0.38 & 0.19 & 0.21 & -- & -- & 0.12 & 0.11 \\
Herrnstadt et al.\ (2021) & 0.90 & 1.00 & 0.91 & 1.00 & 0.87 & 0.97 & -- & -- \\
Ang (2023) & 0.59 & 0.65 & 0.52 & 0.62 & 0.53 & 0.62 & -- & -- \\
Cag\'e et al.\ (2023) & 0.10 & 0.11 & 0.09 & 0.10 & 0.01 & 0.01 & -- & -- \\
Covert and Sweeney (2023) & 0.30 & 0.43 & 0.26 & 0.31 & 0.32 & 0.36 & 0.69 & 0.88 \\
Cervellati et al.\ (2024) & 0.74 & 0.99 & 0.75 & 0.94 & 0.81 & 1.00 & 0.74 & 0.99 \\
Chambru et al.\ (2024) & 0.62 & 0.82 & 0.63 & 0.77 & 0.19 & 0.19 & 0.89 & 1.00 \\
\bottomrule
\end{tabular}

\begin{proof}[Notes]
$\bullet$: Estimand with weights proportional to $\tilde\pi(1-\tilde\pi)$.
$\times$: Estimand with weights proportional to $\one(\tilde\pi_i \in[0.1,
0.9])$. 
\end{proof}

\end{table}

Strikingly, many of these applications  have quite poor overlap (5 out of 9 papers have at
least some $\tilde \pi$ for which less than half of the observations have $\tilde \pi_i
\in [0.1,0.9]$). The treatment effect estimates for these applications are thus
representative for a substantially smaller share of the population. Following
\citet{poirier2024quantifying}, for either estimand we can compute the size of the largest
subpopulation for which the estimand is an average treatment effect---a statistic that
\citet{poirier2024quantifying} call representativeness. Because of poor overlap,
representativeness is poor for many scenarios. Moving from the weighted estimand to the
trimmed estimand improves representativeness by on average 10pp across the 50 scenarios.
This moderate improvement indicates that most of the poor representativeness comes from
poor overlap rather than from continuously downweighting covariates with worse overlap.

Taken together, these empirical results show that it is very common to estimate causal
effects with regression justified under implicit models of treatment. Interpreting
regression estimates as quasi-experimental commits the analyst to a particular model of
treatment---the implicit design---and to a particular weighting of heterogeneous
effects---the implicit estimand. These are rarely made explicit, and in our re-analysis of
several papers, the implicit design fails an exact quasi-experimental interpretation in
every case. While that failure is frequently inconsequential, at least for some explicit
models of treatment---vindicating \citet{angrist2008mostly}'s optimism---it can in
principle be quite large and indeed is large in one application \citep{ang2023birth}.
Since one cannot tell whether this is consequential ex ante without making the design
explicit, we recommend that practitioners make explicit both models of treatment
assignment and target estimand, even when using regression.

\section{Conclusion}
\label{sec:conclusion}

Linear regressions are ubiquitous. Interpreting their results as causal, thanks to
quasi-random assignment, is similarly commonplace, motivated by an optimism that
regression estimands have quasi-experimental interpretation. This paper studies the
necessary conditions that this interpretation imposes on treatment assignment and
evaluates this optimism both theoretically and empirically. 

We do so by studying the comparisons that regression estimands make under random
assignment. Requiring that a regression be minimally quasi-experimental imposes linear
restrictions in the design. The set of designs that satisfy these restrictions can be
thought of as models of treatment assignment that the regression implicitly specifies.
Each design also pinpoints a corresponding estimand that the regression implicitly
chooses. Indeed, the regression is numerically related to AIPW estimators to the extent
that the design used by the latter satisfies these linear restrictions. 

Understanding  quasi-experimental interpretation of regressions in this way essentially
reduces to mechanical computations that can be scaled and automated. These computations
can aid in examining new theoretical properties of particular specifications, itself the
subject of a highly influential recent literature. In several theoretical vignettes,
these computations unify and strengthen disparate strands of the literature.
Additionally, we find that regressions with interactions and with two-way fixed effects
have fragile design-based interpretations. This calls for caution and nuance when using
them and presenting their results.

\FloatBarrier

\bibliographystyle{aer}
\bibliography{main.bib}

\appendix

\counterwithin{theorem}{section}
\counterwithin{prop}{section}
\counterwithin{lemma}{section}
\counterwithin{rmk}{section}
\counterwithin{as}{section}
\counterwithin{ex}{section}

\begin{appendices}

\section{\Cref{thm:main} under i.i.d. sampling}
\label{asec:superpop}

This section states and proves an analogue of \cref{thm:main} under i.i.d. sampling. To
do so, we also state an analogue of \cref{defn:level_irrelevance}. 
Let $(Y_i(\cdot), X_i, W_i) \iid P$ and assume that $Y_i = Y_i(W_i)$. Assume that $Y_i
(w) \in \R^T$ and $W_i \in \mathcal W$. Define the random variables $
    \pi^*_i(w) = \pi(w; Y_i(\cdot), X_i) \equiv P(W=w \mid Y_i(\cdot), X_i). 
$
Let $Y_i'(\cdot) = Y_i(\cdot) + C_i$ for some random variable $C_i \in \R^T$ where \[ C_i
\mid (Y_i(\cdot), X_i, W_i) \sim P_C(Y_i(\cdot), X_i) \text{ independently across $i$.}\]
This restricts $C_i \indep W_i \mid Y_i(\cdot), X_i$. Define $P'$ as the joint
distribution of $ (Y_i' (\cdot), W, X_i)$ where $(Y_i(\cdot), W, X_i) \iid P$. Let
$\bz(X_i, W_i) \in \R^{T
\times k}$ be a covariate transform, and assume that \[
    G_n(P) \equiv \E_P[\bz(X_i, W_i)'\bz(X_i, W_i)]
\] exists and is invertible. For a distribution $Q$ of potential outcomes $Y_i
 (\cdot)$, treatment $W_i$, and covariates $X_i$, let $Y_i^Q = Y_i(W_i)$. Thus $Y_i^P =
 Y_i(W_i)$ and $Y_i^{P'} = Y_i'(W_i)$.  A population regression estimand is defined as \[
    \tau(P) = \Lambda G_n(P)^{-1} \E_P\bk{
        \bz(X_i, W_i)' Y_i^P
    }.
\]
Assume that $P$ has bounded moments so that $\tau(P)$ exists.

\begin{defn} The regression estimand $\tau(P)$ is minimally quasi-experimental under $P$
 if for all $P'$ corresponding to an arbitrary $P_C \sim C \mid (Y_i(\cdot), X_i)$, we
 have $
     \tau(P') = \tau(P)
    $
    whenever $\tau(P')$ exists.
\end{defn}

\begin{prop}
Define $\rho_i(W_i) = \Lambda G_n(P)^{-1} \bz(X_i, W_i)'$.
    $\tau(P)$ is minimally quasi-experimental under $P$ if and only if $
        \sum_{w\in \mathcal W} \pi_i^*(w) \rho_i(w) = 0
    $
    $P$-almost surely. 
\end{prop}
\begin{proof}
 Note that $G_n(P) = G_n(P')$
since $G_n$ does not depend on outcomes. Thus, taking the expectation under the joint
distribution $(Y_i(\cdot), W_i, X_i, C_i)$: \begin{align*}
\tau(P') - \tau(P) &= \E\bk{
        \rho_i(W_i) \pr{
            Y_i' (W_i) - Y_i(W_i)
        }
    }  = \E\bk{
        \rho_i(W_i) C_i
    }  = \E\bk{
        C_i \sum_{w\in \mathcal W} \one(W_i = w) \rho_i(w) 
    } \\
    &= \E\bk{
        C_i \sum_{w\in \mathcal W} \E[\one(W_i = w) \mid X_i, Y_i(\cdot), C_i] \rho_i(w) 
    }  \\&= \E\bk{
        C_i \sum_{w\in \mathcal W} \E[\one(W_i = w) \mid X_i, Y_i(\cdot)] \rho_i(w) 
    }  = \E\bk{
        C_i \sum_{w\in \mathcal W} \pi_i^*(w) \rho_i(w) 
    } \numberthis \label{eq:if_direction_eqn}.
\end{align*}
If one chooses $P_C$ such that $
    C_i = \sum_{w\in \mathcal W} \pi_i^*(w) \rho_i(w),
$
then \[
    0 = \E\bk{
        C_i \sum_{w\in \mathcal W} \pi_i^*(w) \rho_i(w) 
    } = \E\bk{
       \norm[\bigg]{\sum_{w\in \mathcal W} \pi_i^*(w) \rho_i(w) }^2
    } .
\]
This implies that $\sum_{w\in \mathcal W} \pi_i^*(w) \rho_i(w) = 0$ $P$-almost surely.
This proves the only if direction. The if direction follows from 
\cref{eq:if_direction_eqn}.
\end{proof}

\section{Treating the Gram matrix as known}
\label{asec:known_gram_matrix}

This section discusses the ramification of treating $G_n$ as known in
\cref{item:i} in \cref{sub:leading}. This is subtle, since $G_n$ itself depends on the true
design $\bpi^*$. We can think of \eqref {eq:equation_simple} as  nonlinear
equations in $\pi_i^*$: \[
\pi_i^* \rho_i(1; \bpi^*) + (1-\pi_i^*) \rho_i(0; \bpi^*) = 0
\] where $\rho_i(\cdot ; \bpi^*)$ depends on $\bpi^*$ through $G_n$. These equations then
 define a nonlinear surface of designs $
\br{\bpi: \pi_i \rho_i(1; \bpi) + (1-\pi_i) \rho_i(0; \bpi) = 0\, \forall\, i}.$
This surface is defined solely as a function of the regression $(\Lambda, z(x_i, \cdot))$
 and thus does not reflect any information in the data.

 The implicit designs can then be thought of as choosing a point near\footnote{When an
 implicit design generates the same Gram matrix as the true design $G_n(\bpi) =
 G_n(\bpi^*)$, then the implicit design is on the surface.} this surface that is ``most
 consistent'' with the data by fixing $(G_n, \rho_i(\cdot))$ and solving the corresponding
 equations. For the regression
 \eqref{eq:angrist98intro}, this surface is the span of the column space of the covariates
  \[
 \br{\bpi: \pi_i = \tilde x_i'\delta, \delta \in \R^{\dim(x_i) + 1}};\] the implicit
 design chooses the point on the surface by projecting $\bpi^*$ to the covariates, since
 all other points are ruled out by large datasets. When $\bpi^*$ belongs to the surface,
 doing this exactly recovers $\bpi^*$.

Additionally, some joint distributions of treatment imply that $\hat G_n = G_n$ almost
surely. When treatment arises from such a joint distribution, there is no estimation
noise to ignore. We term those fixed Gram designs and discuss them in the following
remark:
\begin{rmksq}[Fixed Gram designs]
\label{rmk:fixed_gram}

Let $\Pi^*$ denote the joint distribution of $\bW_1,\ldots, \bW_n$. For some $\Pi^*$, it
is possible that $\hat G_n = G_n (\bpi^*)$ with probability one. Under such a design,
there is no estimation noise to abstract away. As an example, suppose that we consider
$Y_i = \alpha + \tau  W_i +
\epsilon_i$ with binary $W_i$, but $W_i$ is assigned in a \emph{completely randomized
 experiment}.\footnote{That is, there are $n_1$ treated units, selected uniformly at
 random from the set of $\binom{n}{n_1}$ subsets.} Then, the sample and the population
 Gram matrices are equal and fixed in every realization of the treatment, \[\hat G_n = G_n
 = \begin{bmatrix} 1 & n_1/n\\ n_1/n & n_1/n
\end{bmatrix}.\] We term these designs \emph{fixed Gram} designs for a particular
 specification.  Our results in the population can alternatively be read as exact results
 for fixed Gram designs.\footnote{Regression estimators in the design-based statistical
 inference literature are frequently analyzed under fixed Gram designs \citep[see, e.g.,]
 []
 {rambachan2020design,athey2018design,zhao2022regression,splawa1990application,lin2013agnostic}.}
\end{rmksq}

\section{Proofs of key results}

\begin{proof}[Proof of \cref{thm:main}]
We first show that \cref{defn:level_irrelevance} is equivalent to the design satisfying
\eqref{eq:pop_level_irrelevance_condition}. By 
\eqref{eq:estimand}, $\btau$ satisfies \cref{defn:level_irrelevance} if and only if for
all $i,t$ \[
\sum_{\bw \in \cW} \pi_i^*(\bw) \brho_{it}(\bw) c_{it} = 0 \iff \sum_{\bw \in \cW} \pi_i^*
(\bw) \brho_{it}(\bw) = 0.
\]
The sum-to-one condition is automatically satisfied by $\pi_i^*$. This is then equivalent
to \eqref{eq:pop_level_irrelevance_condition}. This shows that $\btau$ is MQE
if and only if $\bpi^*$ is an implicit design. $\bpi^*$ is proper and generates $G_n$ by
definition.
The statement that $\bpi^*$ is an implicit design is equivalent to (1) and (2). 

When $\btau$ is minimally quasi-experimental, the estimand $\btau$ is by definition the implicit
estimand under $\bpi^*$. This proves the second statement. 
\end{proof}

\begin{proof}[Proof of \cref{cor:binary_main}]
Applying \cref{thm:main}, we need to show that \cref{cor:binary_main}(1) is equivalent to
\cref{thm:main}(1). Some proper implicit design exists if and only if $\rho_i(1)\rho_i(0)
\le0$
for all $i$.  When it exists, it satisfies $\pi_i = -\rho_i(0)/(\rho_i(1) - \rho_i(0))$.
Thus \cref{cor:binary_main}(1) is equivalent to \cref{thm:main}(1). In turn,
\cref{thm:main}(2) is equivalent to \cref{cor:binary_main}(2). \qedhere 

\end{proof}

\begin{as}
\label{as:aipw_recentered_rank}
In the notation of \cref{thm:aipw_new}(2), the matrix
\[
    \tilde G_n \equiv \frac1n\sum_{i=1}^n \sum_{w\in\cW}\tilde\pi_i(w)
    \left\{
        z_1(x_i,w)-\sum_{w'\in\cW}\tilde\pi_i(w')z_1(x_i,w')
    \right\}
    \left\{
        z_1(x_i,w)-\sum_{w'\in\cW}\tilde\pi_i(w')z_1(x_i,w')
    \right\}'
\]
is nonsingular.
\end{as}

\begin{proof}[Proof of \cref{thm:aipw_new}]
For part (1), we have $\sum_w\tilde\pi_i(w)\brho_i(w)=0$, so the weights
$\tilde \omega_i(w) = \tilde\pi_i(w)\brho_i(w)$ define a contrast. Observe that
\[
    \frac1n\sum_i\brho_i(W_i)(Y_i-z(x_i,W_i)'\hat\beta)
    =
    \Lambda G_n^{-1}
    \left\{\frac1n\sum_i z(x_i,W_i)(Y_i-z(x_i,W_i)'\hat\beta)\right\}
    =
    0
\]
by the OLS normal equations. Moreover,
\[
    \frac1n\sum_i\sum_w
    \tilde\pi_i(w)\brho_i(w)z(x_i,w)'\hat\beta
    =
    \Lambda G_n^{-1}G_n(\tilde\bpi)\hat\beta
    =
    \Lambda\hat\beta,
\]
where the last equality uses that $\tilde\bpi$ generates $G_n$. Combining the last two
displays implies the statement. 

For part (2), set
\[
    \tilde A
    =
    -\tilde G_n^{-1}
    \left[
        [I_k,0]G_n^{-1}
        \colvecb{2}{I_k}{0}
    \right]^{-1}
\quad 
    \tilde\omega_i(w)
    =
    \tilde\pi_i(w)\tilde A\{\tilde\brho_i-\brho_i(w)\}. \numberthis 
    \label{eq:tilde_omega}
\]
The matrix $\tilde A$ is nonsingular by \cref{as:aipw_recentered_rank} and invertibility
of $G_n$. Observe that $\sum_w \tilde \omega_i(w) = 0$ because
$\sum_w\tilde\pi_i(w)\{\tilde\brho_i-\brho_i(w)\}=0$ for every $i$. Moreover,
\[
    \tilde\brho_i-\brho_i(w)
    =
    -[I_k,0]G_n^{-1}
   \colvecb{2}{
    z_1(x_i,w)-\sum_{w'\in\cW}\tilde\pi_i(w')z_1(x_i,w')}
    {0}.
\]
Therefore, we may calculate
\[
\begin{aligned}
    \frac1n\sum_i\sum_w \tilde\omega_i(w)z(x_i,w)'
    &=
    -\tilde A
    [I_k,0]G_n^{-1}
   \colvecb{2}{I_k}{0}
    \begin{bmatrix}
    \tilde G_n & 0
    \end{bmatrix} =
    [I_k,0]
    =
    \Lambda,
\end{aligned}
\]
We thus have that \[
    \frac1n\sum_i\sum_w \tilde\omega_i(w)z(x_i,w)'\hat\beta-\Lambda\hat\beta=0.
\]

Also, since
$
    \frac{\tilde\omega_i(W_i)}{\tilde\pi_i(W_i)}
    =
    \tilde A\{\tilde\brho_i-\brho_i(W_i)\},
$
we have that
\begin{align*}
    &\frac1n\sum_i
    \frac{\tilde\omega_i(W_i)}{\tilde\pi_i(W_i)}
    (Y_i-z(x_i,W_i)'\hat\beta) =
    \frac{\tilde A}{n}\sum_i
    \tilde\brho_i\cdot (Y_i-z(x_i,W_i)'\hat\beta)
    -
    \frac{\tilde A}{n}\sum_i
    \brho_i(W_i)(Y_i-z(x_i,W_i)'\hat\beta).
\end{align*}
The second term is again zero by the OLS normal equations.
Combining the last two displays proves part (2).

For part (3), in \eqref{eq:angrist98intro}, the potential weights take the form
$
    \brho_i(w)
    =
    \frac{w-\pi_i}{\frac1n\sum_j\pi_j(1-\pi_j)},
$
where $\pi$ is the implicit design. Thus
\[
    \tilde\brho_i
    =
    \frac{\tilde\pi_i-\pi_i}{\frac1n\sum_j\pi_j(1-\pi_j)}.
\]
The construction in part (2) yields
\[
    \tilde A\tilde\brho_i
    =
    -\frac{\tilde\pi_i-\pi_i}{\frac1n\sum_j\tilde\pi_j(1-\tilde\pi_j)}
\quad 
    \tilde\omega_i(1)
    =
    \frac{\tilde \pi_i(1-\tilde \pi_i)}{\frac{1}{n} \sum_j \tilde \pi_j (1-\tilde \pi_j)}.
\]
Substituting these identities into part (2) gives
\[
    \tilde\tau_{\mathrm{AIPW}} - \Lambda\hat\beta
    =
    -\frac{1}{\sum_j\tilde\pi_j(1-\tilde\pi_j)}
    \sum_i(\tilde\pi_i-\pi_i)(Y_i-z(x_i,W_i)'\hat\beta). \qedhere
\]
\end{proof}

\begin{proof}[Proof of \cref{thm:zoo}]
We will repeatedly use \cref{thm:frisch-waugh} to residualize the treatment
variables against the controls to compute the potential weights.
\begin{enumerate}[wide]
    \item {Let
    $\delta = (\frac1n\sum_i x_ix_i')^{-1}\frac1n\sum_i\pi_i^*x_i$. After residualizing
    $W_i$ on $x_i$, the residualized treatment is $W_i-x_i'\delta$. Hence, for
    $\tilde G_n = \frac1n\sum_i\E[(W_i-x_i'\delta)^2] >0$,}
$
    \rho_i(w)=\tilde G_n^{-1}(w-x_i'\delta).
$

\begin{enumerate}[wide]
    \item By \cref{cor:binary_main}, the unique implicit design is
$
    \pi_i=-\frac{\rho_i(0)}{\rho_i(1)-\rho_i(0)}=x_i'\delta.
$
It generates $G_n$ because the entries of
$\frac{1}{n}\sum_i \E_{\pi_i} z(x_i,W)z(x_i,W)'$ depend on $\pi$ only through the
moments
$n^{-1}\sum_i \pi_i$ and $n^{-1}\sum_i\pi_i x_i$, and these match the corresponding
moments under $\pi_i^*$ by the projection normal equations and the constant in $x_i$.

    \item Immediate from the uniqueness in (a).
    \item The implicit weights satisfy
    $\omega_i=\pi_i\rho_i(1)=\tilde G_n^{-1}\pi_i(1-\pi_i)$. Since
    $\pi_i=x_i'\delta$ is the projection of $\pi_i^*$ on $x_i$,
\[
    \tilde G_n
    =\frac1n\sum_i\E[(W_i-\pi_i)^2]
    =\frac1n\sum_i\pi_i(1-\pi_i),
\]
where the last equality uses
$\frac1n\sum_i(\pi_i^*-\pi_i)=0$ and
$\frac1n\sum_i(\pi_i^*-\pi_i)\pi_i=0$. Thus
$\omega_i=\pi_i(1-\pi_i)/\{\frac1n\sum_j\pi_j(1-\pi_j)\}$, which is nonnegative when
$\pi_i=\pi_i^*\in[0,1]$.
\end{enumerate}

\item For $j=1,\ldots,J$, let
$\delta_j=(\frac1n\sum_i x_ix_i')^{-1}\frac1n\sum_i\pi_i^*(j)x_i$ and
$\Delta_i=(\delta_1'x_i,\ldots,\delta_J'x_i)'$. After FWL \cref{thm:frisch-waugh}, the
residualized treatment regressors
are
$W_{ij}-\delta_j'x_i$, so for a positive definite transformed Gram matrix
$\tilde G_n$,
\[
    \rho_i(\ell)=\tilde G_n^{-1}(e_\ell-\Delta_i),
    \qquad e_0=0_J .
\]

\begin{enumerate}[wide]
     \item The implicit-design equations are
\[
    0=\sum_{\ell=0}^J \pi_i(\ell)\rho_i(\ell)
    =\tilde G_n^{-1}\left\{\sum_{\ell=0}^J\pi_i(\ell)e_\ell-\Delta_i\right\}.
\]
Thus $\pi_i(j)=\delta_j'x_i$ for all $j=1,\ldots,J$ and
$\pi_i(0)=1-\sum_{j=1}^J\pi_i(j)$. As in part (1), this design generates $G_n$ because
$n^{-1}\sum_i\pi_i^*(j)x_i = n^{-1}\sum_i\pi_i(j)x_i$ for every $j$.

\item Again immediate from uniqueness of the implicit design in (a)

\item The implicit-estimand weights are
$
    \omega_i(\bpi,\ell)
    =
    \tilde G_n^{-1}\pi_i(\ell)\{e_\ell-(\pi_i(1),\ldots,\pi_i(J))'\}.
$
Therefore
\[
(\tau_1,\ldots,\tau_J)' = \frac{1}{n} \sum_i \sum_{\ell = 0}^J \omega_i(\bpi,
\ell) y_i(\ell). \numberthis \label{eq:implicit_estimand_pgp}
\]
\Cref{ex:pgp_example} gives an explicit design for which this estimand is
contaminated.
 \end{enumerate} 
 \item Let $x_{i0}=1-\sum_{\ell=1}^L x_{i\ell}$ and let $c_i$ denote the
 cell containing unit $i$. The saturated specification is a reparametrization of
\[
    Y_i = \sum_{\ell=0}^L x_{i\ell}(\alpha_\ell+\tau_\ell W_i)+\epsilon_i,
    \qquad
    \tau=\sum_{\ell=0}^L \bar x_\ell\tau_\ell . \numberthis 
    \label{eq:mean_within_group}
\]
Ordering the regressors as
$[x_{i0},x_{i0}W_i,\ldots,x_{iL},x_{iL}W_i]$, $G_n$ is block diagonal with
\[
    G_n = \begin{bmatrix}
        G_0 \\
    & \ddots \\
    &&G_L
    \end{bmatrix}
    \quad G_\ell = \begin{bmatrix}
        \bar x_\ell & \bar x_\ell \pi(\ell) \\
        \bar x_\ell \pi(\ell) & \bar x_\ell \pi(\ell)
    \end{bmatrix} \quad G_\ell^{-1} =
    \frac{1}{\bar x_\ell \pi(\ell)(1-\pi(\ell)) } \begin{bmatrix}
        \pi(\ell) & -\pi(\ell) \\
        -\pi(\ell) & 1
    \end{bmatrix}\]
where $\pi(\ell)$ is the mean of $\pi_i^*$ among units with $c_i=\ell$. Since
$\Lambda=[0,\bar x_0,\ldots,0,\bar x_L]$,
\[
    \Lambda G_n^{-1} = \bk{\frac{-\pi(0)}{\pi(0) (1-\pi(0))}, \frac{1}{\pi(0)(1-\pi(0))},
    \ldots, \frac{-\pi(L)}{\pi(L)(1-\pi(L))}, \frac{1}{\pi(L)(1-\pi(L))}}.
\]
For $c_i=\ell$,
$
z(x_i, w) =
    [0,\ldots, 0, 1, w, 0,\ldots,0]'
$
with nonzero entries in the $\ell$\th{} pair, so
\[
    \rho_i(w) = \Lambda G_n^{-1} z(x_i,w) = \frac{w - \pi(\ell)}{\pi(\ell)(1-\pi(\ell))}.
\]
\begin{enumerate}[wide]
    \item By \cref{cor:binary_main}, the unique implicit design sets
    $\pi_i=\pi(\ell)$ whenever $c_i=\ell$. The calculation above shows that this
    design generates $G_n$.

    \item Consequently $\pi_i^*=\pi_i$ if and only if the propensity score is
    constant within each cell.

    \item The treated potential-outcome weight is
    $\omega_i=\pi(\ell)\rho_i(1)=1$.
\end{enumerate}

\item \Cref{prop:interact_with_t} gives the potential weights for the
interacted regression for general centering $\bar x_t$. 
\begin{enumerate}[wide]
    \item  Applying it at $t=1$ and using
\cref{prop:kline_prop}, the implicit design has odds
\[
    \frac{\pi_i}{1-\pi_i}=\delta_0+\delta_1'(x_i-\bar x),
\]
after re-centering the weighted least-squares projection from
$x_i-\bar x_0$ to $x_i-\bar x$. This proves the implicit-design formula.
    \item Follows again from uniqueness
    \item 
    In the notation of \cref{prop:interact_with_t},
\[
    \tilde G_n\rho_i(1)\pi_i
    =
    \{1-\Gamma_{21}'(x_i-\bar x_1)-\theta_{01}-\theta_{11}'(x_i-\bar x)\}\pi_i.
\]
The proof of \cref{prop:kline_prop} gives
$\Gamma_{21}=-\theta_{11}$, so the expression in braces is constant in $i$. Therefore
$\omega_i=\rho_i(1)\pi_i$ is proportional to $\pi_i$. Finally,
$n^{-1}\sum_i\omega_i=1$: under the correctly specified constant-effect outcomes
$y_i(w)=vw$, the regression coefficient equals $v$, while
\cref{cor:binary_main} writes it as
$v\,n^{-1}\sum_i\omega_i$. Hence
$\omega_i=\pi_i/\{\frac1n\sum_j\pi_j\}$, the ATT weights.
\end{enumerate}

\item Let $\pi(\bw)=n^{-1}\sum_{i=1}^n\pi_i^*(\bw)$ and
$\bw_\pi=\sum_{\bw\in\cW}\pi(\bw)\bw$. By \cref{thm:frisch-waugh}, the potential weights
are those from the regression on residualized covariates
\[
    \ddot{\bw}
    =
    \bw-1_T\frac{1_T'\bw}{T}
    -\bw_\pi+1_T\frac{1_T'\bw_\pi}{T}.
\]
Its Gram matrix is the positive scalar
$\tilde G_n=\sum_{t=1}^T\sum_{\bw\in\cW}\pi(\bw)\ddot{\bw}_t^2$.

\begin{enumerate}[label=(\alph*), wide]
\item
An implicit design $\tilde\pi_i$ satisfies
\[
    0
    =
    \sum_{\bw\in\cW}\tilde\pi_i(\bw)\ddot{\bw}
    =
    \pr{I_T-\frac{1}{T}1_T1_T'}
    \sum_{\bw\in\cW}\{\tilde\pi_i(\bw)-\pi(\bw)\}\bw .
\]
Thus $\tilde\pi_i=\pi$ is an implicit design. Conversely, the last display implies
$\sum_{\bw\in\cW}\{\tilde\pi_i(\bw)-\pi(\bw)\}\bw\in\Span(1_T)$. This vector also lies in
$\Span(\cW)$, so it is zero by the assumption that $\mathcal W$'s span excludes $1_T$. As
a result, $\sum_{\bw\in\cW}\{\tilde\pi_i(\bw)-\pi(\bw)\}\bw = 0$; that is, some linear
combination of $w \in \cW$ is zero.
Since the nonzero elements of $\cW$ are linearly
independent, $\tilde\pi_i(\bw)=\pi(\bw)$ for all $\bw\ne0$; since both vectors sum to one,
the equality also holds at $\bw=0$. Hence the implicit design is unique.

The same $\pi$ generates the relevant Gram matrix, because
\[
    G_n
    =
    \frac{1}{n}\sum_{i=1}^n\sum_{\bw\in\cW}\pi_i^*(\bw)\bz(\bw)'\bz(\bw)
    =
    \sum_{\bw\in\cW}\pi(\bw)\bz(\bw)'\bz(\bw).
\]

\item
By part (a), $\pi_i^*=\pi_i$ if and only if $\pi_i^*(\bw)=\pi(\bw)$ for every $i$ and
$\bw$.

\item
For $\pi_i(\bw)=\pi(\bw)$, define
\[
    g_t(\bw)=\ddot{\bw}_t
    =
    \bw_t-\bw_{\pi,t}-\frac{1}{T}(1_T'\bw-1_T'\bw_\pi).
\]
Then
\[
    \bomega_t(\bpi,\bw)
    =
    \pi(\bw)\brho_{it}(\bw)
    =
    \frac{\pi(\bw)g_t(\bw)}
    {\sum_{s=1}^T\sum_{\bw\in\cW}\pi(\bw)g_s(\bw)^2}.
\numberthis
\label{eq:twfe_estimand}
\]

Under staggered adoption, let $a(\bw)$ be the adoption time. The residual $g(t,a(\bw))$ in
\citet[][Theorem 1(ii)]{athey2018design} satisfies
\[
    \one\{a(\bw)\le t\}=\bw_t,
    \quad
    \sum_{\bw_1:a(\bw_1)\le t}\pi(\bw_1)=\bw_{\pi,t},
\]
and
\[
    a(\bw)\one\{a(\bw)\le T\}
    =
    T+1-1_T'\bw-(T+1)\one\{\bw=0_T\}.
\]
Substituting these identities into their expression gives
$g(t,a(\bw))=\bw_t-\bw_{\pi,t}-T^{-1}(1_T'\bw-1_T'\bw_\pi)=g_t(\bw)$. Hence
\eqref{eq:twfe_estimand} matches their weights.\qedhere
\end{enumerate}
\end{enumerate}
\end{proof}

\begin{proof}[Proof of \cref{prop:forbidden_interactions}]
By \cref{thm:frisch-waugh}, residualize $[W_i, W_i x_{1i}]$
against $1,x_{2i}$. Let $\delta_0,\delta_1$ be the projection coefficients of
$\pi_i^*$ on $1,x_{2i}$, and let $\Gamma_0,\Gamma_1$ be the projection coefficients of
$\pi_i^*x_{1i}$ on $1,x_{2i}$. The FWL-transformed regressors are
$[
 {W_i-\delta_0-\delta_1'x_{2i}}, 
    {W_ix_{1i}-\Gamma_0-\Gamma_1x_{2i}}
]'$
so the potential weights are
\[
    \rho_i(w)
    =
    \tilde G_n^{-1}
    \colvecb{2}{w-\delta_0-\delta_1'x_{2i}}
    {wx_{1i}-\Gamma_0-\Gamma_1x_{2i}},
\]
for a nonsingular transformed Gram matrix $\tilde G_n$. Therefore the
implicit-design equations are equivalent to
\[
    0
    =
    \pi_i
    \colvecb{2}{1-\delta_0-\delta_1'x_{2i}}{x_{1i}-\Gamma_0-\Gamma_1x_{2i}}
    +(1-\pi_i)
    \colvecb{2}{-\delta_0-\delta_1'x_{2i}}{-\Gamma_0-\Gamma_1x_{2i}}
\implies 
\pi_i=\delta_0+\delta_1'x_{2i},\quad
    \pi_i x_{1i}=\Gamma_0+\Gamma_1x_{2i}.
\]
Hence an implicit design exists exactly when
$(\delta_0+\delta_1'x_{2i})x_{1i}=\Gamma_0+\Gamma_1x_{2i}$ for every $i$, and then it
is unique with $\pi_i=\delta_0+\delta_1'x_{2i}$. If $\btau$ satisfies
\cref{item:M}, \cref{thm:main} further requires $\pi_i^*=\pi_i$, giving the two
necessary linearity conditions stated in the proposition.
\end{proof}

\begin{proof}[Proof of \cref{prop:time_varying_cov}] Let double dots $\ddot{v}$ denote
 population residualization of variable $v$ with respect to individual and time fixed
 effects. By \cref{thm:frisch-waugh}, we can study the regression of $\by_i$ on
\[
    \ddot{\bW}_i-\ddot{\bx}_i\beta_{w\to x}.
\]
Thus the potential weights are proportional to
$\ddot{\bw}-\ddot{\bx}_i\beta_{w\to x}$. The level-irrelevance equations therefore imply
\begin{align}
    \ddot{\bx}_i\beta_{w\to x}
    =
    \sum_{\bw\in\cW}\pi_i(\bw)\ddot{\bw}.
    \label{eq:level_irrelevance_panel_with_covariates}
\end{align}
The right-hand side lies in $\Span(\cW\cup\{1_T\})$: residualizing a
treatment path by two-way fixed effects only subtracts a linear combination of treatment
paths and multiples of $1_T$. On the other hand, for
$\bar x=n^{-1}\sum_j\bx_j$ and some scalar $c_i$,
\[
    \ddot{\bx}_i\beta_{w\to x}
    =
    (\bx_i-\bar x)\beta_{w\to x}+c_i1_T .
\]
Since $1_T$ already belongs to the span, the necessary condition in the
proposition follows. If $\beta_{w\to x}=0$, then
\eqref{eq:level_irrelevance_panel_with_covariates} reduces to
$\sum_{\bw}\pi_i(\bw)\ddot{\bw}=0$. Under the stated linear-independence condition, the
same argument as in \cref{thm:zoo}(\ref{item:athey-imbens})(a) gives the unique solution
$\pi_i(\bw)=n^{-1}\sum_j\pi_j^*(\bw)$. \qedhere
\end{proof}

\begin{proof}[Proof of \cref{iv:main}]
    Technically, we never precisely defined \cref{item:M} for TSLS. For
    completeness' sake, we say that
    $\btau$ is minimally quasi-experimental if under $\bpi^*$ it is invariant to replacing
    $y_i(\cdot)$ with $y_i(\cdot) + c_i$. 
    Then exactly analogous to \cref{thm:main}, it is
    minimally quasi-experimental if and only if an implicit design exists and equals $\bpi^*$.
    The estimand $\btau$ depends on the profile of compliance types $d_i^*(\cdot),
    i=1,\ldots,n$. The estimand is \eqref{eq:estimand_tsls}. By definition, $\btau$
    assigns proper weights if and only if $\bpi^*$ belongs to the set of implicit
    compliance profiles
    \eqref{eq:implicit_compliance_profile}. This is in turn equivalent to (3) and (4).
\end{proof}

\begin{proof}[Proof of \cref{item:tsls}]
    See \cref{sub:proof_tsls}.
\end{proof}

\end{appendices}

\newpage 

\begin{appendices}

\begin{center}
    \textbf{Online Appendix for} \\\textbf{``Potential weights and implicit causal
    designs in
    linear
    regression''}

Jiafeng Chen

\today
\end{center}

\DoToC

\section{Proofs of auxiliary results}

\subsection{Auxiliary results for \cref{thm:zoo}}

\begin{exsq}
\label{ex:pgp_example}
    It is simple to analyze \cref{thm:zoo}(\ref{item:pgp22}) fully numerically, at least
given a concrete data-generating process. To do so, we follow the numerical example in
Section 2.2 of \citet{goldsmith2022contamination}. Consider the specification \[ Y_i
=\alpha_0 + \tau_1 W_{i1} + \tau_2 W_{i2} + \gamma'x_i + \epsilon_i
\] with binary $x_i$ and $J=2$. Suppose exactly half the units have $x_i = 1$. For the
units with $x_i = 0$, $\pi^*_i(0) = 0.5$, $\pi^*_i(1) = 0.05$, and $\pi^*_i(2) = 0.45$.
For the units with $x_i = 1$, $\pi_i^*(0) = 0.1$, $\pi_i^*(1) = 0.45$, and $\pi_i^*(2) =
0.45$. The coefficients of interest are $\btau = [\tau_1, \tau_2]'$.

Numerically, we can verify that the assignment
probabilities $\bpi^*$ are the only solution to
\eqref{eq:pop_level_irrelevance_condition}, and thus $\bpi^*$ is the only implicit design
for  this regression specification. The corresponding
implicit estimand for $\tau_j$ is
then \[
\tau_j = \frac{1}{n} \sum_{i=1}^n \sum_{k=0}^J \underbrace{\pi_i(k)\rho_{ij}(k)}_{\omega_
{ij} (\bpi, k)} y_i(k),
\]
which is the  sum of potential outcomes weighted by $\omega_{ij}(\bpi, k)$.

We can calculate the implicit estimand as well, and inspecting the implicit estimand
allows us to recover the contamination bias in \citet{goldsmith2022contamination}. For each $x_i$ value and each estimand, we
tabulate $\omega_{ij}(\cdot)$ for individuals with those $x_i$: \[
\begin{tabular}{ccccccc} \toprule
 & \multicolumn{3}{c}{$\tau_1$} & \multicolumn{3}{c}{$\tau_2$} \\
        & $\omega_{i1}(0)$ &$\omega_{i1}(1)$ & $\omega_{i1}(2)$ & $\omega_{i2}(0)$
        &$\omega_{i2} (1)$ & $\omega_{i2}(2)$\\ \midrule
$x_i = 0$ & $-140/106$  &   $41/106$   &   {\color{red}$99/106$} & $-160/106$  &   {\color{red} $9/106$}   &   $151/106$\\
$x_i = 1$ & $-72/106$                &  $171/106$ & {\color{red}$-99/106$}  & $-52/106$
&  {\color{red} $-9/106$} & $61/106$ \\ \bottomrule
\end{tabular}
\] We find that the implicit estimand for $\tau_j$ does not solely involve $y_i(0)$ and
$y_i (j)$, contrary to intuition; moreover, $\tau_j$ does not measure the same estimand
for individuals with $x_i=0$ and $x_i=1$. This echoes the result by
\citet{goldsmith2022contamination}. In fact, under $\bpi^*$, the estimand
$\tau_1$ is \[ \tau_1 = \frac{1}{2} \bk{\frac{41}{106} \bar\tau_{1 \mid x=0} +
\frac{99}{106} \bar\tau_ {2 \mid x=0}} + \frac{1}{2} \bk{\frac{171}{106} \bar\tau_{1 \mid
x=1} - \frac{99}{106} \bar\tau_ {2 \mid x=1}} \] where $\bar \tau_{k \mid x=j}$ is the
mean of $y_i(k) - y_i(0)$ among those with $x_i = j$. This decomposition exactly matches
the decomposition (7) in \citet{goldsmith2022contamination}.
\end{exsq}

\subsection{Proof of \cref{item:tsls}}
\label{sub:proof_tsls}

\begin{proof}
    \begin{enumerate}[wide]
        \item 

For \eqref{eq:tsls-setup}, partition conformably \[
    H_n^{-1} = G_{zt} = \begin{bmatrix}
        H_{zt} & H_{zx} \\ 
        H_{xt} & H_{xx}
    \end{bmatrix}.
\]
Let $G_{1|2} = H_{zt} - H_{zx} H_{xx}^{-1} H_{xt}$ be a $J \times J$ matrix. It is
invertible if $H_n$ is invertible.  We compute that \[
    \rho_i(j) = [I_J, 0] H_n \colvecb{2}{e_j}{x_i} = G_{1|2}^{-1} \pr{
        e_j - H_{zx} H_{xx}^{-1} x_i
    }.
\]
The implicit design is therefore $\pi_i(\cdot)$ such that \[
    0 = G_{1|2}^{-1} \pr{ (\pi_i(1),\ldots,\pi_i(J))' - H_{zx} H_{xx}^{-1}
    x_i} \iff \pi_i(j) = (H_{zx} H_{xx}^{-1}
        x_i)_j.
\]
We conclude the proof by observing that $(H_{zx} H_{xx}^{-1} x_i)_j$ is exactly the
        population projection of $\pi_i^* (j)$ on $x_i$. That is, $(H_{zx} H_{xx}^{-1}
        x_i)_j
        = x_i'\delta_j$.

\item \begin{enumerate}[wide]
    \item  When $J+1=2$, we can compute \begin{align*}
    G_{1|2} &= \frac{1}{n} \sum_{i=1}^n \pi^*_i d_i^*(1)  - \frac{1}{n} \sum_{i=1}^n
        \underbrace{H_{zx} H_{xx}^{-1} x_i}_{\pi_i} \br{\pi_i^* d_i^*(1) + 
        (1-\pi_i^*) d_i^*(0))} \\ 
        &= \frac{1}{n} \sum_{i=1}^n \pi_i^* (1-\pi_i) d_i^*(1) - \frac{1}{n}\sum_{i=1}^n
        \pi_i (1-\pi_i^*) d_i^*(0) \\
        &= \frac{1}{n} \sum_{i=1}^n \pi_i^*(1-\pi_i^*) (d_i^*(1) - d_i^*(0)).
    \end{align*}

    Given a compliance profile $d_i(\cdot)$, 
the estimand is in turn characterized by \[\omega_i(1) = \rho_i(1) \pi_i^* =
    \frac{1}{G_{1|2}} (1-\pi_i^*) \pi_i^* = \frac{ (1-\pi_i^*) \pi_i^*}{\frac{1}{n}
    \sum_{i=1}^n \pi_i^*(1-\pi_i^*) (d_i^*(1) - d_i^*(0))},\] where $
        \tau = \frac{1}{n}\sum_{i=1}^n \omega_i(1) (y_i(d_i(1)) - y_i(d_i(0))).
    $
   For this estimand to assign proper weights, $\omega_i(1)$ should weakly have the same
   sign as
   $d_i(1) - d_i(0)$. Since the sign of $\omega_i(1)$ is determined globally by $\frac{1}{n}
    \sum_{i=1}^n \pi_i^*(1-\pi_i^*) (d_i(1) - d_i(0))$, this implies that either $d_i
    (1) \ge d_i(0)$ or the reverse for all $i$.

\item Without covariates, let \[
    Q = \pr{[0, I_{J}] - (\pi(1),\ldots,\pi(J))' 1_{J+1}'} \cdot
    \diag
    (\pi(0),\ldots, \pi(J)).
\]

Because $x_i$ is just a constant, the implicit design $\pi(\cdot) = \pi^*_i(\cdot)$ is the
same for every $i$. Observe that with this definition, the $\R^{J \times (J+1)}$ matrix
representing the implicit estimand is equal to \[ [\bomega(0; \bpi), \ldots, \bomega
(J; \bpi)] = G_{1|2}^{-1} Q.
\]

A given compliance type $d(\cdot)$ can be represented with a matrix $C(d) \in \R^{(J+1)
\times J}$, for which \[
    C_{kk'}(d) = \one(d(k) = k').
\]
For this compliance type, the matrix $\Omega(d)$ whose entries are $
    \Omega_{kk'}(d) = \omega_i^{(k,k')}(d)
$
can be represented as $
    \Omega(d) = G_{1|2}^{-1} Q C(d). 
$
This matrix does not depend on $i$. 

If $d_i = d$ is included in an implicit compliance profile, then  $\Omega(d)$ is a
positive semidefinite diagonal matrix. When that happens, since
$
QC(d) = G_{1|2} \Omega(d),
$
each column of $QC(d)$ is equal to the corresponding column of $G_
{1|2}$, scaled by some nonnegative scalar $c_{d,k}$. To that end, note that the $k$\th{}
column of $Q C(d)$ is \[
    (Q C(d))[:, k] = -\pr{\sum_{w: d(w) = k} \pi(w)} \colvecb{3}{\pi(1)}{\vdots}
    {\pi(J)} + e_d \odot \colvecb{3}{\pi(1)}{\vdots}
    {\pi(J)} 
\]
where $e_d$ is the binary vector for which $e_{d,m} = 1$ if $d(m) = k$, for $m \in 
\br{1,\ldots, J}$. 

We first show common compliance. Take two compliance types $d_i, d_j$ in the implicit
compliance profile. Without loss, suppose $i=1, j=2$. Note that both $QC(d_\ell)[:,k]$ for
$\ell \in\br{ 1,2}$ are proportional to $G_{1|2}[:,k]$. 
Then for all $k = 1,\ldots, J$, either (i) one of $(Q C (d_\ell))[:,k] = 0$ for some
$\ell = 1,2$, or (ii)
\[
    (Q C(d_1))[:,k] = c (Q C(d_2))[:,k]
\]
for some $c > 0$. In case (i), the displayed formula for $(Q C(d_\ell))[:,k]$ implies
that $e_{d_\ell}$ is a constant binary vector and equal to $\sum_{w: d(w)=k} \pi(w)$, so
$d_\ell$ is a $k$-always taker or a
$k$-never taker. In case (ii), after dividing by the strictly positive vector
$(\pi(1),\ldots,\pi(J))'$, we have
\[
    e_{d_1} - c e_{d_2}
    =
    \pr{\sum_{w:d_1(w)=k}\pi(w) - c\sum_{w:d_2(w)=k}\pi(w)}1_J .
\]
For two binary vectors, this can hold only if $e_{d_1}=e_{d_2}$, except for the already
covered case where one type is a $k$-always taker and the other is a $k$-never taker.
Given $e_{d_1}=e_{d_2}$, the same display forces
$\one(d_1(0)=k)=\one(d_2(0)=k)$ unless both types are again $k$-always/never takers. Thus
any
two $k$-compliers have the same set of instrument values leading to treatment $k$.
Because this holds for every $k=1,\ldots,J$, all implicit compliance profiles satisfy
common compliance.

Suppose now that the profile contains a full complier $d_f$. Each column of $QC(d_f)$ is
nonzero, so each diagonal entry of $\Omega(d_f)$ is positive and $QC(d_f)$ has rank $J$.
Since $d_f$ takes every nonzero treatment value, either $d_f$ is a permutation of
$\br{0,\ldots,J}$ or $d_f$ never takes value $0$. The latter case gives
$C(d_f)1_J=1_{J+1}$ and hence $QC(d_f)1_J=Q1_{J+1}=0$, contradicting rank $J$.
Thus, writing $f$ for the inverse permutation, $d_f(f(k))=k$ for every
$k=0,\ldots,J$.

Take any other $d_j(\cdot)$. If $d_j$ is not constant, then any nonzero value it takes is
taken by a $k$-complier. Common compliance with $d_f$ then implies that, for every
$k=1,\ldots,J$, $d_j^{-1}(k)$ is either empty or equal to $\br{f(k)}$. Therefore
$d_j(f(k))\in\br{0,k}$ for all $k$, so the profile satisfies extended
monotonicity. \qedhere
\end{enumerate}
    \end{enumerate}
\end{proof}

\section{Additional results}

\subsection{Relation to \citet{blandhol2022tsls} and \citet{bhuller20242sls}}
\label{sub:relation_to_blandhol}

\subsubsection{Relation to \citet{blandhol2022tsls}, Theorem 1}

We now discuss the relation to \citet{blandhol2022tsls} between their Theorem 1 and our
\cref{item:tsls}(1) and (2)(a). In our finite-population setting, the terminology of
\citet{blandhol2022tsls} has analogues: \begin{enumerate}
    \item Montonicity / Assumption MON: $d_i^*(1) \ge d_i^*(0)$ for all $i$
    \item $\tau$ is weakly causal: This is equivalent to that $\tau$ assigns proper
    weights and can be decomposed into \begin{enumerate}
        \item $\tau$ is minimally quasi-experimental: $\tau = \frac{1}{n} \sum_i \omega_i^* (y_i(1)
        -
        y_i(0))$ for some $\omega_i^* \in \R$
        \item $\tau$ has convex weights $\omega_i^* \ge 0$. 
   
    \end{enumerate}
     \item  Exogeneity/Assumption EXO: $\pi_i^*$ not a function of $y_i(\cdot),
     d_i(\cdot)$.
\end{enumerate}

Theorem 1 in \citet{blandhol2022tsls} states that under exogeneity and monotonicity,
$\tau$ from \eqref{eq:tsls-setup} (with $J=2$) is weakly causal if and only if $\pi_i^*$
is linear in $x_i$. \cref{iv:main} and \cref{item:tsls}(1) shows that $\tau$ is minimally
quasi-experimental if and only if $\pi_i^*$ is linear in $x_i$. It does not need to impose
exogeneity because exogeneity must be satisfied were the regression to estimate treatment
effect contrasts for arbitrary outcomes, per \cref{item:M}.  Compared to
\citet{blandhol2022tsls}, it allows for $J > 2$ but restricts to finitely valued
instruments. \cref{item:tsls}(1) thus clarifies that the equivalence in Theorem 1 of
\citet{blandhol2022tsls} is coming from the equivalence between level independence and
linearity of $\pi_i^*$.

When $\tau$ is minimally quasi-experimental and $J=2$, \cref{iv:main} and \cref{item:tsls}(2)(a) shows
that $\tau$ is furthermore weakly causal if and only if monotonicity holds. Thus,
together, we show the following: Under exogeneity, $\tau$ is weakly causal if and only if
$\pi_i^*$ is linear in $x_i$ and monotonicity holds. This moves monotonicity in
\citet{blandhol2022tsls} from an assumption to an implication.
\subsubsection{A counterexample to \citet{bhuller20242sls}'s Proposition 5}

Following \citet{bhuller20242sls}, we work in a sampling framework. Let $Z$ be an
instrument that takes values $\br{0,1,2}$ and $D$ be a treatment that takes values 
$\br{0,1,2}$. Let $Z_z = \one(Z=z)$ and similarly define $D_1, D_2$. Consider the
two-stage least-squares regression
$
   Y = \alpha + \beta_1 D_1 + \beta_2 D_2 + \epsilon,
$
instrumenting for $[1, D_1, D_2]$ using $[1, Z_1, Z_2]$. Assume that the instrument $Z$
randomly assigned. Let $s (\cdot)$ denote a response
type, mapping from instrument values to treatment values.

Proposition 5 in \citet{bhuller20242sls} states that TSLS assigns proper weights in this
regression if and only if there exists a one-to-one $f: \br{0,1,2} \to \br{0,1,2}$ such
that for all $k \in \br{1,2}$ and response type $s(\cdot)$, either (i) $s(z) \neq k$ for
all $z \in \br{0,1,2}$, (ii) $s(z)=k$ for all $z \in \br{0,1,2}$, or (iii) $s(z) = k \iff
f (z) = k$.

They then conclude that, up to permuting the instrument values, we must then have response
types only being $
    (s(0), s(1), s(2)) \in \br{
        (000), (111),(222), (010), (002), (012)
    }.
$

Unfortunately, this claim is not quite correct as stated. Consider instead three types of
individuals characterized by $ (A,B,C) \equiv (000), (011), (002)$. That is, type $A$
always picks $D=0$. Type $B$ picks $D=1$ when $Z=1$ or $Z=2$. Type $C$ picks $D=2$ only
when $Z=2$. These response types do satisfy common compliance---the only 1-complier is
type $B$ and the only 2-complier is type $C$. Compliance does not satisfy extended
monotonicity: $Z=2$ moves type $C$ to $2$ but type $B$ to $1$. 

Indeed, consider the only if clause of Proposition 5 in \citet{bhuller20242sls}:
\begin{itemize}
    \item  For $k=1$, types $A$ and $C$ satisfy (i). However, type B satisfying (iii)
    necessitates $f (1) = f(2) = 1$. This contradicts $f(\cdot)$ being one-to-one. 
\end{itemize}
Thus, Proposition 5 in \citet{bhuller20242sls} would claim that TSLS does not assign
proper weights. However, in this example, one can show that \[
    \beta_1 = \E[Y(1) - Y(0) \mid B] \quad \beta_2 = \E[Y(2) - Y(0) \mid C].
\]
Thus TSLS does assign proper weights.

To verify this claim, observe that since $Y = D_1 (Y(1) - Y(0)) + D_2 (Y(2)-Y(0)) + Y(0) $
\begin{align*}
& Y - \alpha - \beta_1 D_1 + \beta_2 D_2 %
= Y(0) - \alpha + \begin{cases}
    0, &\text{$Z=0$ or type $A$}\\
    Y(1) - Y(0) - \beta_1  & \text{ $Z=1$, type $B$ }\\
    0 & \text{ $Z=1$, type $C$ } \\ 
     Y(1) - Y(0) - \beta_1  & \text{$Z=2$, type $B$}\\
     Y(2) - Y(0) - \beta_2 & \text{$Z=2$, type $C$}.
\end{cases}.
\end{align*}
Therefore, \begin{align*}
 &\E[Y- \alpha - \beta_1 D_1 + \beta_2 D_2 \mid Z=z] \\
 &= 
 \begin{cases}
        \E[Y(0) - \alpha]& z=0\\
        \E[Y(0) - \alpha] + \P(B)\E[Y(1) - Y(0) - \beta_1 \mid B]& z=1 \\ 
        \E[Y(0) - \alpha] + \P(C)\E[Y(1) - Y(0) - \beta_2 \mid C] + \P(B)\E[Y(1) - Y(0) -
        \beta_1
        \mid B]& z=2.
    \end{cases}
\end{align*}
Thus, $\alpha = \E[Y(0)]$, $\beta_1 = \E[Y(1) - Y(0) \mid B]$ and $\beta_2 = \E[Y(2)-Y
(0)\mid C]$ satisfies the moment condition for the TSLS specification: $\E[\epsilon \mid
Z=z] = 0$.

\subsection{Invariances}
\label{asub:invariances}

\begin{theorem}[Frisch--Waugh--Lovell, in population]
    \label{thm:frisch-waugh}
Consider a population regression specification $(\Lambda, G_n, z_1,\ldots, z_n)$.
Partition $z_t (\bx_i, \bw)$ into $z_{t1}(\cdot)$ and $z_{t2}(\cdot)$. Suppose $\Lambda =
[\Lambda_1, 0]$ loads solely on entries in $z_{t1}$: $\Lambda z_t = \Lambda_1 z_{t1}$. Let
\[\Gamma = \pr{\frac{1}{n}
\sum_{i,t} \sum_{\bw \in \cW} \pi_i^*(\bw) z_{t1}(\bx_i, \bw) z_{t2}(\bx_i, \bw)'}
\pr{\frac{1}{n} \sum_{i,t} \sum_{\bw \in \cW} \pi_i^*(\bw) z_{t2}(\bx_i, \bw)
z_{t2}(\bx_i, \bw)'}^{-1}
\]
 be the population projection matrix of $z_{t1}$ onto $z_{t2}$. Define $\tilde z_{t}
 (\bx_i, \bw) = z_{t1}(\bx_i, \bw) - \Gamma z_{t2}(\bx_i, \bw)$. Consider the regression
 specification defined by $\Lambda_1, \tilde z_t(\cdot),$ and \[
 \tilde G_n = \frac{1}{n} \sum_{i,t} \sum_{\bw \in \cW} \pi_i^*(\bw) \tilde z_t(\bx_i,
 \bw) \tilde z_t(\bx_i, \bw)'.
 \]
 Then:
 \begin{enumerate}[wide, label=(\roman*)]
    \item $\Gamma, \tilde G_n$ are functions of the original Gram matrix $G_n$. Since
    $G_n$ is assumed to be invertible, $\Gamma$ is well-defined and $\tilde G_n$ is
    positive definite.
    \item The potential weights associated with the two population regression specifications
    are the same.
 \end{enumerate}
\end{theorem}

\begin{proof}
Note that we can partition \[
G_n = \frac{1}{n} \sum_{i,t} \sum_{\bw \in \cW} \pi_i^*(\bw) z_t(\bx_i, \bw) z_t(\bx_i, \bw)' = \begin{bmatrix}
    G_{11} & G_{12} \\
    G_{21} & G_{22}
\end{bmatrix}
\]
conformably where $G_{11}$ is the Gram matrix associated with $z_{t1}$ and $G_{22}$ is the
Gram matrix associated with $z_{t2}$. Since $G_n$ is assumed to be invertible and thus
positive definite, $G_{11}, G_{22}$ are invertible and positive definite as well. This
implies that $\Gamma = G_{12} G_{22}^{-1}$ and $\tilde G_n = G_{11} - \Gamma G_{21}$ are
well-defined. This also implies that the Schur complement $\tilde G_n = G_{11} - \Gamma
G_{21}$ is invertible (and positive definite).  This proves (i).

By the block matrix inversion formula, \[
G_n^{-1} = \begin{bmatrix}
(G_{11} - G_{12} G_{22}^{-1} G_{21})^{-1} & -(G_{11} - G_{12} G_{22}^{-1} G_{21})^{-1} G_{12}
G_{22}^{-1} \\
    -G_{22}^{-1} G_{21} (G_{11} - G_{12} G_{22}^{-1} G_{21})^{-1} & G_{22}^{-1} +
    G_{22}^{-1} G_{21} (G_{11} - G_{12} G_{22}^{-1} G_{21})^{-1} G_{12} G_{22}^{-1}
\end{bmatrix}.
\]
Hence \[
    \Lambda G_n^{-1} = \bk{\Lambda_1(G_{11} - \Gamma G_{21})^{-1} \quad -\Lambda_1 (G_{11}
    - \Gamma G_{21})^{-1} \Gamma}
\]
Now, the potential weights for the original regression specification are \begin{align*}
    \brho_{it}(\bw) &= \Lambda_1(G_{11} - \Gamma G_{21})^{-1} z_{t1}(\bx_i, \bw) -\Lambda_1
    (G_{11} - \Gamma G_{21})^{-1} \Gamma z_{t2}(\bx_i, \bw) \\&= \Lambda_1(G_{11} - \Gamma
    G_{21})^{-1} \tilde z_{t}(\bx_i, \bw)  = \Lambda_1 \tilde G_n^{-1} \tilde z_{t}(\bx_i, \bw).
\end{align*}
This proves (ii).
\end{proof}

\begin{theorem}[Invariance under reparameterization]
    \label{thm:reparameterization}
    Consider a population regression specification $(\Lambda, G_n, z_1,\ldots, z_n)$. Consider an invertible matrix $M$ and
    another population specification defined by $\tilde z_t (\bx_i, \bw) = M z_t (\bx_i,
    \bw)$ and $\tilde \Lambda = \Lambda M'$, so that the two vectors of coefficients
    represent the same underlying contrasts. Then the potential weights associated with
    the two regression specifications are the same. The estimated potential weights are
    also the same.

\end{theorem}

\begin{proof}
    We can easily see that the second regression specification has Gram matrix $\tilde
    G_n = M G_n M'$ (and $\widehat {\tilde {G}_n} = M \hat G_n M'$). Thus the potential weights for the second specification are \[
    \tilde \brho_{it}(\bw) = \tilde \Lambda (M')^{-1} G_n^{-1} M^{-1} M z_t(\bx_i, \bw) =
    \brho_{it}(\bw).
    \]
    The corresponding equalities similarly hold for the estimated potential weights.
\end{proof}

\subsection{Additional results for interacted regression}
\label{sub:interact_more}

\Cref{prop:interact_with_t} provides a unified analysis of the regression 
\eqref{eq:interact_general}. It implies our results in \cref{thm:zoo}(\ref
 {item:kline}) when $\bar x_t = \bar x_1$ or $\bar x_0$, but allows for evaluating on
 other points in the line segment connecting $\bar x_1$ and $\bar x_0$. The implicit
 design for the model-based ATE is analyzed when we set $t = \alpha_0 = \frac{1}{n}\sum_
 {i=1}^n \pi_i^*$. 

Inspecting the resulting implicit designs \eqref{eq:implicit_design_interact}, we find
that they are in general fractional linear, encompassing the linear-odds result of $t=1$
as a special case. Moreover, a necessary implication of $\pi_i^* = \pi_i$ is the condition
\eqref{eq:necessary_interact}, which is very difficult to satisfy unless $t=0$ or $t=1$.
In this sense, the model-based ATU and the model-based ATT are the only estimands whose
implicit design is reasonable. See \cref{rmk:ate_wrong} for an example where the implicit
model for the ATE depends on the support of $x_i$. 

\begin{prop}
\label{prop:interact_with_t}
Consider the specification indexed by $t \in \R$, \[
Y_i = \gamma_0 + W_i \tau_t + \gamma_0 + \gamma_1'x_i + W_i (x_i - \bar x_t)'\gamma_2 +
\epsilon_i \numberthis \label{eq:interact_general}
\]
where $\bar x_t = t \bar x_1  + (1-t) \bar x_0$ for $\bar x_1 = \frac{\frac{1}{n} \sum_
{i=1}^n
\pi_i^* x_i}{\frac{1}{n}\sum_{i=1}^n \pi_i^*}$ and $\bar x_0 = \frac{\frac{1}{n} \sum_
{i=1}^n
(1-\pi_i^*) x_i}{\frac{1}{n}\sum_{i=1}^n (1-\pi_i^*)}$.

Let \begin{align*}
A &= \frac{1}{n}\sum_{i=1}^n (x_i-\bar x )(x_i-\bar x)' \\
V_t &= \frac{1}{n}\sum_{i=1}^n \pi_i^* (x_i-\bar x_t )(x_i-\bar x_t)' \quad U_t = 
\frac{1}{n}\sum_{i=1}^n (1-\pi_i^*) (x_i-\bar x_t )(x_i-\bar x_t)' 
\\
\alpha_0 &= \frac{1}{n}\sum_{i=1}^n \pi_i^* \\
\Gamma_{0t} &= \frac{1}{n} \sum_{i=1}^n \pi_i^* (x_i - \bar x_t) \\ 
\alpha_1 &= A^{-1} \frac{1}{n} \sum_ {i=1}^n \pi_i^* (x_i - \bar x) = A^{-1} (\Gamma_{0t}
+ \alpha_0 (\bar x_t - \bar x)) \\
\Gamma_{1t} &= \frac{1}{n} \sum_{i=1}^n \pi_i^* (x_i - \bar x_t)(x_i - \bar x)' = \pr{V_t
+
\Gamma_{0t}(\bar x_t - \bar x)'} A^{-1} \\ 
\Gamma_{2t} &= \pr{\frac{1}{n} \sum_{i=1}^n \E\bk{
    W_i(x-\bar x_t) \pr{W_i (x-\bar x_t) - \Gamma_{0t}- \Gamma_{1t} (x_i-\bar x)}'
}}^{-1}  \\
& \quad \quad \times \frac{1}{n} \sum_{i=1}^n\E\bk{
    (W_i - \alpha_0 - \alpha_1'(x_i-\bar x)) \pr{W_i (x_i-\bar x_t) - \Gamma_{0t} -
    \Gamma_
    {1t} (x-\bar x)}
} \\ 
&= \pr{V_t - \Gamma_{0t}\Gamma_{0t}' - V_t \Gamma_{1t}'  - \Gamma_{0t}(\bar x_t - \bar
x)'\Gamma_{1t}'}^{-1} \pr{
    \Gamma_{0t} - \alpha_0 \Gamma_{0t} - V_t \alpha_1 - \Gamma_{0t} (\bar x_t - \bar x)'
    \alpha_1
}
\end{align*} 

Then, the potential weights satisfy \[
    \rho_i(w) = \tilde G_n^{-1} \pr{
        w (1-\Gamma_{2t}'(x_i - \bar x_t)) - \underbrace{(\alpha_0 - \Gamma_{2t}'\Gamma_
{0t})}_{\theta_{0t}} -
\underbrace{(\alpha_1 - \Gamma_{1t}'\Gamma_{2t})'}_{\theta_{1t}'}(x_i - \bar x)
    },
\]
for some scalar $\tilde G_n > 0$. The implicit design $\pi_i$ satisfies \[
\pr{1-\Gamma_{2t}'(x_i - \bar x_t)} \pi_i = \theta_{0t} + \theta_{1t}'(x_i-\bar x). 
\numberthis 
\label{eq:implicit_design_interact}
\]

Moreover, if $\pi_i^* = \pi_i$, then a necessary condition is that \[
t V_1 (\theta_{1t} + \Gamma_{2t}) = (1-t) U_0 \theta_{1t}.
\numberthis
\label{eq:necessary_interact}
\]
\end{prop}

\begin{proof}
The FWL residual of $W_i$ after partialling out the intercept, $x_i-\bar x$,
and $W_i(x_i-\bar x_t)$ is
\[
R_i(w)
=
w\{1-\Gamma_{2t}'(x_i-\bar x_t)\}
-\theta_{0t}-\theta_{1t}'(x_i-\bar x).
\]
Thus, by \cref{thm:frisch-waugh},
\[
    \rho_i(w)=\tilde G_n^{-1}R_i(w)
\]
for some $\tilde G_n > 0$ from FWL. This is the displayed formula for the
potential weights. The implicit-design equation is
\[
    0=\pi_i\rho_i(1)+(1-\pi_i)\rho_i(0),
\]
which is equivalent to \eqref{eq:implicit_design_interact}.

Now suppose $\pi_i^*=\pi_i$. Since $R_i(W_i)$ is a regression residual, it is orthogonal to
the intercept, $x_i-\bar x$, and $W_i(x_i-\bar x_t)$. Because $x_i-\bar x_t$ is a linear
combination of the intercept and $x_i-\bar x$, it follows that $R_i(W_i)$ is orthogonal to
both $x_i-\bar x_t$ and $W_i(x_i-\bar x_t)$, so
\[
    0
    =
    \frac{1}{n}\sum_{i=1}^n
    \E\bk{(1-W_i)(x_i-\bar x_t)R_i(W_i)}
    =
    -\frac{1}{n}\sum_{i=1}^n
    (1-\pi_i^*)(x_i-\bar x_t)
    \{\theta_{0t}+\theta_{1t}'(x_i-\bar x)\}.
\]
Using \eqref{eq:implicit_design_interact} with $\pi_i=\pi_i^*$ gives
\[
    \theta_{0t}+\theta_{1t}'(x_i-\bar x)
    =
    \pi_i^*\{1-\Gamma_{2t}'(x_i-\bar x_t)\},
\]
and hence
\[
    0=
    \frac{1}{n}\sum_{i=1}^n
    \pi_i^*(1-\pi_i^*)(x_i-\bar x_t)
    \{1-\Gamma_{2t}'(x_i-\bar x_t)\}.
\]
Finally, write $x_i-\bar x_t=t(x_i-\bar x_1)+(1-t)(x_i-\bar x_0)$. For the first term,
substitute
\[
    (1-\pi_i^*)\{1-\Gamma_{2t}'(x_i-\bar x_t)\}
    =
    1-\Gamma_{2t}'(x_i-\bar x_t)
    -\theta_{0t}-\theta_{1t}'(x_i-\bar x);
\]
constants drop out because $\sum_i\pi_i^*(x_i-\bar x_1)=0$, and the remaining cross-products
are both $V_1$. For the second term, substitute
\[
    \pi_i^*\{1-\Gamma_{2t}'(x_i-\bar x_t)\}
    =
    \theta_{0t}+\theta_{1t}'(x_i-\bar x);
\]
constants drop out because $\sum_i(1-\pi_i^*)(x_i-\bar x_0)=0$, and the remaining
cross-product is $U_0$. Therefore
\[
    0=-tV_1(\theta_{1t}+\Gamma_{2t})+(1-t)U_0\theta_{1t},
\]
which is \eqref{eq:necessary_interact}.
\end{proof}

\begin{prop}
\label{prop:kline_prop}
In \cref{prop:interact_with_t}, for $t = 1$, the implicit design is equal to \[
\pi_i = \frac{\alpha_0 + ((I-\Gamma_1')^{-1} \alpha_1)' (x-\bar x)}{1 + ((I-\Gamma_1')^
{-1} \alpha_1)'(x-\bar x_1)}.
\]
The associated odds is \[
\frac{\pi_i}{1-\pi_i} = \frac{\alpha_0 + ((I-\Gamma_1')^{-1} \alpha_1)' (x-\bar x)}{
    1-\alpha_0 + ((I-\Gamma_1')^{-1} \alpha_1)' (\bar x - \bar x_1)
} \equiv \delta_0 + \delta_1'(x_i-\bar x_0).
\]
Here, \begin{align*}
\delta_1 &= \pr{\frac{1}{n} \sum_{i=1}^n (1-\pi_i^*) (x_i - \bar x_0)(x_i-\bar x_0)}^{-1}
\frac{1}{n} \sum_{i=1}^n (1-\pi_i^*) (x_i - \bar x_0) \frac{\pi_i^*}{1-\pi_i^*}
\\
\delta_0 &= \frac{\alpha_0}{1-\alpha_0} =  \frac{1}{n} \sum_{i=1}^n \frac{(1-\pi_i^*)}
{1-\alpha_0} \frac{\pi_i^*} {1-\pi_i^*}
\end{align*}
\end{prop}
\begin{proof}
We omit the $t=1$ subscript. Since $\Gamma_0=0$ by construction, the definition of
$\Gamma_2$
gives
\[
    \Gamma_2=-(I-\Gamma_1')^{-1}\alpha_1,
    \quad
    \theta_0=\alpha_0,
    \quad
    \theta_1=\alpha_1-\Gamma_1'\Gamma_2=(I-\Gamma_1')^{-1}\alpha_1.
\]
\Cref{prop:interact_with_t} then gives
\[
    \pi_i
    =
    \frac{\alpha_0+\theta_1'(x_i-\bar x)}
    {1+\theta_1'(x_i-\bar x_1)}
\]
and therefore
\[
    \frac{\pi_i}{1-\pi_i}
    =
    \frac{\alpha_0+\theta_1'(x_i-\bar x)}
    {1-\alpha_0+\theta_1'(\bar x-\bar x_1)}.
\]

It remains to identify the intercept and slope in this last expression. Since
$\Gamma_1=V_1A^{-1}$ and
$\alpha_1=A^{-1}\alpha_0(\bar x_1-\bar x)$,
\[
    (A-V_1)\theta_1=\alpha_0(1-\alpha_0)(\bar x_1-\bar x_0).
\]
Using $\bar x=\alpha_0\bar x_1+(1-\alpha_0)\bar x_0$,
\[
    A-V_1
    =
    U_0+\alpha_0(1-\alpha_0)(\bar x_1-\bar x_0)(\bar x_1-\bar x_0)'.
\]
Hence
\[
    U_0\theta_1
    =
    \alpha_0(1-\alpha_0)
    \{1-\theta_1'(\bar x_1-\bar x_0)\}(\bar x_1-\bar x_0).
\]
Also,
\[
    1-\alpha_0+\theta_1'(\bar x-\bar x_1)
    =
    (1-\alpha_0)\{1-\theta_1'(\bar x_1-\bar x_0)\}.
\]
Dividing the preceding two displays gives
\[
    \frac{\theta_1}{1-\alpha_0+\theta_1'(\bar x-\bar x_1)}
    =
    U_0^{-1}\alpha_0(\bar x_1-\bar x_0).
\]
Because
\[
    \frac{1}{n}\sum_{i=1}^n
    (1-\pi_i^*)(x_i-\bar x_0)\frac{\pi_i^*}{1-\pi_i^*}
    =
    \frac{1}{n}\sum_{i=1}^n\pi_i^*(x_i-\bar x_0)
    =
    \alpha_0(\bar x_1-\bar x_0),
\]
this is the displayed formula for $\delta_1$.

Finally,
\[
    \frac{\alpha_0+\theta_1'(\bar x_0-\bar x)}
    {1-\alpha_0+\theta_1'(\bar x-\bar x_1)}
    =
    \frac{\alpha_0\{1-\theta_1'(\bar x_1-\bar x_0)\}}
    {(1-\alpha_0)\{1-\theta_1'(\bar x_1-\bar x_0)\}}
    =
    \frac{\alpha_0}{1-\alpha_0}.
\]
This is the displayed $\delta_0$, since
\[
    \frac{\alpha_0}{1-\alpha_0}
    =
    \frac{1}{n}\sum_{i=1}^n
    \frac{1-\pi_i^*}{1-\alpha_0}\frac{\pi_i^*}{1-\pi_i^*}.
\]
 \end{proof}

\begin{rmksq}
\label{rmk:ate_wrong}
Here we consider the ramifications of \eqref{eq:necessary_interact} for the model-based
ATE ($\tau_{\alpha_0}$). For some true design $\bpi^*$ to justify causal interpretation,
we must have that this design is fractional linear: \[
    \pi_i^* = \frac{\theta_0 + \theta_1'(x_i - \bar x)}{1-\Gamma_2'(x_i - \bar x)}
\]
for some $\theta_0, \theta_1, \Gamma_2$. Moreover, the choices $\theta_0, \theta_1,
\Gamma_2$ needs to satisfy \eqref{eq:necessary_interact}, which depends on the covariates
$x_1,\ldots, x_n$. We may thus consider the parameter space as \[
   \Theta(t; x_{1:n}) = \br{(\theta_0, \theta_1, \Gamma_2) : \text{ $\tau_{t}$
   is a causal contrast under the corresponding
    $\bpi^*$ under $(\theta_0, \theta_1, \Gamma_2)$ }}.
\]
Note that for $t=0$ and $t=1$, this parameter space does not depend on $x_{1:n}$ aside
from the obvious support restriction. For instance, if $x_{1:n}$ and $\tilde x_{1:n}$ are
two sets of covariate values, and \[
    \pi_i^*/(1-\pi_i^*) = \delta_0 + \delta_1' (x_i - \bar x_1) \ge 0 \quad \tilde\pi_i^*/
    (1-\tilde\pi_i^*) = \delta_0 + \delta_1' (\tilde x_i - \bar {\tilde x}_1)
    \ge 0,
\]
then the corresponding parameters are in $\Theta(1; x_{1:n}) \cap \Theta(1; \tilde x_
{1:n})$---similarly for $t=0$. This is true by \cref{prop:kline_prop}.  However, for other
$t$, and in particular for $t = \alpha_0$, it is possible that some
$(\theta_0,\theta_1,\Gamma_2 )\in \Theta(\alpha_0; x_{1:n})$, still does not belong to
$\Theta(\alpha_0; \tilde x_{1:n})$,  even though no support restriction is violated \[
    \tilde \pi_i^* \equiv \frac{\theta_0 + \theta_1'(\tilde x_i - \bar{\tilde x})}
    {1-\Gamma_2'(\tilde x_i - \bar {\tilde x})} \in [0,1] \text{ for all $i$}.
\]
It is not very difficult to come up with a numerical example for such a case. %
\end{rmksq}

 \begin{restatable}[No simple regression estimates the ATE under linear design]{prop}
 {propnoatereg}
\label{prop:no-ate-regression}
Let $n \ge 3$. Let $W_i \in \br{0,1}$, covariates $x_i \in \R^d$, and $y_i(\cdot) \in
\R$. Suppose the true design is linear $\pi_i^* = \delta_0 + \delta_1'x_i$ for some
$\delta_0
\in
\R, \delta_1
\in \R^d$. There is no regression $ (\Lambda, z(x, w))$---where $\Lambda$ may\footnote{This 
is to accommodate for estimands like the model-based ATT, where we may consider contrasts
that depend on $\bar x_1 = \sum_i \pi_i^* x_i/\sum_i \pi_i^*$} depend on
$x_{1:n}, \pi^*_{1:n}$---such that:
\begin{enumerate}
    \item (Regression is linear in covariates) For all $m$ and all $w$, the $m$\th{} entry
    of $z (x_i, w)$ is of the form $a_m(w) + b_m(w)'x_i$ for some fixed conformable $a_m
    (\cdot), b_m(\cdot)$.
    \item ($\Lambda\beta$ is the ATE) The corresponding estimand $\Lambda\beta \in \R$ is
    equal to the ATE, regardless of the configuration of $d, x_{1:n}, \delta_0, \delta_1$
    (such that $\pi_i^* \in [0,1]$ for all $i$).
\end{enumerate}
\end{restatable}

\begin{proof}[Proof of \cref{prop:no-ate-regression}]
It suffices to exhibit one linear design under which no such regression can
estimate the ATE. Take scalar covariates with support containing $-1,0,1$ and contained
in $[0,1]$. Set
$\pi_i^*=1/2+x_i/4$, so $\pi_i^*\in(0,1)$. Since every entry of $z(x_i,w)$ is affine in
$x_i$, the potential weights are affine in $x_i$:
\[
    \rho_i(w)=\Lambda G_n^{-1}z(x_i,w)=c_{0w}+c_{1w}x_i .
\]
If the regression estimated the ATE, then level irrelevance and the
requirement that every treated-arm ATE weight equals one would imply
\[
    \pi_i^*\rho_i(1)+(1-\pi_i^*)\rho_i(0)=0,
    \qquad
    \pi_i^*\rho_i(1)=1 .
\]
Eliminating $\pi_i^*$ gives
\[
    \rho_i(1)\rho_i(0)=\rho_i(1)+\rho_i(0).
\]
It is not difficult to show that substituting $\rho_i(w) = c_{0w} + c_{1w}x_i$ and
evaluating at $x=-1,0,1$ implies $c_{10}=c_{11}=0$. Thus both potential weights are
constant in $x$. But then
$    \pi_i^*\rho(1)+(1-\pi_i^*)\rho(0)=0 $
cannot hold for three distinct propensity-score values unless
$\rho(1)=\rho(0)=0$, contradicting $\pi_i^*\rho_i(1)=1$.
\end{proof}

\subsection{Additional results for panel specifications}
\label{sub:additional_panel}

\begin{prop}[One-way fixed effects]
    \label{prop:owfe}
    Consider the regression $Y_{it} = \alpha_i + \tau W_{it} + \epsilon_{it}$. If the span
    of $\cW \subset \br{0,1}^T$ excludes $1_T$, then the set of implicit designs is empty.
\end{prop}

\begin{proof}
    Let $\dot{\bw} = \bw - 1_T \frac{1_T'\bw}{T}$.  Then, by \cref{thm:frisch-waugh}, the
    potential weights are equal to \[
    \rho_i(w) = \frac{\dot\bw}{\frac{1}{n} \sum_{i=1}^n \sum_{\bw\in \cW} \pi_i^*(\bw) \dot \bw' \dot \bw}.
    \]
    Thus, if $\pi_i$ corresponds to some implicit design $\bpi$, \[
    \sum_{\bw \in \cW}\pi_i(\bw) \pr{\bw - 1_T \frac{1_T'\bw}{T}} = 0.
    \]
    In particular, this implies that $1_T$ is in the linear span of $\cW$. By assumption,
    this is not the case. Therefore the set of implicit designs must be empty.
\end{proof}

Intuitively, this is because one-way fixed effects fail to account for the case where the
treatment path correlates with underlying potential outcomes. If there are no treatment
effects whatsoever, but the baseline potential outcome correlates with the treatment path
(e.g., later potential outcomes tend to be larger in staggered adoption), then one-way
fixed effects would in general estimate a nonzero coefficient, violating level
irrelevance.
{These results contrast with the results in
\citet{arkhangelsky2023fixed}, where the one-way fixed effect regression admits a
design-based interpretation under within-unit random assignment of $W_{it}$. To reconcile,
\citet{arkhangelsky2023fixed} consider a sampling-based setup and the unconfoundedness
restriction \[ W_{it} \indep (Y_{it}(0), Y_{it}(1)) \mid \frac{1}{T}\sum_{s=1}^T W_{is},
\]
which is their (2.2) in our notation. Under staggered adoption, for instance, the
associated propensity score $
\P\pr{W_{it}=1 \mid \frac{1}{T}\sum_{s=1}^T W_{is}}
$ is always degenerate, since $\frac{1}{T}\sum_{s=1}^T W_{is}$ perfectly distinguishes
which treatment path unit $i$ is assigned. Conversely, if $W_{it}$ is randomly assigned
within a unit and that permuting the time index results in valid counterfactual
assignments, then $\cW$ is large enough to contain $1_T$.
}

\begin{prop}[Event-study and more complex TWFE]
\label{prop:twfe}
Consider the regression $Y_{it} = \alpha_i + \gamma_t + \beta' f_{t}(\bW_i)$ with target
estimand $\btau = \Lambda\beta$ under some design $\bpi^*$. We have:
\begin{enumerate}[wide, label=(\roman*)]
    \item Whether or not $\btau$ satisfies level
irrelevance under the true design $\bpi^*$, one proper and Gram-consistent implicit design
is \[\pi_i (\bw) \equiv \pi(\bw) \equiv \frac{1}{n}\sum_
{j=1}^n \pi_j^* (\bw). \numberthis \label{eq:panel_random_assignment}\]
\item Suppose $\Lambda = I$. If any column $k$ of $f_t(\bW_i)$ is such that the multiset 
\[\pr{\colvecb{3}{f_{1k}(\bw)}{\vdots}{f_{Tk}(\bw)} : \bw \in \cW}
\]
(a) contains the zero vector at most once, (b) has nonzero vectors that are linearly
independent, and (c) does not span the subspace $\br{c1_T : c\in \R}$, then $\bar\pi$ is
also the unique implicit design.
\end{enumerate}
\end{prop}

\begin{proof}
\begin{enumerate}[wide, label=(\roman*)]
    \item By \cref{thm:frisch-waugh}, it suffices to compute potential weights for the two-way
residualized specification. That is, let \[
\bz_i(\bw) = \colvecb{3}{f_{1}(\bw)'}{\vdots}{f_{T}(\bw)'} \equiv \bz(\bw)
\]
be the covariate transform. Note that the population residual of projecting $\bz_i(\bW_i)$
on unit and time fixed effect is \begin{align*}
    \ddot{\bz}(\bw) &= \bz(\bw) - 1_T \frac{1_T'}{T} \bz(\bw) - \frac{1}{n} \sum_{i=1}^n
\sum_{\bw' \in \cW} \pi_i^*(\bw') \bz(\bw') + 1_T \frac{1_T'}{T} \frac{1}{n} \sum_{i=1}^n
\sum_{\bw' \in \cW} \pi_i^*(\bw') \bz(\bw') \\
&= \bz(\bw) - 1_T \frac{1_T'}{T} \bz(\bw)  -
\sum_{\bw' \in \cW} \pi(\bw')\bz(\bw') + 1_T \frac{1_T'}{T}
\sum_{\bw' \in \cW}  \pi(\bw') \bz(\bw')
\end{align*}
Thus the potential weights are $
\brho_i(\bw) =  \Lambda \pr{\frac{1}{n} \sum_{i=1}^n
\E[\ddot{\bz}(\bw)'\ddot{\bz}(\bw)]}^{-1} \ddot{\bz}(\bw)'.
$
Note that \[
\sum_{\bw \in \cW}\pi(\bw) \brho_i(\bw)  = \Lambda \pr{\frac{1}{n} \sum_{i=1}^n
\E[\ddot{\bz}(\bw)'\ddot{\bz}(\bw)]}^{-1} \underbrace{\pr{\sum_{\bw \in \cW}\pi(\bw)
\ddot{\bz}(\bw)'}}_{ = 0}
= 0.
\]
Therefore $\bpi$ defined by $\pi_i(\bw) = \pi(\bw)$ is a valid implicit design. It is proper by
definition. Now, let $\dot{\bz}(\bw)$ be the within-transformed covariate transform for this regression (which includes the time
fixed effects), which does not
depend on $i$. Note that the Gram matrix is \[
G_n (\bpi^*) = \frac{1}{n} \sum_{i=1}^n \E[ \dot{\bz}(\bw)' \dot{\bz}(\bw) ] = \sum_{\bw \in \cW}
\pi(\bw)  \dot{\bz}(\bw)' \dot{\bz}(\bw) = G_n(\bpi).
\]
Thus $\bpi$ is Gram-consistent.

\item In this case, any implicit design $\tilde \pi_i$ needs to satisfy\[
    0 = \sum_{\bw\in \cW} \tilde \pi(\bw) \brho_i(\bw) = \pr{\frac{1}{n} \sum_{i=1}^n
\E[\ddot{\bz}(\bW)'\ddot{\bz}(\bW)]}^{-1} \pr{\sum_{\bw \in \cW} \tilde \pi_i(\bw) 
\ddot{\bz}(\bw)'}
\]
Since $\Lambda = I$ and $\frac{1}{n} \sum_{i=1}^n
\E[\ddot{\bz}(\bW)'\ddot{\bz}(\bW)]$ is positive definite by \cref{thm:frisch-waugh}, this
implies that \[
\sum_{\bw \in \cW} \tilde \pi_i(\bw) 
\ddot{\bz}(\bw)' = 0
\]
Inspecting this expression, we have \[
    \sum_{\bw \in \cW} (\tilde \pi_i(\bw) - \pi(\bw)) \bz(\bw) = 1_T\sum_{\bw \in
    \cW} (\tilde \pi_i(\bw) - \pi(\bw)) \frac{1_T'\bz (\bw)}{T}
\]
Column $k$ of this expression is in turn \[
    \sum_{\bw \in \cW} (\tilde \pi_i(\bw) - \pi(\bw)) \colvecb{3}{f_{1k}(\bw)}{\vdots}{f_
    {Tk}(\bw)} = 1_T C_k
\]
for some scalar $C_k$. The left hand side is in the span of $(f_{1k}(w) ,\ldots, f_{Tk}
(w))'$, which excludes $1_T$ by assumption. Thus,  \[
    \sum_{\bw \in \cW} (\tilde \pi_i(\bw) - \pi(\bw)) \colvecb{3}{f_{1k}(\bw)}{\vdots}{f_
    {Tk}(\bw)} = 0.
\]
Since at most only one $\bw$ leads to $(f_{1k}(w) ,\ldots, f_{Tk}(w))' = 0$ and nonzero
vectors $(f_{1k}(w) ,\ldots, f_{Tk}(w))'$ are linearly independent, we have that $
    \tilde \pi_i(\bw) = \pi(\bw)
$
for all but one $\bw$ by linear independence. Since both designs sum to 1, they also equal
for remaining $\bw$. 
\end{enumerate}
\end{proof}

\begin{prop}[Forbidden comparisons]
    \label{prop:forbidden_comparisons}
    Consider the TWFE specification $Y_{it} = \alpha_i + \gamma_t + \tau W_{it}$ under
    staggered adoption ($\bw_t \le \bw_{t+1}$ for all $\bw \in \cW$ and all $t \in [T]$).
    Suppose the treatment time is randomly assigned so that
    $\pi_i^*(\bw) = \pi^*(\bw)$. If there are two treatment paths with one being
    never-treated, $\bW = \br{0, \bw}$, then $\brho_{it}(\bw) \ge 0$ for all post
    treatment periods $\bw_t = 1$ and $i \in [n]$. Otherwise, in all other configurations
    where $|\cW| \ge 2$, there exists a choice of $\pi^*(\bw)$ such that the corresponding
    $\brho_{it}(\bw) < 0$ for some treatment path $\bw$ with positive assignment
    probability ($\pi^*(\bw ) > 0$) and some post-treatment period $t$ ($\bw_t = 1$).
\end{prop}

\begin{proof}
    The potential weight can be computed in closed form. For some $V > 0$, \[
    \brho_{it}(\bw) = V^{-1}\bk{\bw_t - \frac{1'}{T}\bw - \sum_{\tilde \bw \in \cW} \pi^*(\tilde
    \bw) \tilde \bw_t + \sum_{\tilde \bw \in \cW} \pi^*(\tilde \bw) \frac{1'}{T}\tilde
    \bw}.
    \]
    When $\cW = \br{0, \bw}$, then for a post-treatment $t$, \[
    \brho_{it}(\bw) = V^{-1} (1- \pi^*(\bw))(1-1'\bw/T) \ge 0.
    \]

    Otherwise, let $\bw$ be the treatment path with the earliest adoption date, and
    consider $t= T$. By assumption, $\cW$ contains a path that adopts later than $\bw$. Then \[
    \brho_{it}(\bw) = V^{-1} \pr{
        (1- \pi^*(\bw)) (1-1'\bw/T) - \sum_{\tilde \bw \neq \bw} \pi^*(\tilde \bw) \pr{1-1'\tilde \bw/T}
    }
    \]
    Pick $\pi^*(\bw) = 1/2 = \pi^*(\tilde \bw)$ where $\tilde \bw$ adopts later than
    $\bw$. Then \[
        \brho_{it}(\bw) = V^{-1} \pr{
           1/2 (1-1'\bw/T) - 1/2 \pr{1-1'\tilde \bw/T}
        } = \frac{1}{2} V^{-1} 1'(\tilde \bw - \bw)/T < 0
    \]
    since $\tilde \bw$ adopts later than $\bw$.
\end{proof}

\begin{prop}[Time-varying covariates by interacting unit-specific covariates with time
fixed effect]
\label{prop:unit-covariate-interact}

Suppose $\bW$ satisfies staggered adoption and excludes an always-treated unit. Consider
the TWFE specification $
    \bY_{it} = \alpha_i + \gamma_t + \tau \bW_{it} + x_i'\eta_t.
$
$\tau$ is a causal contrast only if there is some $\delta_0(\cdot), \delta_1(\cdot)$ such
that $
    \pi_i^*(\bw) = \delta_{0}(\bw) + \delta_1(\bw)' x_i
$
for all $i \in [n], \bw \in \cW$.
    
\end{prop}

\begin{proof} 

Let $\bw(t)$ be the path in $\cW$ that is first treated in time $t$. The
 never-treated unit is denoted by $\bw(\infty)$. By excluding the always-treated path,
 $\bw(t) \in \cW$ implies that $t \ge 2$. 

 By the same argument leading to \eqref{eq:level_irrelevance_panel_with_covariates}, we
 have that if $\tau$ is a causal contrast, then for some $\rho_i, \lambda_t, \kappa_t$,
 for all $t \in [T]$ and $i \in [n]$,
 \[
     \sum_{\bw \in \cW} \pi_i^*(\bw) \bw_t = \rho_i + \lambda_t + \kappa_t'x_i
 \]
 For $s \ge 2$ a treatment time, consider taking the difference of the above equation
 between $s$ and $s-1$. Because  \[\bw_s - \bw_{s-1} = \begin{cases}
     1, &\text{ if } \bw = \bw(s)\\
     0 & \text{otherwise},
 \end{cases}\]we have $
     \pi_i^*(\bw(s)) = \lambda_s - \lambda_{s-1} + (\kappa_{s} - \kappa_{s-1})'\bx_i
     \equiv
     \delta_{0}(\bw(s)) + \delta_1(\bw(s))' x_i.
 $
 This shows the desired representation for all $i \in [n]$ and all $0\neq \bw \in \cW$.
 Lastly, $
     \pi_i^*(0) = 1 - \sum_{\bw \in \cW\setminus\br{0}} \pi_i(\bw)
 $ can also be written as $\delta_{0}(0) + \delta_1(0)' x_i$.
\end{proof}

\subsubsection{Imbalanced panels}
\label{asub:imbalance}

To introduce our result on imbalanced panels, suppose units are observed in some subperiod
$\cT_i \subset
\br{1,\ldots, T}$. We consider the TWFE specification \[ \bY_{it} = \alpha_i + \gamma_t +
\tau \bW_{it} + \epsilon_{it}
\numberthis
\label{eq:unbalanced_twfe}
\]
over $i =1,\ldots, n$ and $t \in \cT_i$. The objects $\cT_1,\ldots, \cT_n$ are considered
fixed and observed. 

Let $\cW_{\cT}$ be a multiset collecting the non-zero treatment paths restricted to some
subperiod $\cT \subset [T]$. That is, $\cW_{\cT}$ collects the vectors $(\bw_t: t \in
\cT)$ for all $\bw \in \cW \setminus \br{0_T}$, possibly with repetition. We say that
$\cW$ has \emph{rich variation} on $\cT$ if the vectors in $\cW_\cT$ are linearly
independent and their linear span excludes the vector of all ones.\footnote{As an example,
note that if $\cW$ satisfies staggered adoption and excludes an always treated unit, then
the period that spans all adoption dates, $\cT = \br{t_{\min} - 1,\ldots,
t_{\max}}$---where $t_{\min}$ is the first adoption date and $t_{\max}$ is the last
adoption date, is a subperiod on which $\cW$ has rich variation.}

\begin{restatable}{prop}{propunbalanced}
    \label{prop:unbalanced}
    Let the population regression specification be characterized by
    \eqref{eq:unbalanced_twfe} under $\bpi^*$, where $\tau$ is the coefficient of
    interest. Suppose $\cW$ satisfies staggered adoption and excludes always-treated
    units. Assume further that there is a commonly observed period $\cT \subset
    \bigcap_{i=1}^n \cT_i$
    such that the treatment paths have rich variation in $\cT$.

    Let $\bar \pi(\bw) \equiv \frac{1}{n} \sum_{i=1}^n \pi_i^*(\bw)$. Let $Q_i(\bw) =
    \frac{\sum_{t\in \cT_i} \bw_t}{|\cT_i|}$ be the proportion of treated periods within
    unit $i$'s observed period for treatment path $\bw$.
    Then an implicit design exists if and only if $\bpi^*$ is uncorrelated with the
    missingness pattern in the sense that for all $t \in [T]$: \[
    \underbrace{\sum_{\bw \in \cW} \bar \pi(\bw) \sum_{i: t \in \cT_i} (\bw_t -  Q_i
    (\bw))}_{\sum_{i: t \in \cT_i} \E_{\bw \sim \bar\pi}[(\bw_t -  Q_i
    (\bw))]} = \underbrace{\sum_ {\bw
        \in \cW} \sum_{i: t \in \cT_i}  \pi_i^*(\bw) (\bw_t - Q_i(\bw))}_{\sum_{i: t \in
        \cT_i} \E_{\bw \sim \pi_i^*}[(\bw_t -  Q_i
    (\bw))]}.
    \numberthis
    \label{eq:unbalanced_condition}
    \]
    When this happens, the implicit design is unique and satisfies $\pi_i(\bw) = \bar
    \pi(\bw)$ for all $i$.

\end{restatable}

\Cref{prop:unbalanced} gives a necessary and sufficient condition for an implicit design
to exist, and characterizes the implicit design when it does. In short, an implicit design
exists if and only if the missingness patterns are uncorrelated with the treatment timing
in a particular sense \eqref{eq:unbalanced_condition}. When it exists, the implicit design
is the same for all units. Therefore, if $\bpi^*$ justifies $\tau$ as a design-based
estimand in the sense of \cref{defn:level_irrelevance}, then $\bpi^*$ must equal to the
implicit design and thus must describe random assignment.

The condition \eqref{eq:unbalanced_condition} describes a situation where for any given
time period $t$, among units for whom $t$ is observed, the average expected de-meaned
treatment path under $\bpi^*$, $\E_{\bw \sim \pi_i^*}[\bw_t - Q_i(\bw)]$, is equal to its
analogue under the random assignment $\bar\pi$. This condition is difficult to satisfy
unless the design is constant $\pi_i^* (\bw) = \bar \pi(\bw)$ in the first place. As a
result, \cref{prop:unbalanced} shows that TWFE is fragile to imbalanced panels.

\begin{proof}[Proof of \cref{prop:unbalanced}]
    For a given unit $i$ and $t \in \cT_i$, let $L_i = |\cT_i|$ and let $R_i(\bw) =
    \sum_{s \in \cT_i} \bw_s$ so that $Q_i(\bw) = L_i^{-1} R_i(\bw)$. The covariate
    transform---where we
    specify the unit fixed
    effect via within transformation---is \[z_{it}(\bw)' =
    [\dot{\one}_{it1},\ldots,\dot{\one}_{itT}, \dot \bw_{t}] \text{ where } \dot{\one}_
    {its} = \begin{cases}
        0 & s \not \in \cT_i \\
        1-\frac{1}{L_i} & s = t \\
        -\frac{1}{L_i} & s \in \cT_i \setminus \br{t}
    \end{cases}\]
    is the demeaned time dummy and $
    \dot{\bw}_{t} = \bw_{t} - Q_i(\bw).
    $

    By \cref{thm:frisch-waugh}, it suffices to analyze the potential weights of regressing
    $Y_{it}$ on $\ddot{W}_{it}$, where \[
    \ddot{W}_{it} = \dot{W_{it}} - \sum_{s = 1}^T \delta_s \dot{\one}_{its} = \dot{W_{it}}
    - \sum_{s \in \cT_i} \delta_s \dot{\one}_{its}.
    \]
    and $\delta_s$ are the population projection coefficients of $\dot{W}_{it}$ on
    $\dot{\one}_{its}$. Let $V \equiv \frac{1}{n}\sum_{i=1}^n \sum_{t \in \cT_i}
    \E[\ddot{W}_{it}^2] $, which is strictly positive by \cref{thm:frisch-waugh}. The potential weights are then \[
    \brho_{it}(\bw) = V^{-1} \pr{\dot{\bw}_t - \sum_{s \in \cT_i} \delta_s
    \dot{\one}_{its}} = V^{-1} \pr{\dot{\bw}_t - \delta_t + \frac{1}{L_i} \sum_{s\in
    \cT_i} \delta_s}.
    \]

    An implicit design $\bpi$ exists if and only if there is some $\pi_i(\cdot)$ and
    $\sum_{\bw} \pi_i(\bw) =1 $ such that \[
        \sum_{\bw \in \cW} \pi_i(\bw) \brho_{it}(\bw) = 0 \text{ for all $i$ and $t\in \cT_i$}.
    \]
    This is further equivalent to that for all $i$, $t\in \cT_i$, \[
    \sum_{\bw \in \cW} \pi_i(\bw) \pr{\bw_t - \frac{R_i(\bw)}{L_i}} = \delta_t - \frac{1}{L_i} \sum_{s\in
    \cT_i} \delta_s.
    \numberthis
    \label{eq:unbalanced_level_condition_it}
\]

Suppose first that such a $\bpi$ exists. Fix $t \in \cT$ and consider two units $i,j$,
where $\cT_i \cap \cT_j \supset \cT$ by assumption.
Then, by \eqref{eq:unbalanced_level_condition_it}, we have that
\[
\delta_t = \sum_{\bw \in \cW} \pi_i(\bw) \bw_t - \sum_{\bw \in \cW} \pi_i(\bw) \frac{R_i}{L_i} + \frac{1}{L_i} \sum_{s\in \cT_i} \delta_s = \sum_{\bw \in \cW} \pi_j(\bw) \bw_t - \sum_{\bw \in \cW} \pi_j(\bw) \frac{R_j}{L_j} + \frac{1}{L_j} \sum_{s\in \cT_j} \delta_s.
\]
Thus, rearranging, we have for all $t \in \cT$,   \[
\sum_{\bw \in \cW, \bw \neq 0} (\pi_i(\bw) - \pi_j(\bw)) \bw_t = M(i,j)
\]
where $M(i,j)$ does not depend on $t$. Viewed as linear combinations for vectors in
$\cW_\cT$, the left-hand side collecting over all $t \in \cT$ lies in the span of
$\cW_\cT$. The right-hand side lies in the span of the constant vector on $|\cT|$
dimensions. By assumption, these two spans contain only the zero vector in common. Thus,
both sides are equal to zero: \[
    \sum_{\bw \in \cW, \bw \neq 0} (\pi_i(\bw) - \pi_j(\bw)) \bw_t = 0 \text{ for all $t \in \cT$}.
\]
By linear independence of vectors in $\cW_\cT$, we conclude that $\pi_i(\bw) =
\pi_j(\bw)$. Therefore, the implicit design must be constant across units: \[
\pi_i(\bw) = \pi(\bw)
\]
for some $\pi(\cdot)$.

Next, we show that if $\bpi$ exists, then it must equal to $\bar\pi(\cdot)$. The fact that
$\delta_t$ are projection coefficients is equivalent to the following orthogonality
conditions holding for all $t \in [T]$: For all $s \in \br{1,\ldots, T}$, \begin{align*}
    0 &= \frac{1}{n}\sum_{i=1}^n \sum_{t \in \cT_i} \sum_{\bw \in \cW} \pi_i^*(\bw)
    \ddot{\bw}_{it}  \dot{\one}_{its} \\
    &= \frac{1}{n}\sum_{i: s \in \cT_i} \sum_{\bw \in \cW} \pi_i^*(\bw) \sum_{t \in \cT_i}
    \pr{\bw_t - R_i(\bw)/L_i - \sum_{\ell \in \cT_i} \delta_\ell \dot{\one}_{it\ell}}
    \dot{\one}_{its} \\
    &= \frac{1}{n}\sum_{i: s \in \cT_i} \sum_{\bw \in \cW} \pi_i^*(\bw) \sum_{t \in \cT_i}
    \pr{\bw_t - R_i(\bw)/L_i} \dot{\one}_{its} - \frac{1}{n}\sum_{i: s \in \cT_i} \sum_{t \in \cT_i}\sum_{\ell \in \cT_i} \delta_\ell \dot{\one}_{it\ell}
    \dot{\one}_{its} \\
    &= \frac{1}{n}\sum_{i: s \in \cT_i} \sum_{\bw \in \cW} \pi_i^*(\bw) \sum_{t \in \cT_i}
    \pr{\bw_t - R_i(\bw)/L_i}\one(t=s) - \frac{1}{n}\sum_{i: s \in \cT_i} \sum_{t \in
    \cT_i}\sum_{\ell \in \cT_i} \delta_\ell \one(t=\ell) \dot{\one}_{its} \tag{$\dot{\bw}$ and
    $\dot \one_i$ sum
    to zero over $t\in \cT_i$} \\
    &= \frac{1}{n}\sum_{i: s \in \cT_i} \sum_{\bw \in \cW} \pi_i^*(\bw) \pr{\bw_s -
    R_i(\bw)/L_i}
    - \frac{1}{n}\sum_{i: s \in \cT_i} \sum_{\ell \in \cT_i} \delta_\ell \dot{\one}_{i\ell s}
    \\
    &= \frac{1}{n}\sum_{i: s \in \cT_i} \sum_{\bw \in \cW} \pi_i^*(\bw) \pr{\bw_s -
    R_i(\bw)/L_i}
    - \frac{1}{n}\sum_{i: s \in \cT_i} \sum_{\ell \in \cT_i} \delta_\ell (\one(\ell = s) - 1/L_i)
    \end{align*}
    Rearranging, we have that the orthogonality condition is equivalent to \[
        \frac{1}{n}\sum_{i: s \in \cT_i} \sum_{\bw \in \cW} \pi_i^*(\bw) \pr{\bw_s -
        Q_i(\bw)} = \frac{1}{n} \sum_{i: s\in \cT_i} \pr{\delta_s - \frac{1}{L_i} \sum_{\ell
        \in \cT_i} \delta_\ell}
        \numberthis
        \label{eq:orthogonality_condition}
    \]
    for all $s \in [T]$.

    Pick $t \in \cT$ and consider the corresponding \eqref{eq:orthogonality_condition}.
    Note that $t \in \cT_i$ for all $i$, and hence  \[
    \delta_t = \frac{1}{n} \sum_{i=1}^n \frac{1}{L_i}\sum_{\ell \in \cT_i} \delta_\ell +
    \sum_{\bw \in \cW} \underbrace{\pr{\frac{1}{n} \sum_{i=1}^n \pi_i^*(\bw)}}_{\bar \pi(\bw)} \bw_t - \frac{1}{n}
    \sum_{i=1}^n \sum_{\bw \in \cW} \pi_i^*(\bw) Q_i(\bw).
    \]
    Given \eqref{eq:unbalanced_level_condition_it}, we can plug in \[
        \delta_t = \sum_{\bw \in \cW} \pi_i(\bw)\pr{\bw_t - Q_i(\bw)} + \frac{1}{L_i}
        \sum_{s \in \cT_i } \delta_s = \sum_{\bw \in \cW} \pi(\bw )\pr{\bw_t - Q_i(\bw)}+ \frac{1}{L_i}
        \sum_{s \in \cT_i } \delta_s.
    \]
Rearranging, we obtain that \[
        \sum_{\bw \in \cW, \bw \neq 0} (\pi(\bw) - \bar \pi(\bw)) \bw_t = M'(i)
    \]
    for some $M'(i)$ that does not depend $t \in \cT$. Again, the left-hand side, over $t
    \in \cT$, is in the
    span of $\cW_\cT$, whereas the right-hand side is in the span on $1_{\cT}$. As a
    result, both sides are equal to zero by assumption. Since vectors in $\cW_\cT$ are
    linearly independent, we
    have that \[
    \pi(\bw) = \bar\pi(\bw) \text{ for all $\bw \neq 0$} \implies \pi(\bw) = \bar\pi(\bw)
    \text{ for all $\bw \in \cW$}.
    \]

    Now, we plug $\pi_i(\bw) = \bar\pi(\bw)$ into
    \eqref{eq:unbalanced_level_condition_it}, and then plug the expression into
    \eqref{eq:orthogonality_condition}. We obtain that for all $t \in [T]$,\[
    \sum_{i : t \in \cT_i} \sum_{\bw \in \cW} \pi_i^*(\bw) (\bw_t - Q_i(\bw))
    = \sum_{i: t \in \cT_i} \sum_{\bw \in \cW} \bar\pi(\bw) (\bw_t - Q_i(\bw)).
    \]
    This is \eqref{eq:unbalanced_condition}. Thus, we proved that if an implicit design exists,
    then the above condition holds and it must be uniquely equal to $\bar\pi(\cdot)$.

    For the reverse direction, let us assume that \eqref{eq:unbalanced_condition} holds.
    For a given $t$, set \[
    \delta_t = \sum_{\bw \in \cW} \bar\pi(\bw) \bw_t,
    \numberthis
    \label{eq:delta_t}
    \]
    and thus $
    \delta_t - \frac{1}{L_i} \sum_{\ell \in \cT_i} \delta_\ell = \sum_{\bw \in \cW}
    \bar\pi(\bw) \pr{\bw_t - Q_i(\bw)}.
    $
    This means that $\pi_i(\bw) = \bar\pi(\bw)$ satisfies
    \eqref{eq:unbalanced_level_condition_it}. Thus $\bar\pi(\cdot)$ is an implicit design
    upon checking that \eqref{eq:delta_t} defines the projection coefficients.

    Summing over $i$ where $t \in \cT_i$:
    \[
    \frac{1}{n} \sum_{i: s\in \cT_i} \br{
        \delta_s - \frac{1}{L_i} \sum_{\ell \in \cT_i} \delta_\ell
    } = \frac{1}{n} \sum_{i: s\in \cT_i}  \sum_{\bw \in \cW}
    \bar\pi(\bw) \pr{\bw_t - Q_i(\bw)} = \frac{1}{n}  \sum_{i: s\in \cT_i}  \sum_{\bw \in \cW} \pi_i^*(\bw) (\bw_t - Q_i(\bw)),
    \]
where the last equality follows from \eqref{eq:unbalanced_condition}. This verifies the
orthogonality condition \eqref{eq:orthogonality_condition} for $t$. Therefore, our choice
of \eqref{eq:delta_t} does indeed equal the projection coefficients. This concludes the
proof. \qedhere

\section{Details of the replication exercise}
\label{asec:emp}

\citet{xu2017costs} studies the effect of social connections to one's political patron on
 the career outcomes of colonial governors (salary). The treatment is whether a serving
 governor is connected to the current Secretary of State for the Colonies, defined as
 sharing certain social ties. They argue connections are quasi-experimentally assigned
 within a governor's career: ``the turnover of secretaries of state induced by the
 electoral cycle in London generated shocks in social connections among serving
 governors. These within-governor shocks enable me to examine how changes in connections
 affected the allocation and performance of the same governor"; accordingly, ``The
 identification is therefore driven by governors who change their connections during
 their career.''

\citet{blakeslee2020way} study the effect of losing access to groundwater on farm
 household income and labor allocation in rural Karnataka, India. The treatment is an
 indicator that the household's first-drilled borewell has failed, the
 household's initial loss of irrigation access. They argue failure timing is
 quasi-experimental given the region's fractured hard-rock hydrogeology: ``The
 identifying
 assumption is that conditional on the year of drilling, the timing of the first
 borewell's failure is determined by exogenous geological attributes.''

\citet{dippel2021leadership} study the effect of the local presence of a
 Forty-Eighter---an exiled leader of the failed 1848–49 German revolutions---on an
 American town's Union Army enlistments in the Civil War, instrumenting Forty-Eighter
 presence with the destinations named by the exiles' shipmates. The instrument is an
 indicator that ordinary co-passengers aboard the ships that carried Forty-Eighters
 across the Atlantic declared the town as their intended destination. They argue the
 instrument is quasi-experimental because ``the Forty-Eighters did not plan their trip
 ahead in the same way that other immigrants did. Due to their haphazard departure, they
 rarely had a specific destination in mind when they departed,'' and because ``the
 Forty-Eighters selected their ships quasi-randomly.'' The paper includes a design in a
 permutation test: treating which 136 of the 1,746 immigrant voyages carried
 Forty-Eighters as a uniform random draw.

\citet{herrnstadt2021air} study the effect of exposure to traffic pollution on violent
 crime, comparing the two neighborhoods flanking the same Chicago interstate on the same
 day. The treatment is an indicator that a highway-side neighborhood is downwind of the
 interstate that day. They argue the wind
 direction makes exposure as-good-as-random across the road: ``In essence, this approach
 uses the upwind side of the interstate on a given day as a control for the 'treated'
 downwind side on the same day"; "For our estimate to be biased, an omitted variable must
 differentially affect crime on the side of the road to which the pollution is blowing.''

\citet{ang2023birth} studies the effect of screenings of \emph{The Birth of a Nation} on
 the local revival of the KKK, instrumenting the film's touring ``road show'' screenings
 with the presence of a movie theater before the film's release. The instrument is an
 indicator that the county had at least one movie theater by 1914. He argues the
 pre-existing theater network is quasi-experimental with respect to later racial animus.

\citet{cage2023heroes} study the effect of a municipality's World War I soldiers having
 served under Marshal Pétain on that municipality's Nazi collaboration a generation
 later. The treatment is an indicator that the municipality's home regiment rotated
 through the Battle of Verdun during Pétain's command window (rather than under his
 successors). They argue the rotation was quasi-experimental: ``Both in its conception
 and, as we show, in its implementation, the rotation to Verdun was based upon the needs
 of the moment and unrelated to the home characteristics of the regiments involved."

\citet{covert2023relinquishing} study the effect of allocating a mineral lease by
 centralized auction rather than decentralized negotiation on the terms the state
 obtains. The treatment is an indicator that a Texas Permanent School Fund lease was
 awarded by first-price sealed-bid auction rather than negotiated by a private surface
 owner. They argue the mechanism is quasi-experimental within a narrow geography: ``a
 lease's mechanism type is determined not by randomization, but by the date on which the
 parcel underlying it was first sold by the state. Our first identifying assumption is
 thus that variation in parcel privatization dates, within a narrowly defined geographic
 area, is uncorrelated with the unobservable geological quality within that area." The
 authors themselves estimate a treatment propensity with a cross-fitted random forest
 (on lease location, date, and size) inside a double-ML analysis, which we use as their
 reported candidate design.

\citet{cervellati2024random} study the effect of a random boost of votes to a 
party---arising from a ballot ordering---on the party's share of votes within its
 coalition and downstream outcomes. The treatment is an indicator
 that the party's symbol landed in the center of a ballot slot. 
 Here the assignment is explicitly randomized: The ordering of parties within a coalition
 on a ballot can be exactly permuted. 

\citet{chambru2024dynamic} study the effect of proximity to state administration on
 long-run local development: for the small towns surrounding the candidate cities in
 France's 1790 reorganization, the effect of one's nearest candidate city becoming the
 departmental capital. The treatment is an indicator that a peripheral town's closest candidate city
 is the one that won the captial designation. They argue the winner was
 quasi-experimental among candidates: "In some departments, the choice of local capitals
 over rival candidate cities was plausibly exogenous." The paper states a design where
 the designated capital is treated as a uniform draw among the department's candidate
 cities.

  For each paper, we first replicate the corresponding reported regression
  result. We then
     compute the four alternative designs $\tilde \bpi$ described in the main text.
     Additional details are below:
  \begin{enumerate}[label=(\roman*)]
      \item (Recalibrated) We subclassify on $\hat\bpi$ using the data-driven procedure of
      \citet{imbens2015causal} (chapter 17) by recursively partitioning $\hat\pi_i$.
      Starting from
      the block of all
      units with $\hat\pi_i \in (0,1)$, we linearize the estimated propensities,
      $\ell_i = \log(\hat\pi_i / (1-\hat\pi_i))$, and compute the two-sample $t$-statistic
      comparing $\ell_i$ between treated and untreated units within the block. If
      $|t| > 1$, the block contains more than 11 units, and a split-on-median is
      feasible---the
      block contains both treated and untreated units and its median $\hat\pi_i$ separates
      it into two non-empty halves---we split the block at its median $\hat\pi_i$ and
      recurse on the two halves; otherwise the division terminates. Units with
      $\hat\pi_i < 0$ and $\hat\pi_i > 1$ form two additional dedicated blocks. This
      procedure determines knots for the partition, from which intervals of the form $
      (a, b]$ are constructed.
       Each
      unit's $\tilde\pi_i$ is then the empirical treatment frequency within its interval,
      clipped to $[10^{-4}, 1-10^{-4}]$.

      \item (Logit) We fit a fixed-effects logistic regression of $W$ on the same
      covariates and fixed effects as the headline specification and take
      $\tilde\pi_i$ to be the fitted probabilities, clipped to $[10^{-4}, 1-10^{-4}]$.
      The logit regression failed to converge for \citet{xu2017costs}---the
      within-governor variation quasi-separates the outcome---and thus we omit it from
      the paper.

      \item (ML) For applications with more than 1{,}000 observations, we predict $W$
      with gradient-boosted trees (LightGBM). The feature set consists of the
      specification's non-fixed-effect covariates; one-hot encodings of any
      low-dimensional fixed effects (at most 50 levels); and the implicit design
      $\hat\pi_i$ itself as one additional feature. We
      exclude covariates that mechanically determine treatment (in
      \citet{chambru2024dynamic}, raw latitude and longitude, which pin down the
      nearest-candidate assignment).\footnote{The booster minimizes binary log-loss with
      learning
      rate $0.05$, $15$ leaves, a minimum of $20$ observations per leaf, row subsampling
      of $0.8$ per iteration, column subsampling of $0.8$, an $\ell_2$ penalty of $1.0$,
      and $300$ boosting rounds, run single-threaded with a fixed seed.} Predictions are
      generated by five-fold cross-fitting---each unit's raw score comes from a booster
      trained on the other four folds, with the folds drawn once by a seeded random
      partition---so no unit contributes to its own prediction. The out-of-fold scores
      are then recalibrated by Platt scaling: a logistic regression of $W_i$ on the
      log-odds of the raw score, whose fitted probabilities, clipped to
      $[10^{-4}, 1-10^{-4}]$, constitute $\tilde\bpi$.

      \item (Paper) For \citet{covert2023relinquishing}, the paper reports an AIPW
       estimate where the propensity is estimated by random forest. We use exactly the
       same random forest setup: a \texttt{grf} regression forest ($1{,}000$ trees) of
       the auction indicator on the four continuous covariates the authors use, cross-fit
       over five randomly assigned folds so that predictions are out-of-fold, with the
       five-fold split repeated $101$ times and $\tilde\pi_i$ taken as the per-lease
       median across repetitions, clipped to $[10^{-4}, 1-10^{-4}]$.

  \end{enumerate}

\end{proof}

\end{appendices}
\end{document}